\documentclass[10pt,journal,compsoc]{IEEEtran}

\usepackage{comment}
\usepackage{ifpdf}
\usepackage{ragged2e}

\usepackage{tikz}
\usepackage{bm}
\usepackage[colorlinks, linkcolor=hotpink, anchorcolor=blue, citecolor=hotpink, ]{hyperref}
\usepackage{enumitem}

\ifCLASSOPTIONcompsoc
  \usepackage[nocompress]{cite}
\else
  \usepackage{cite}
\fi
\ifCLASSINFOpdf
\else
\fi

\usepackage{amsmath,amssymb,amsfonts}
\usepackage{amsthm}
\usepackage{algorithmic}
\usepackage{graphicx}
\usepackage{textcomp}
\usepackage{xcolor}
\usepackage{multirow}
\usepackage{booktabs} 
\usepackage{threeparttable}
\usepackage{array}
\usepackage{mdwmath}
\usepackage{mdwtab}
\usepackage{hyperref}
\usepackage{mathrsfs}
\usepackage{threeparttable}
\usepackage{stfloats}
\usepackage[switch]{lineno}
\hypersetup{
    colorlinks = true,
    linkcolor = black,
    anchorcolor = black,
    citecolor = black,
    filecolor = black,
    urlcolor = black,
    breaklinks = true
    }
\usepackage{diagbox} 

\newtheorem{Corollary}{Corollary}

\newtheorem{Proposition}{Proposition}
\newtheorem{Definition}{Definition}
\newtheorem{Theorem}{Theorem}
\newtheorem{Lemma}{Lemma}
\usepackage[numbers]{natbib}
\usepackage[ruled,boxed,commentsnumbered]{algorithm2e}

\usepackage{makecell}
\usepackage{url}

\usepackage[caption=false,font=normalsize,labelfont=rm,textfont=rm]{subfig}

\usepackage{comment}

\ifCLASSOPTIONcompsoc
  \usepackage[nocompress]{cite}
\else
  \usepackage{cite}
\fi

\ifCLASSINFOpdf
\else
\fi

\begin{document}

\title{
Learning Deep Tree-based Retriever for Efficient Recommendation: Theory and Method
}



\author{Ze Liu, Jin Zhang, Defu Lian\textsuperscript{*}, Chao Feng,  Jie Wang, and Enhong Chen,~\IEEEmembership{Fellow,~IEEE}



\thanks{
Ze Liu and Chao Feng are with the School of Computer Science and Technology, University of Science and Technology of China, Hefei, Anhui 230027, China (e-mail: lz123@mail.ustc.edu.cn, chaofeng@mail.ustc.edu.cn)
}
\thanks{Jin Zhang is with the School of Artificial Intelligence and Data Science, University of Science and Technology of China, Hefei, Anhui 230027, China (e-mail: jinzhang21@mail.ustc.edu.cn).}%
\thanks{Defu Lian and Enhong Chen are with the State Key Laboratory of Cognitive Intelligence, School of Computer Science and Technology, School of Artificial Intelligence and Data Science, University of Science and Technology of China, Hefei, Anhui 230027, China (e-mail: liandefu@ustc.edu.cn; cheneh@ustc.edu.cn).}%
\thanks{Jie Wang is with the Department of Electronic Engineering and Information Science, University of Science and Technology of China, Hefei, Anhui 230027, China (e-mail: jiewangx@ustc.edu.cn).}%
\thanks{\textsuperscript{*}Corresponding author: Defu Lian (e-mail: liandefu@ustc.edu.cn).}
\thanks{This work has been submitted to the IEEE for possible publication. Copyright may be transferred without notice, after which this version may no longer be accessible.}
}

\markboth{Journal of \LaTeX\ Class Files,~Vol.~14, No.~8, August~2021}%
{Shell \MakeLowercase{\textit{et al.}}: A Sample Article Using IEEEtran.cls for IEEE Journals}


\IEEEtitleabstractindextext{
    \begin{abstract}
    With the advancement of deep learning, deep recommendation models have achieved remarkable improvements in recommendation accuracy.  
    However, due to the large number of candidate items in practice and the high cost of preference computation, these methods still suffer from low recommendation efficiency. The recently proposed tree-based deep recommendation models alleviate the problem by directly learning tree structure and representations under the guidance of recommendation objectives. 
    To guarantee the effectiveness of beam search for recommendation accuracy, these models strive to ensure that the tree adheres to the max-heap assumption, where a parent node's preference should be the maximum among its children's preferences.
    However, they employ \textcolor{black}{a one-versus-all strategy}, framing the training task \textcolor{black}{as a series of independent binary classification objectives for each node}, which limits their ability to fully satisfy the max-heap assumption.
    To this end, we propose a Deep Tree-based Retriever (DTR for short) for efficient recommendation.
    DTR frames the training task as a softmax-based multi-class classification over tree nodes at the same level, \textcolor{black}{enabling explicit horizontal competition and more discriminative top-k selection among them, which mimics the beam search behavior 
    during training}.
    To mitigate the suboptimality induced by the labeling of non-leaf nodes, we propose a rectification method for the loss function, which further aligns with the max-heap assumption in expectation.  
    As the number of tree nodes grows exponentially with the levels, we employ sampled softmax to approximate optimization and thereby enhance efficiency. 
    Furthermore, we propose a tree-based sampling method to reduce the bias inherent in sampled softmax. 
    Theoretical results reveal DTR's generalization capability, and both the rectification method and tree-based sampling contribute to improved generalization.
    The experiments are conducted on four real-world datasets, validating the effectiveness of the proposed method.
    \end{abstract}
    
}

\maketitle

\IEEEdisplaynontitleabstractindextext

\IEEEpeerreviewmaketitle

\IEEEraisesectionheading{\section{Introduction}}
\IEEEPARstart{I}{n} the information era, the overwhelming volume of daily information leads to significant information overload. The recommendation plays a crucial role in mitigating information overload by providing a personalized ranking list tailored to individual preferences. With the advancements of deep learning techniques, recommendation techniques also achieve remarkable improvements in ranking performance. The widespread application of these techniques has generated considerable economic benefits for various kinds of content providers in industrial companies.

Through the use of deep learning, we not only learn better representations for users, items, and contexts but also provide a more generalized expression for users' preference scores via neural networks than the widely-used inner product in matrix factorization. Both lead to stronger recommendation performance, as demonstrated by models such as DIN~\cite{zhou2018deep}, DIEN~\cite{zhou2019deep},
and CKE~\cite{zhang2016collaborative}. However, the use of neural network-based preference functions brings online-serving challenges for recommender systems due to the high cost of preference computation. Generally speaking, immediate responses to adaptive recommendations are prerequisites for excellent customer experiences and custom retention. 
Existing popular and well-performed graph-based indexes (e.g., HNSW~\cite{malkov2018efficient}), quantization-based indexes (e.g.,
SCANN~\cite{guo2020accelerating}), and hash-based indexes (e.g., SGDH~\cite{li2020weakly}) are typically constructed based on inner product or Euclidean distance, making them unsuitable for accelerating the recommendations of deep models.
In particular, neural network-based preference functions do not form a valid metric and are incompatible with Euclidean distance and inner product. As a result, two items that are close in Euclidean or inner-product space may have significantly different neural preference scores.

To guarantee the compatibility between search indexes with the neural network-based preference functions, it is a good solution to learn the search indexes together with the recommendation model under the guidance of recommendation objectives. The representative work is the tree-based deep model (i.e., TDM~\cite{zhu2018tdm} and its improved version JTM~\cite{Zhu2019JointOO}). These models use the balanced tree index, which is constructed by hierarchically clustering item representations from top to bottom. 
Given the tree, the top-k ranked items are retrieved through a layer-wise beam search, which selects the k largest tree nodes at each level based on neural preference scores from the top layer to the bottom layer. 
In this way, beam search achieves logarithmic computation complexity w.r.t. the number of items. To guarantee the accuracy of beam search, these tree-based deep models rely on the following max-heap assumption: the preference scores of query for a parent node should be the maximum between the preference scores of its children node. For the sake of satisfying the assumption, they cast the overall problem as a series of independent node-wise binary classification problems, treating nodes in the path from the root node to the leaf nodes corresponding to positive samples as positive and other randomly sampled nodes as negative. However, these tree-based models suffer from the following drawback: the max-heap assumption is not well satisfied by the used binary classification objectives due to the lack of explicit horizontal competition among the tree nodes at the same level. 
This drawback provides an opportunity to improve the accuracy of the efficient recommendation.

We propose a Deep Tree-based Retriever (DTR for short) for efficient recommendation. To satisfy the max-heap assumption as much as possible and mimic the beam search behavior in the training stage, DTR regards the training task as \textcolor{black}{softmax-based multi-class classification} over tree nodes at the same layer, which enables \textcolor{black}{explicit horizontal competition and more discriminative top-k selection among them}. Within this training mode, DTR utilizes a multi-class cross-entropy loss to optimize the deep model; however, theoretical analysis from the aspect of Bayes optimality indicates that this loss function still results in suboptimal performance under beam search induced by the labeling of non-leaf nodes, prompting us to propose a rectification method. Additionally, as the number of tree nodes increases exponentially with levels, the softmax loss becomes computationally expensive. Therefore, we resort to sampled softmax for approximation to promote training efficiency. We develop a tree-based negative sampling method, guided by sampled softmax theory, to estimate the gradient of the original softmax loss more accurately. Moreover, we propose a tree learning method compatible with this training mode, enabling the alternating learning of the preference model and the tree index. Furthermore, we provide a generalization analysis of our proposed method, showing that it has good generalization capabilities and that the generalization could be enhanced through \textcolor{black}{a negative sampling distribution closer to the softmax distribution} and an increased number of branches in the tree structure.

As an extension of the preliminary paper \cite{feng2022forest}, which contributed to the layer-wise softmax-based multi-class classification training mode and a tree learning method compatible with this training mode, we further make the following contributions:
\begin{itemize}
    \item We identify the issue of suboptimality of multi-class classification training mode under beam search from the aspect of Bayes optimality and propose a rectification method to address this issue.
    \item We propose a tree-based negative sampling method for the multi-class classification problem in the tree. This sampling method can lead to more accurate estimations of loss gradient to reduce the inherent bias in sampled softmax.
    \item We provide a generalization analysis of the deep tree-based retriever, demonstrating the great generalization capability of our proposed methods.
    \item We evaluate the proposed methods on four real-world datasets and validate the superiority of DTR to the baselines and the effectiveness of the rectification method and the sampling method.
\end{itemize}

\section{Related work}
\label{sec: related work}
This study aims to enhance the accuracy of efficient recommendation through the deep tree-based retriever, focusing on both the methodologies and theoretical foundations. We begin by reviewing recent advancements in efficient recommendation. Subsequently, we survey recent efficient training techniques of recommendation and then delve into the theories closely connected to our study.

\subsection{Efficient Recommendation}
Efficient recommendation relies on building a search index, including LSH~\cite{datar2004locality}, inverted index~\cite{norouzi2012fast,babenko2014inverted}, tree index~\cite{preparata2012computational} and graph index~\cite{malkov2018efficient} for all items. The recommender system usually uses the inner product for computing preference scores, and the top-k recommendation can be cast into the maximum inner product search (MIPS) problem. The search index is usually constructed based on Euclidean distance and has been extended to the inner product. This extension can be achieved by establishing the relationship between nearest neighbor search (NNS) and MIPS~\cite{bachrach2014speeding,neyshabur2015symmetric,douze2024faiss}, or learning from either item representations~\cite{guo2016quantization,morozov2018non} or the raw data directly~\cite{lian2017discrete,mazur2019beyond,lian2021product}. With the introduction of deep learning, the preference score function becomes complicated, so that it is challenging to transform neural ranking into MIPS. Existing work either directly used metric-based index~\cite{tan2020fast}, or learns search index from raw data directly together with recommendation models~\cite{zhu2018tdm, Zhu2019JointOO,lian2020lightrec}.

\subsection{Efficient Training Techniques of Recommendation} 
\subsubsection{Negative Sampling in RecSys}
Negative sampling is a critical technique for mitigating missing negatives and exposure bias, while also improving training efficiency in recommender systems~\cite{lian2020personalized,subbiah2024improved,prakash2024evaluating,yang2024does}. It includes static methods like uniform sampling~\cite{rendle2014improving}, popularity-based methods, and context-aware adaptive sampling. Representative adaptive methods include adaptive oversampling~\cite{rendle2014improving}, rejection sampling~\cite{hsieh2017collaborative}, clustering-based sampling~\cite{lian2020personalized}, and dynamic negative sampling~\cite{zhang2013optimizing}. Their central idea is to assign higher sampling probabilities to items with larger predicted preference scores.

\subsubsection{Techniques of Speeding Up Softmax Computation}
Computing the softmax over a large vocabulary is often costly. To improve efficiency, various approximation methods have been proposed. For example, hierarchical softmax~\cite{morin2005hierarchical} and LightRNN~\cite{li2016lightrnn} decompose the output space, while contrastive divergence~\cite{hinton2002training} approximates gradients via MCMC. Negative sampling is another widely used approach, including noise-contrastive estimation~\cite{gutmann2010noise} with unigram sampling, self-contrastive estimation~\cite{goodfellow2014distinguishability} using the previous model to generate negative samples, self-adversarial sampling~\cite{sun2019rotate}, kernel-based sampling~\cite{blanc2018adaptive} with a tree index, and adaptive top-k softmax~\cite{baharav2024adaptive}.

\subsection{Theoretical Work}
\subsubsection{Bayes Optimality in Multi-class Classification and Hierarchical Classification}
Bayes optimality characterizes the ideal classifier that minimizes expected loss w.r.t. the true data distribution. It has been extensively studied~\cite{zhang2004statistical, tewari2007consistency} in multi-class classification. To handle label ambiguity, top-k Bayes optimality focuses on the $k$ most probable classes~\cite{lapin2017analysis, yang2020consistency}. %
Recent work extends Bayes optimality to hierarchical classification: \cite{wydmuch2018no} utilizes it to assess probability estimates in tree models, while \cite{zhuo2020learning} analyzes performance degradation under beam search and proposes a method to achieve Bayes optimality. 

\subsubsection{Generalization Bounds for Multi-class Classification} 
Generalization error bounds provide theoretical guarantees for the generalization ability of a learned classifier. \cite{koltchinskii2002empirical} provides data-dependent generalization bounds for multi-class classification based on margins. \cite{cortes2013multi} further develop these bounds with multiple kernels. To accommodate the scenario with a large number of classes, data-dependent generalization error bounds with a mild dependency on the number of classes have also been studied \cite{lei2019data}. The aforementioned bounds primarily focus on flat structures. In the context of hierarchical structures, \cite{babbar2016learning} proposes a multi-class, hierarchical data-dependent generalization error bound for kernel classifiers in large-scale taxonomies.

\section{Preliminaries}
\label{sec: preliminaries}

\subsection{Problem Definition and Notation}
Consider a user space $\mathcal{U}$ and an item set $\mathcal{Y}$, where each user in $\mathcal{U}$ is characterized by its historical interaction sequence and other available profiles. Assume that the data is generated i.i.d. from some distribution $\mathbb{P}$ over $\mathcal{U} \times \mathcal{Y}$, and an instance is represented as $(u, y) \in \mathcal{U} \times \mathcal{Y}$, with the corresponding joint probability $p(u, y)$. For a given user $u$, the conditional probability of observing item $y$ is denoted by $\eta_y(u) = p(y|u)$.  The objective of the recommendation is to return the top-$k$ items with the highest conditional probabilities $\eta_y(u)$ for the given user $u$.

In this paper, we use the following notational conventions: bold lowercase and uppercase letters for vectors and matrices respectively, such as $\boldsymbol{a}$ and $\boldsymbol{A}$, and non-bold letters for scalars or constants, such as $k$ and $C$. For vector $\boldsymbol{a}$, $a_i$ denotes its $i$-th component and $\|\boldsymbol{a}\|_p$ denotes the $\ell_p$ norm. For matrix,
$\|\boldsymbol{A}\|_p=\sup\limits_{\|\boldsymbol{x}\|_p=1}\|\boldsymbol{A}\boldsymbol{x}\|_p$
denotes matrix norms induced by vector $p$-norms.  We denote the set $\{1,2, \dots, m\}$ by $\left[m\right]$ for any natural number $m$. Other notations will be explicitly defined in the corresponding section.
Additionally, we summarize some important and commonly used notations of this paper in Appendix \ref{sec: notations}.

\subsection{Tree-based Model}
Modern recommender systems need to retrieve top-$k$ items from a large-scale corpus (i.e., $\left|\mathcal{Y}\right|$ ranges from millions to billions or even tens of billions). 
To tackle the top-$k$ retrieval for such a large corpus efficiently and effectively, TDM ~\cite{zhu2018tdm} and JTM~\cite{Zhu2019JointOO} propose the tree-based recommendation models, which consist of a max-heap-like tree index and an advanced preference model.
In this subsection, we detail the components of the tree-based model, including the tree index and the preference model, followed by the introduction of its top-$k$ retrieval and training processes.

\subsubsection{Tree Index}
The tree index is a key component of tree-based models, enabling logarithmic-time retrieval over a large item corpus. It is constructed as a $B$-ary tree $\mathcal{T}$, consisting of a set of nodes $\mathcal{N}$. Nodes at the $j$-th level are denoted by $\mathcal{N}^j$,  and the entire node set is $\mathcal{N} = \bigcup_{j = 0}^H \mathcal{N}^j$, where the tree height $H = \lceil \log_B |\mathcal{Y}| \rceil$ to accommodate all items. The bijective mapping $\pi: \mathcal{N}^H \rightarrow \mathcal{Y}$ maps each leaf node to a unique item, with the inverse mapping denoted by $\pi^{-1}: \mathcal{Y} \rightarrow \mathcal{N}^H$. Accordingly, an item $y$ is mapped to a leaf and corresponds to a path in the tree from the root to that leaf. 

For each level $j$, we denote the number of nodes in that level as $N_j = \left|\mathcal{N}^j\right|$ and assume the nodes are listed in a fixed order, denoted as $(n_1^j, n_2^j, \dots, n_{N_j}^j)$. For the $i$-th node in the $j$-th level, denoted as $n_i^j$, let $\rho^l(n_i^j) \in \mathcal{N}^l$ represent its ancestor in level $l$ ($0 \le l \le j$), where $\rho^j(n_i^j)$ refers to the node itself. The set $\mathcal{C}(n_i^j) \subseteq \mathcal{N}$ represents child nodes of $n_i^j$, while $\mathfrak{L}(n_i^j) \subseteq \mathcal{N}^H$ denotes the leaf nodes within the subtree rooted at $n_i^j$. The function $\delta$ is introduced to return the index of a node within its level, i.e., $\delta(n_i^j) = i$. The $j$-th ancestor in item $y$'s corresponding path is $\rho^j(\pi^{-1}(y))$ and its index in this level is denoted as $\delta^j(y)=\delta(\rho^j(\pi^{-1}(y)))$. 

\subsubsection{Preference Model}
\label{sec: formalization of DIN}
With the tree index, tree-based models employ a preference model $\mathcal{M}: \mathcal{U} \times \mathcal{N} \rightarrow \mathbb{R}$ to predict user preference scores over tree nodes. 
These scores guide a top-down search to select the top-$k$ nodes at each level,
ultimately retrieving the top-$k$ items. 
DIN~\cite{zhou2018deep}, or one of its variants, is commonly adopted as the preference model.  We provide a detailed description of DIN in this subsection.

In the DIN model, the interaction sequence of a user $u$ serves as the input. Specifically, the model takes as input a matrix of $K$ item embedding vectors, denoted as $\boldsymbol{A}^{(u)}$:
\begin{small}
\begin{equation}
\label{eq:embedding_of_item_sequence}
    \boldsymbol{A}^{(u)} = \left[\boldsymbol{a}_1^{(u)}, \boldsymbol{a}_2^{(u)}, \dots, \boldsymbol{a}_{K}^{(u)}\right] \in \mathbb{R}^{d\times K},
\end{equation}
\end{small}where $\boldsymbol{a}_k^{(u)}$ represents the embedding of $k$-th item. The model employs a two-layer fully connected network to compute the weight between the node $n$ and the $k$-th interacted item, which can be expressed as:
\begin{small}
\begin{equation}
\label{eq:weights_of_target}
w_k^{(u)}(n)=\phi\left(\boldsymbol{W}_w^{(2)}\phi\left(\boldsymbol{W}_w^{(1)}\left[\boldsymbol{a}_k^{(u)};\boldsymbol{a}_k^{(u)}\odot \boldsymbol{w}_n;\boldsymbol{w}_n\right]\right) \right)\in \mathbb{R},
\end{equation}    
\end{small}where $\boldsymbol{W}_w^{(1)}\in \mathbb{R}^{h\times 3d}, \boldsymbol{W}_w^{(2)}\in \mathbb{R}^{1\times h}$, $\boldsymbol{w}_n$ is the embedding of the node $n$, and $\phi$ is the activation function.
For notational simplicity, we omit $n$ and denote the weight by $w_k^{(u)}$ in the following.

To capture dynamic user interests, the interaction sequence is partitioned into $K'$ time windows, where the $t$-th time window $T_t$ of length $k$ is a set of continuous indices, i.e., $T_t=\{i_0+1,\dots,i_0+k\}$. These time windows are mutually exclusive, and their union is exactly $\left[K'\right]$.  
In the $t$-th time window, the interacted item embeddings are aggregated using weights calculated in {Eq. (\ref{eq:weights_of_target})}: 
\begin{small}
\begin{equation}
    \boldsymbol{z}_t^{(u)}=\sum_{k\in T_t} w_k^{(u)} \boldsymbol{a}_k^{(u)}\in \mathbb{R}^d.
\end{equation}
\end{small}Then, the concatenation of $K'$ aggregated embeddings and the embedding of node $n$ are fed into an $L$-layer MLP to output the user's preference score for node:
\begin{small}
\begin{equation}
    f_{\text{DIN}}(u,n) 
    =\boldsymbol{W}_L \cdot \phi_{L-1} \circ \dots \circ \phi_1  \Big( \big[ \boldsymbol{z}_1^{(u)}; \boldsymbol{z}_2^{(u)}; \dots; \boldsymbol{z}_{K'}^{(u)}; \boldsymbol{w}_{n} \big]\Big)
\end{equation}
\end{small}where $\boldsymbol{W}_L\in \mathbb{R}^{1\times d_{L-1}}$, and the function $\phi_k(\boldsymbol{x})$ is defined:
\begin{small}
\begin{equation}
\phi_k(\boldsymbol{x})=\phi(\boldsymbol{W}_k \boldsymbol{x})\in \mathbb{R}^{d_k\times 1}, \forall k\in [L-1].
\end{equation}
\end{small}The element-wise activation function $\phi$ is Lipschitz continuous with Lipschitz constant $c_{\phi}$, has the property $\phi(0)=0$ and $\boldsymbol{W}_k \in \mathbb{R}^{d_k \times d_{k-1}}$ represents the weight matrix. 

The preference model $\mathcal{M}$ is exactly the function space of $f_{\text{DIN}}$. Given a user $u$, the preference score for node $n_i^j$ is denoted as $o_i^j(u) = f_{\text{DIN}}(u, n_i^j)$. For simplicity, we omit $u$ and adopt the notation $o_i^j$ for the remainder of this paper.

\subsubsection{Top-k Retrieval Process}
With the tree index and the preference model, tree-based models enable efficient and effective  top-$k$ item retrieval. Since each item is mapped to a leaf, the  retrieval task is transformed into retrieving the top-$k$ leaf nodes for a given user $u$. As the tree is maintained as a max-heap, this can be achieved by only searching the top-$k$ nodes at each level from top to bottom along the tree. Thus, the retrieval time complexity is logarithmic w.r.t. the number of items, as a constant number of nodes are searched per layer, and the number of layers scales logarithmically with the number of items.

The retrieval process is presented in {Algorithm \ref{alg:beam_search}}, where $o^{\mathcal{M}}(n|u)$ is the user $u$'s preference for node $n$ computed by the preference model $\mathcal{M}$. Essentially, this process is a layer-wise beam search with beam size $k$ along the tree $\mathcal{T}$, guided by the preference $o_i^j$,
and can be formulated as follows:
\begin{equation}
    \label{eq: beam search process}
    \mathcal{B}^{j}(u)\in\operatornamewithlimits{argTopk}\limits_{i\in \widetilde{\mathcal{B}}^{j}(u)}\ o_i^{j},
\end{equation}
where $\mathcal{B}^j(u)$ denotes the set of indices of selected nodes in the $j$-th level for the user $u$, and $\widetilde{\mathcal{B}}^{j}(u) = \left\{ \delta(n) \mid n \in \bigcup_{i \in \mathcal{B}^{j-1}(u)} \mathcal{C}(n^{j-1}_i) \right\}$ represents the index set of nodes to be expanded in the $j$-th level. It's important to note that Eq. (\ref{eq: beam search process}) employs the symbol “$\in$" instead of “$=$" because ties in probabilities may occur, leading to multiple potential sets of top-$k$ nodes. If $\left| \widetilde{\mathcal{B}}^{j}(u) \right| < k$, then all nodes in $j$-th level will be selected.
By recursively applying {Eq. (\ref{eq: beam search process})} until the $H$-th level, the top-$k$ leaf nodes $\mathcal{B}^H(u)$ are retrieved, and the final retrieved item set for the user $u$ is
\begin{equation}
    \widehat{\mathcal{Y}}(u) = \{\pi(n_i^H) \mid  i\in\mathcal{B}^H(u)\}.
\end{equation}

\begin{algorithm}[t]
   \caption{\textsc{Beam Search: Layer-wise Retrieval}}\label{alg:beam_search}
   \textbf{Input:} User $u$, tree $\mathcal{T}$, beam size $k$, the preference model $\mathcal{M}$\\
   \textbf{Output:} top-$k$ items\\
\begin{algorithmic}[1]
   \STATE Result set $A=\emptyset$, candidate set $Q=\{root\ node\ n_1^0\}$;\\
   \WHILE{Q is not empty}
   \STATE Remove all leaf nodes from $Q$ and add them into result set $A$ if Q contains leaf nodes;\\
   \STATE Compute preference $o^{\mathcal{M}}(n|u)$ through $\mathcal{M}$ for each node in set $Q$;\\
   \STATE Parents=\{top-$k$ nodes of Q according to $o^{\mathcal{M}}(n|u)$\};\\
   \STATE Q=\{children of node $n\ |\ n$ $\in$ Parents\};\\
   \ENDWHILE
   \STATE \textbf{return} The top-$k$ items w.r.t. the top-$k$ leaf nodes according to $o^{\mathcal{M}}(n|u)$, $n\in A$.
\end{algorithmic}
\end{algorithm}

\subsubsection{Training Process in TDM and JTM}
To facilitate understanding of the optimization of tree-based models, we now introduce the training process in TDM and JTM.
Both TDM and JTM learn the preference model and tree alternately, fixing one while updating the other until convergence. To learn the preference model, they regard the training task as a series of independent node-wise binary classification problems. Supposing user $u$ has an interaction with an item that corresponds to a leaf node, this leaf node and its ancestors are positive nodes with label 1, while all other tree nodes  are negative nodes with label 0. Then,  the training loss for user $u$ is 
\begin{small}
\begin{equation}
\begin{aligned}
\label{binary_corss_entropy}
   \mathcal{L}(u,S_u^+,S_u^-)=&-\sum_{n\in S_u^+}
    z_u(n)\log p(\hat{z}_u(n)=1|u)\\
    &-\sum_{n\in S_u^-}(1-z_u(n))\log p(\hat{z}_u(n)=0|u).
\end{aligned}
\end{equation}
\end{small}$S_u^+$ and $S_u^-$ denote the positive nodes and negative nodes for user $u$, respectively. $z_u(n)$ denotes the label of node $n$ for user $u$. $p(\hat{z}_u(n)=1|u)$ and $p(\hat{z}_u(n)=0|u)$ denote the like/dislike probabilities for user $u$ to node $n$, predicted by the preference model. Since using all negative nodes is prohibitive in time and memory, both TDM and JTM
sample a subset of negative nodes $S_u^{'-}$ uniformly from $S_u^{-}$, 
and the loss $\mathcal{L}(u,S_u^+,S_u^{'-})$ is used instead. 

To learn the tree, both TDM and JTM adopt hierarchical clustering-like strategies to reassign the items to leaf nodes.  TDM recursively applies k-means clustering on the updated item embeddings until each cluster contains only one item, thereby constructing a tree where each item is assigned to a leaf node. JTM assigns items layer by layer from the root to the leaves in a way that maximizes the total log-likelihood of user preferences for positive nodes at each layer.
In our work, we propose a tree learning method compatible with our proposed preference model learning mode.

\section{Learning Deep Tree-based Retriever}
The performance of the deep tree-based retriever depends on both the preference model and the tree index. 
To learn the deep tree-based retriever, we need to optimize the preference model (i.e., the deep model that outputs the scores between users and nodes) and update the tree index (i.e., the mapping $\pi$ between items and leaf nodes).  Since the tree index updating is discrete and non-differentiable, we learn the preference model and the tree index alternately. Concretely, we fix the mapping $\pi$ and apply stochastic gradient descent to optimize the preference model; then, we fix the preference model and use a discrete optimization method to update the mapping $\pi$. 
This alternating process of model optimization and tree index updating improves retrieval performance gradually until convergence.
 
In this section, we elaborate the learning process of the proposed DTR. We (i) first introduce a softmax-based multi-class training mode to close the gap between training and inference (Subsection~\ref{sec: section of multi-class classification}); (ii) mitigate suboptimality of the softmax loss under beam search via a label rectification method for internal nodes (Subsection~\ref{sec: rectification loss}); (iii) enable efficient training in practice with a tree-based sampling method (Subsection~\ref{sec: negative sampling}); and (iv) update the tree index via top-down item reassignment to align it with the optimized preference model (Subsection~\ref{sec: tree index updating}). An algorithmic overview is provided in Appendix~\ref{appendix: overview of learning process of DTR}.

\subsection{Softmax-based Multi-class Classification Training}
\label{sec: section of multi-class classification}
We develop a softmax-based multi-class classification training mode to address the shortcomings of the binary classification training mode in TDM and JTM.  We regard the training task as a multi-class classification problem at each layer, and employ the classic multi-class cross-entropy loss to optimize the preference model. Upon examining the Bayes optimality of this training scheme, we observe that the traditional multi-class cross-entropy loss still leads to suboptimality under beam search.

\subsubsection{Multi-class Cross-entropy Loss}
\label{subsubsec:multi-class cross-entropy loss}
TDM and JTM formulate the training task as a series of independent node-wise binary classification problems, employing binary cross-entropy as the training loss.  In this mode,  each node contributes to the loss independently, resulting in a lack of explicit horizontal competition among nodes within the same layer. 
However, the inference process involves a top-down, layer-wise retrieval that requires comparison among candidate nodes within each layer.
This creates a gap between training and prediction in TDM and JTM.
To bridge this gap, we reformulate the training at each layer as a multi-class classification problem. By employing a multi-class cross-entropy loss, we encourage direct competition among nodes within the same layer, which better aligns the training process with inference and ultimately improves retrieval accuracy.

In our proposed training mode, we treat each layer $j$ as an independent multi-class classification task with $N_j$ classes, where each node corresponds to one class.
Given an instance $(u,y)$, we attach a label $z_i^j(u,y)$ to each node $n_i^j$ to indicate whether it is positive or negative. For simplicity, we will omit $u,y$ and use the notation $z_i^j$ throughout the paper.
The positive node can be identified by backtracking from the item $y$'s mapping leaf node to the root node, as in TDM and JTM, and the corresponding label will be assigned as $1$, i.e., $z^j_{\delta^j(y)}=1$. The remaining $N_j-1$ nodes are the negative nodes with the label $0$. 
The preference model outputs a score for each node, resulting in $N_j$ preference scores for the $j$-th layer,  which is tailored for the multi-class classification task with $N_j$ classes.
For the given instance $(u,y)$, the index of the positive node in layer $j$ is $\delta^j(y)$. Then, the training loss at layer $j$, specifically the multi-class cross-entropy loss, is calculated as follows:
\begin{equation}
\label{eq:multi_cross_entropy_single_layer}
    \mathcal{L}_{j}(u,y)=-\log p^j_{\delta^j(y)}=-\log \frac{\exp o_{\delta^j(y)}^j}{\sum_{k=1}^{N_j} \exp o_k^j}.
\end{equation}
As shown in Eq. (\ref{eq:multi_cross_entropy_single_layer}), the loss for a given layer involves all nodes in that layer simultaneously, which mitigates the aforementioned training-inference discrepancy. The training loss w.r.t. the whole tree is the sum of the losses for each layer.

\subsubsection{Suboptimality of Multi-class Cross-entropy Loss Under Beam Search}
\label{sec: suboptimality of ce loss}
Treating the training task as a multi-class classification problem at each layer can alleviate shortcomings associated with binary classification. However, the multi-class cross-entropy loss still leads to suboptimality under beam search. 
In this subsection, we demonstrate how multi-class cross-entropy loss results in suboptimal outcomes under beam search.

Similar to Bayes optimality defined in \cite{yang2020consistency, zhuo2020learning}, we define the top-$k$ retrieval Bayes optimality under beam search to adapt to our hierarchical multi-class classification setting:
\begin{Definition}[Top-$k$ Retrieval Bayes Optimal]Given the beam size $k$, a tree-based model consist of a tree $\mathcal{T}$ and a preference model $\mathcal{M}$ is top-$k$ retrieval Bayes optimal, if the following equation
    \begin{equation}
    \widehat{\mathcal{Y}}(u) \in \operatornamewithlimits{argTopk}\limits_{y\in \mathcal{Y}} \eta_y(u)
    \end{equation}
\label{def: top-k optimality}
holds for any $u\in \mathcal{U}$.
\end{Definition}
The underlying idea of {Definition \ref{def: top-k optimality}} is straightforward: for each user $u$, a top-$k$ retrieval Bayes optimal tree-based model should return the $k$ items with the highest conditional probabilities. 

Building on the {Definition \ref{def: top-k optimality}}, we show that optimizing tree-based models using multi-class cross-entropy loss can result in suboptimal outcomes under beam search. 
Our demonstration leverages Bregman divergence~\cite{banerjee2005optimality}. Let $N$ denote the number of classes. For a strictly convex and differentiable function $\psi: \mathbb{R}^N \to \mathbb{R}$, the Bregman divergence $D_\psi: \mathbb{R}^N \times \mathbb{R}^N \to \mathbb{R}$ induced by $\psi$ is defined as follows:
\begin{equation}
    D_{\psi}(\boldsymbol{z},\boldsymbol{o}) = \psi(\boldsymbol{z})-\psi(\boldsymbol{o})-\nabla\psi(\boldsymbol{o})^T(\boldsymbol{z}-\boldsymbol{o}),
\end{equation}
where $\boldsymbol{z}, \boldsymbol{o}\in \mathbb{R}^N$. Previous work \cite{yang2020consistency, banerjee2005optimality} demonstrated that loss functions expressible as Bregman divergence have a specific rank-preserving property, and we present their theoretical result in the following. We first introduce the definition of rank consistent as follows: 
\begin{Definition}[Rank Consistent]
Given $\boldsymbol{x},\boldsymbol{y}\in\mathbb{R}^N$, we say that $\boldsymbol{x}$ is rank consistent w.r.t. $\boldsymbol{y}$, denoted as 
$R(\boldsymbol{x},\boldsymbol{y})$, if and only if for all $i,j\in \left[N\right]$,
\begin{equation*}
    x_i > x_j \Longleftrightarrow y_i > y_j.
\end{equation*}
\end{Definition}
\noindent Then, we state their theorem with modifications to adapt to our problem setting:
\begin{Theorem}[Theorem 3.1 of \cite{yang2020consistency}]
 Given a convex, differentiable function $\psi:\mathbb{R}^N\mapsto \mathbb{R}$. Let $\boldsymbol{z}$ be a random vector taking values in $\mathbb{R}^N$ for which both $\mathbb{E}[\boldsymbol{z}]$ and $\mathbb{E}[\psi(\boldsymbol{z})]$ are finite. If continuous function $g:\mathbb{R}^N\mapsto \mathbb{R}^N$ satisfies that $R(\boldsymbol{s},g(\boldsymbol{s}))$ holds for $\forall \boldsymbol{s}\in \operatorname{domain}(g)$, and $\mathbb{E}[\boldsymbol{z}]\subseteq \operatorname{range}(g)$, then 
\begin{equation*}
    \operatornamewithlimits{argmin}\limits_{\boldsymbol{s}\in \mathbb{R}^N} \mathbb{E}_{\boldsymbol{z}}\left[D_{\psi}(\boldsymbol{z},g(\boldsymbol{s})\right] \subseteq  \left\{\boldsymbol{s}\in \mathbb{R}^N\mid R(\boldsymbol{s},\mathbb{E}_{\boldsymbol{z}}[\boldsymbol{z}])\right\}.
\end{equation*}
\label{theo: property of minimum of Bregman divergence}
\end{Theorem}

    


As mentioned in Section~\ref{subsubsec:multi-class cross-entropy loss}, a label $z_i^j$ is attached to the node $n_i^j$ for a given instance $(u, y)$. \textcolor{black}{When the user $u$ is fixed, $z_i^j$ is a conditional random variable, with its randomness depending on the item $y$. The variables in the $j$-th layer compose the random vector $\boldsymbol{z}^j$}.
As the label $z_{\delta^j(y)}^j$ is set to $1$ and other labels in the $j$-th layer are set to $0$ for the instance $(u,y)$, the expectation of $\boldsymbol{z}^j$ for a given user $u$ is as follows:
\begin{small}
\begin{equation}
    \mathbb{E}[\boldsymbol{z}^j|u]=\left(\sum\limits_{n\in \mathfrak{L}(n_1^j)}\eta_{\pi(n)}, \sum\limits_{n\in \mathfrak{L}(n_2^j)}\eta_{\pi(n)},\ \cdots,\sum\limits_{n\in \mathfrak{L}\left(n_{N_j}^j\right)}\eta_{\pi(n)}\right).
    \label{eq: expectation of layer random vector}
\end{equation}
\end{small}For simplicity, we omit $u$ and use notation $\eta_{\pi(n)}$ to represent $\pi(n)$’s conditional probability given the user $u$.

With the fact that multi-class cross-entropy loss is a Bregman divergence~\cite{yang2020consistency} and the softmax function $g$ satisfies $R(\boldsymbol{s},g(\boldsymbol{s}))$ for any $\boldsymbol{s}\in \mathbb{R}^N$, 
according to  Theorem \ref{theo: property of minimum of Bregman divergence}, when given user $u$, minimizing the conditional risk of multi-class cross-entropy loss of layer $j$, i.e.,
\begin{small}
\begin{equation}
    \mathcal{R}_j(u)=\mathbb{E}_{y\sim \eta_y(u)}\left[\mathcal{L}_j(u,y)\right]
    =
    \mathbb{E}_{y\sim \eta_y(u)} \left[-\log \frac{\exp o_{\delta^j(y)}^j}{\sum_{k=1}^{N_j} \exp o_k^j}\right],
    \label{eq: conditional risk}
\end{equation}
\end{small}aligns the ranks of each component in $\boldsymbol{o}^j$ with those in $\mathbb{E}\left[\boldsymbol{z}^j|u\right]$ (i.e., $R(\boldsymbol{o}^j, \mathbb{E}[\boldsymbol{z}^j|u])$), where  $\boldsymbol{o}^j$ is the preference score vector for the $j$-th layer whose $i$-th component is $o^j_i$. Since the multi-class cross-entropy loss ({Eq. (\ref{eq:multi_cross_entropy_single_layer})}) is optimized per layer, $R(\boldsymbol{o}^j, \mathbb{E}[\boldsymbol{z}^j|u])$ holds for $1\le j\le H$. Therefore, suboptimality under beam search can be identified.

\begin{Proposition}
 The multi-class cross-entropy loss results in the tree model not being top-$k$ retrieval Bayes optimal. 
\label{pro:suboptimality of multi-class cross-entropy loss}
\end{Proposition}
\begin{proof}
See Appendix~\ref{appendix: proof of proposition suboptimality of multi-class cross-entropy loss}.
\end{proof}

\subsection{Rectification Under Beam Search}
\label{sec: rectification loss}
This section introduces our method for mitigating suboptimality under beam search. Our approach involves modifying the label assignment for non-leaf nodes and adapting the loss function accordingly.
Furthermore, we provide the theoretical support for these modifications and propose effective solutions for their practical implementation.

\subsubsection{Label Rectification}
For an observed instance \((u,y)\), the ancestors of the leaf node \(\pi^{-1}(y)\) are not labeled as \(1\) by default. The label \(z_i^j\) is set to \(1\) only when item \(y\) attains the largest conditional probability among the items in \(\pi(\mathfrak{L}(n_i^j))\).
We slightly abuse notation by using \(\pi(\mathfrak{L}(n_i^j))\) to denote the set of items mapped from the leaf nodes in \(\mathfrak{L}(n_i^j)\).
Denote the rectified label by $\bar{z}^j_i$, the formalization of assignment is as follows:
\begin{equation}
    \bar{z}_i^j = \begin{cases}
    1, & n_i^j = \rho^j\left(\pi^{-1}(y)\right) \wedge y=\operatornamewithlimits{argmax}\limits_{y'\in \pi\left(\mathfrak{L}(n_i^j)\right)} \eta_{y'}\\
    0, & \text{else}.
    \end{cases}
    \label{eq: unnormaliz}
\end{equation}
An illustrative example of the label-rectification assignment is provided in Appendix~\ref{appendix: illustration of label rectification}.
Under this rectified assignment, given the user $u$, for each node $n_i^j$, we can calculate the expectation of the corresponding random variable: $$\mathbb{E}_{y\sim\eta_y(u)}[\bar{z}_i^j|u]=\sum_{y\in \mathcal{Y}}\bar{z}_i^j*\eta_y(u)= \max\limits_{n\in \mathfrak{L}(n_i^j)} \eta_{\pi(n)},$$ 
so the expectation of $j$-th layer's random vector is 
\begin{small}
\begin{equation}
\label{eq: the expectation of unnormalized rectified label}
    \mathbb{E}[\bar{\boldsymbol{z}}^j|u]=\left(\max\limits_{n\in \mathfrak{L}(n_1^j)} \eta_{\pi(n)}, \max\limits_{n\in \mathfrak{L}(n_2^j)} \eta_{\pi(n)}, \dots, \max\limits_{n\in \mathfrak{L}(n_{N_j}^j)} \eta_{\pi(n)}\right),
\end{equation}
\end{small}indicating the rectified labels can align the tree with the max-heap in expectation.

Since we regard the recommendation task as a multi-class classification problem, the conditional probabilities for a given user $u$ satisfy the normalization condition $\sum_{y\in \mathcal{Y}} \eta_y(u) = 1$. The current expectations of each layer's random variables corresponding to the rectified label $\bar{z}_i^j$ are not normalized; hence, a normalization term $\alpha^j(u)=\sum_{i=1}^{N_j}\bar{z}_i^j(u)$ should be applied to the unnormalized rectified label, resulting in the following label assignment:
\begin{Definition}[Label Rectification]
For a given instance $(u,y)$, the normalized rectified label $\tilde{z}_i^{j}$ of node $n_i^j$ is defined as follows:
\begin{equation}
    \tilde{z}_i^j = \begin{cases}
    \frac{1}{\alpha^j(u)}, & n_i^j = \rho^j\left(\pi^{-1}(y)\right) \wedge y=\operatornamewithlimits{argmax}\limits_{y'\in \pi\left(\mathfrak{L}(n_i^j)\right)} \eta_{y'}\\
    0, & \text{else}.
    \end{cases}
\label{eq: normalized rectified label}
\end{equation}
\label{def: label rectification}
\end{Definition}

\vspace{-2em}
\subsubsection{The Induced Loss Function}
With label rectification, the target vector for each layer is no longer a one-hot vector; instead, the target label is now a real number $\tilde{z}_i^j$.
To adapt to the label rectification and still utilize the  property of Bregman divergence, the following loss function is proposed:
\begin{equation}
        \widetilde{\mathcal{L}}(\boldsymbol{e}^{(i)}(w), \boldsymbol{o})=-w\log\left(\frac{\exp{o}_i}{\sum_{k=1}^N\exp {o}_k}\right)+ w\log w + w-1.
        \label{eq: modified cross-entropy}
\end{equation}
 Here, $w\in [0, +\infty)$, $\boldsymbol{e}^{(i)}(w),\boldsymbol{o}\in \mathbb{R}^N$,  and  $\boldsymbol{e}^{(i)}(w)=(0,\dots,0,w,0,\dots,0)$ where the $i$-th component is $w$, which will be set to $\tilde{z}^j_i$ as presented in Eq. \eqref{eq: normalized rectified label} to adapt to the label rectification, and other components are $0$. 
 
 Compared to the original multi-class cross-entropy loss, the loss function in {Eq. (\ref{eq: modified cross-entropy})} includes a slight modification by incorporating a weighting factor and a constant addition for a given $w$, 
 which is fully compatible with the rectified label. Importantly, it retains its nature as a Bregman divergence, allowing for theoretical analysis based on Theorem~\ref{theo: property of minimum of Bregman divergence}.
\begin{Proposition}
      The loss function in {Eq. (\ref{eq: modified cross-entropy})} can be expressed as Bregman divergence.
    \label{pro: modified softmax Bregman divergence}
\end{Proposition}
\begin{proof}
    See Appendix~\ref{appendix: proof of proposition modified softmax Bregman divergence}.   
\end{proof}

\subsubsection{Theoretical Support}
With the above label rectification method and the induced loss function, for a given instance $(u,y)$, the following loss is optimized in the $j$-th layer: $$\widetilde{\mathcal{L}}_j(u,y) = \widetilde{\mathcal{L}}(\boldsymbol{e}^{(\delta^j(y))}(\tilde{z}^j_{\delta^j(y)}), \boldsymbol{o}^j).$$ 
As established in {Proposition \ref{pro: modified softmax Bregman divergence}}, $\widetilde{\mathcal{L}}$ still owns the property of Bregman divergence.
Notice with the normalization, $\mathbb{E}[\tilde{\boldsymbol{z}}^j|u]\subseteq \operatornamewithlimits{range}(g)$ holds for $1\le j \le H$. Then, for any user $u\in \mathcal{U}$, by the {Theorem \ref{theo: property of minimum of Bregman divergence}}, we have $$\operatornamewithlimits{argmin}\limits_{\boldsymbol{o}^j\in \mathbb{R}^{N_j}}\widetilde{\mathcal{R}}_j(u)\subseteq \left\{\boldsymbol{o}^j\in \mathbb{R}^{N_j}| R(\boldsymbol{o}^j,\mathbb{E}[\tilde{\boldsymbol{z}}^j|u])\right\},$$
where $ \widetilde{\mathcal{R}}_j(u)=\mathbb{E}_{y\sim \eta_y(u)}[\widetilde{\mathcal{L}}_j(u,y)]$ is the conditional risk for user $u$ of $j$-th layer. This indicates that optimizing the modified loss with rectified labels guides the tree model towards being  top-$k$ retrieval Bayes optimal under beam search. Specifically, we have the following theoretical result:

\begin{Proposition}
    For any user $u\in \mathcal{U}$, if the tree model, which consists of a tree $\mathcal{T}$ and a preference model $\mathcal{M}$, satisfies that $R(\boldsymbol{o}^j, \mathbb{E}[\tilde{\boldsymbol{z}}^j|u])$ for all $1\le j \le H$, then the tree model is top-$k$ retrieval Bayes optimal for any beam size $k$.
    \label{prop: optimality of label z under beam search}
\end{Proposition}
     
The proof of Proposition~\ref{prop: optimality of label z under beam search} is provided in Appendix
~\ref{appendix: proof of proposition optimality of label z under beam search}.
The above theoretical analysis demonstrates that optimizing the modified loss function with the normalized rectified label can mitigate the issue of suboptimality under beam search of standard multi-class cross-entropy loss. 

\subsubsection{Practical Implementation}
\label{sec: practical implementation}
The above analysis and methodology rely on the conditional probability given user $u$, which is unknown in practice. A direct way would be to estimate the conditional probability using the preference model $\mathcal{M}$. However, since $\mathcal{M}$ is a complicated neural network, estimating the conditional probabilities through it is computationally expensive. Additionally, because $\mathcal{M}$ is randomly initialized, its estimation can be highly inaccurate, leading to convergence issues or the risk of falling into local optima. Considering the above shortcomings, we employ a pre-trained small model to estimate the conditional probability; detailed implementation is provided in the experimental section (Subsection~\ref{subsec: settings}).

In practice, we optimize the loss using stochastic gradient descent. The term $w\log w + w -1$ has zero gradient w.r.t. preference score for a given $w$, and can be disregarded in the practical optimization. Therefore, the loss function (i.e., {Eq. (\ref{eq: modified cross-entropy})}) simplifies to the standard multi-class cross-entropy loss multiplied by $w$. While $w$ should be set as the normalized rectified label $\tilde{z}$  based on our above analysis, we find that this normalization can impede stable training and use the unnormalized rectified label $\bar{z}$ in practice. 

In fact, the unnormalized rectified label is sufficient to mitigate suboptimality, as its expectation still aligns the tree with a max-heap  ({Eq. (\ref{eq: the expectation of unnormalized rectified label})}). The role of normalization is to make the sum of expectations of random variables $\bar{z}^j_i$ of the $j$-th layer equal $1$, to satisfy the constraints in the theoretical analysis.
Compared to the unnormalized label $\bar{z}$, normalization actually only influences the learning rate in the stochastic gradient descent. Note that the normalization coefficients across different layers vary dramatically for a given training instance, but the loss of all layers is aggregated together for optimization, so the normalization will make the optimization hard in practice due to the difficulty to set a suitable learning rate. Therefore, we optimize the modified loss with the original rectified label $\bar{z}^j_i$ without normalization, which can be interpreted as an adaptive adjustment of the learning rate in each layer, and the loss for $j$-th layer is
\begin{small}
\begin{equation}
    \widetilde{\mathcal{L}}_j(u,y)=-\bar{z}^j_{\delta^j(y)}\log\left(\frac{\exp o^j_{\delta^j(y)}}{\sum_{k=1}^{N_j}\exp(o^j_k)}\right),
    \label{eq: rectified loss single layer}
\end{equation}   
\end{small}and the loss of the whole tree is
\begin{small}
\begin{equation}
    \widetilde{\mathcal{L}}(u,y)=\sum_{j=1}^H \widetilde{\mathcal{L}}_j(u,y).
    \label{eq: rectified loss the whole tree}
\end{equation}    
\end{small}

\subsection{Negative Sampling to Speed Up Training}
\label{sec: negative sampling}

For the training instance $(u,y)$, calculating its training loss (i.e., {Eq. (\ref{eq: rectified loss the whole tree})}) requires computing the softmax probability of the positive node at each layer. In practical applications, where the corpus size ranges from millions to billions or even tens of billions, computing softmax probabilities at such scales is extremely inefficient. To promote training efficiency, we employ the sampled softmax technique to estimate the gradient of softmax loss and propose a tree-based sampling method for more accurate gradient estimation.

\subsubsection{Unbiased Estimation of Softmax Loss Gradient}
Sampled softmax is proposed to approximate full softmax during training \cite{ bengio2008adaptive,wu2024effectiveness,chen2025adaptive}. Instead of calculating the training loss of layer $j$ (i.e., {Eq. (\ref{eq:multi_cross_entropy_single_layer})}) over all classes, sampled softmax only considers the positive class and $M$ negative classes. These $M$ negative classes are sampled from all the negative classes according to certain sampling distribution $Q$ with replacement. In the following, we use $\mathcal{I}_M^j$ to denote the set of indices of $M$ negative nodes at layer $j$.
Then, the indices of training nodes in layer $j$ is $\mathcal{I}^j_M \cup \{\delta^j(y)\}$ for a given instance $(u,y)$.
For example, if $\delta^j(y)=2$, $M=5$ and the negative indices set is  $\mathcal{I}^j_M=\{3,4,8,7,4\}$, this indicates that the $2$nd node of layer $j$ is the positive node, the $4$th node of layer $j$ is sampled twice, and the other nodes (indexed at $3,8,7$) are each sampled once.

To reduce the approximation bias, we do not directly use the outputs w.r.t. sampled nodes to approximate the original loss~\cite{bengio2008adaptive}. For $i\in \mathcal{I}^j_M \cup \{\delta^j(y)\}$, a slight adjustment is conducted for each output by 
\begin{equation}
\label{eq:adjusted_output}
    \hat{o}_{i}^j=\left\{
    \begin{aligned}
    &o^j_{i}-\ln(Mq^j_{i})\ \ \ \text{if}\ \ i\neq\delta^j(y), \\
    &o^j_{i}-\ln(1)\ \ \ \ \ \ \ \ \ \text{if}\ \ i=\delta^j(y)
\end{aligned}
\right.
\end{equation}
where $q_{i}^j$ 
denotes the probability of sampling node $n_{i}^j$ from the negative nodes of layer $j$. 
This adjustment guarantees that the gradient of the sampled softmax loss is an unbiased estimator of the true gradient as $M\to \infty$ \cite{bengio2008adaptive}. Then, the loss is calculated over the adjusted outputs, and the training loss {Eq. (\ref{eq: rectified loss single layer})} at layer $j$ can be adjusted to
\begin{small}
\begin{equation}
\label{eq:sampled_softmax_loss_single_layer}
    \widehat{\mathcal{L}}_{j}(u,y)=-\bar{z}^j_{\delta^j(y)}\log \hat{p}^j_{\delta^j(y)}=-\bar{z}^j_{\delta^j(y)}\log \frac{\exp \hat{o}_{\delta^j(y)}^j}{\sum\limits_{k\in \mathcal{I}^j_M \cup \{\delta^j(y)\}} \exp \hat{o}_{k}^j}.
\end{equation}
\end{small}

Previous literature~\cite{bengio2008adaptive,blanc2018adaptive} has proved that proper specified sampling distribution $Q$ used in sampled softmax can lead to an unbiased estimation of the gradient for the original loss, i.e., $\mathbb{E}_{\mathcal{I}^j_M\sim Q}[\frac{\partial \widehat{\mathcal{L}}_j(u,y)}{\partial o_i^j}]=\frac{\partial \widetilde{\mathcal{L}}_j(u, y)}{\partial o_{i}^{j}}$. We state their theorem in a different way so that it can be compatible with our method:

\begin{Theorem}[Theorem 2.1 of ~\cite{blanc2018adaptive}]
\label{theo: unbiased gradient estimation theory}
The gradient of loss {Eq. (\ref{eq:sampled_softmax_loss_single_layer})} 
w.r.t. sampled softmax is an unbiased estimator of the gradient of the loss  {Eq. (\ref{eq: rectified loss single layer})} 
w.r.t. full softmax if and only if $q^j_i\propto \exp{o_i^j}$ holds where $1\le i\le N_j$ and the $i$-th node isn't a positive node at layer $j$. 
\label{theo:unbiased_estimator_of_gradient}
\end{Theorem}
This theorem indicates that if the sampling probability is proportional to the exponential of the negative node's preference, we can get an unbiased estimation of the training loss gradient. However, it still  incurs linear time complexity w.r.t. the corpus size, similar to calculating the full softmax, because it requires computing the preference score for every node. To address this problem, previous literature ~\cite{blanc2018adaptive, rawat2019sampled} develops kernel-based approximations for inner product models. Nonetheless, their kernel-based methods aren't suitable for our approach, as our preference model is a complicated neural network. 
To reduce the computational complexity while maintaining compatibility with our approach, we propose a tree-based sampling method.

\subsubsection{Tree-based Sampling}
In many studies, uniform sampling is employed due to its simplicity and efficiency, which also applies to our method, where the negative nodes at each layer can be obtained through a uniform distribution.
However, the uniform distribution often diverges significantly from the exact softmax distribution. To better approximate the gradient, we propose a tree-based sampling method that leverages both the tree index and the preference model to generate a distribution closer to exact softmax distribution.

Our proposed tree-based sampling method is implemented by modifying the beam search retrieval process ({Algorithm \ref{alg:beam_search}}). Firstly, during sampling along the tree, only one node is retained at each layer (i.e., we set $k=1$ in {Algorithm \ref{alg:beam_search}}).  Secondly, instead of directly choosing the child node with the highest preference score, we sample one child for expansion, according to local softmax probabilities based on the preference scores of all child nodes under the current node. Concretely, after expanding node $n_i^j$, the expanding probability for node $n_c^{j+1}\in \mathcal{C}(n_i^j)$ is:
\begin{small}
\begin{equation}
\tilde{q}_c^{j+1}=\frac{\exp{o_c^{j+1}}}{\sum\limits_{k\in \left\{\delta(n)\mid n\in \mathcal{C}(n_i^j)\right\}} \exp{o_k^{j+1}} }.
\label{eq: expanding probability}
\end{equation}    
\end{small}Then, by expanding one node per layer, we obtain a sample path from root to leaf; the sampled nodes along this path serve as the negative nodes for their respective layers. This sampling process is repeated until $M$ negative nodes are obtained for each layer. During sampling, the sampling probability for node $n_i^j$ is the product of expanding probabilities along the sample path from the root to itself:
\begin{small}
\begin{equation}
    q_i^j= \prod_{m=0}^j \tilde{q}_{\delta(\rho^m(n_i^j))}^{\mspace{2mu}m}
    \label{eq: sampling probability in tree-based sampling}
\end{equation}  
\end{small}The correctness of the above sampling process can be affirmed by the fact that the sum of sampling probabilities of nodes within the same layer equals $1$:
\begin{Proposition}
\label{prop: the correctness of sampling probability}
    For any layer $j\ (0\le j \le H)$, the sum of  sampling probabilities via tree-based sampling for $j$-th layer's nodes is exactly $1$, i.e.,
    \begin{small}
    $$
     \sum_{i=1}^{N_j} q_i^j=\sum_{i=1}^{N_j} \prod_{m=0}^j \tilde{q}_{\delta(\rho^m(n_i^j))}^{\mspace{2mu}m} = 1.
    $$    
    \end{small}  
\end{Proposition}\noindent The proof of the above proposition is provided in Appendix~\ref{appendix: proof of correctness of sampling probability}. 
Moreover, our proposed sampling method has a lower time complexity, as the number of sampling operations is proportional to the tree height and the sampling probability (i.e., local softmax probability) can be computed efficiently because it only relies on scores of a limited number of child nodes.


{Figure \ref{fig: Tree-based sampling}} illustrates the tree-based sampling method. The sampling process starts from the root node $0$, whose child nodes are node $1$ and node $2$, with expanding probabilities of  $0.3$ and $0.7$, respectively. In this case, node $1$ is sampled for expansion, and the process continues from node $1$, repeating until the leaf node $8$ is sampled.  Finally, the path $0\rightarrow 1 \rightarrow 3 \rightarrow 8$ is sampled, and the corresponding negative nodes are $\{1,3,8\}$. For any node, its sampling probability is the product of the expanding probabilities along the path from the root to itself. For example, in {Figure \ref{fig: Tree-based sampling}}, the sampling probability for the node $3$ is \(0.18 = 0.3 \times 0.6\), and for the node $8$ is \(0.144 = 0.3 \times 0.6 \times 0.8\). Through simple calculation, we can verify that the sum of the probabilities of each node being sampled is exactly $1$ for any layer.
\begin{figure}[ht]
    \centering
    
 \includegraphics[width=\linewidth]{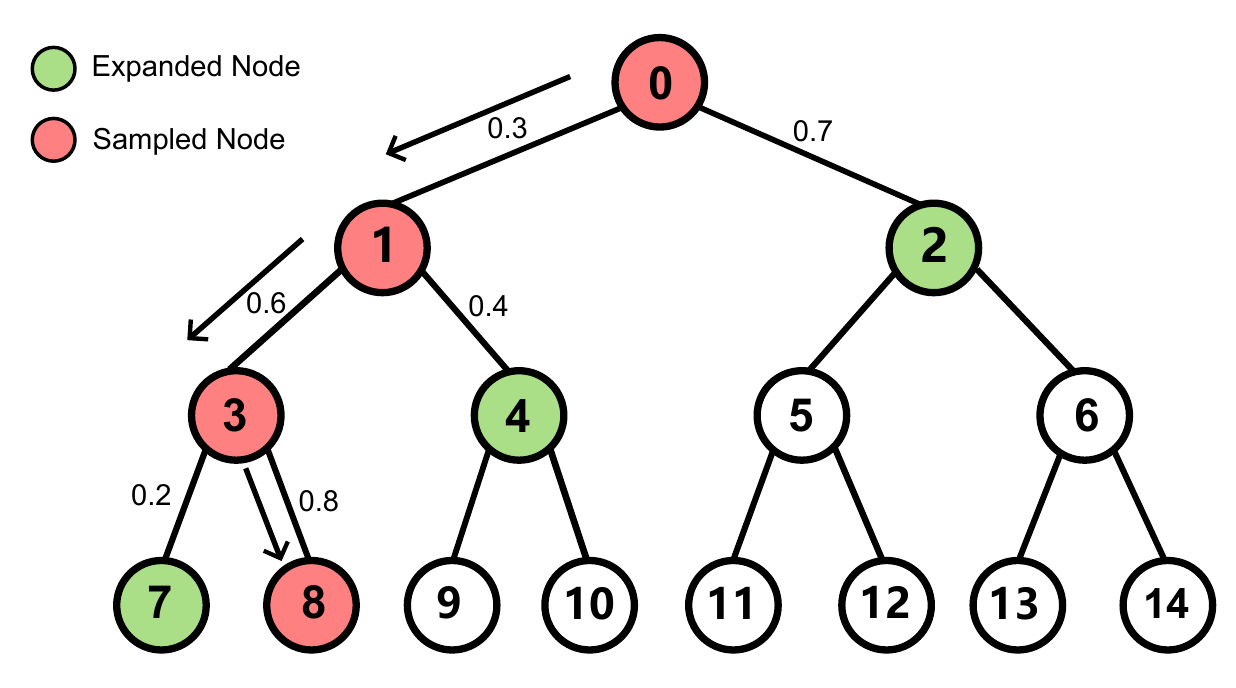}
    \caption{Illustration of tree-based sampling along the tree.  The number beside the edge is the child's expanding probability given the parent.}
    \label{fig: Tree-based sampling}
\end{figure}

\subsubsection{Theoretical Support}
\label{sec: theoretical support of tree-based sampling }
To obtain an unbiased estimation of the gradient for the original training loss, Theorem \ref{theo: unbiased gradient estimation theory} dictates that the sampling probability for each node should be proportional to the exponential of its preference score. Our proposed tree-based sampling method is grounded in this principle. By leveraging the tree index and the preference model during the sampling, it directly utilizes the node's local softmax probability for negative sampling, which promotes the sampling distribution closer to the softmax distribution. Furthermore, when the preference scores 
satisfy certain conditions, the sampling probability of any node in the tree-based sampling will align precisely with that node's softmax probability, as formalized in the following theorem:
\begin{Theorem}
\label{theo: tree-based sampling eq to softmax distribution}
For any node \(n_i^j\) of the tree, if the preference score $o_i^j$ satisfies the following condition:
    \begin{small}
     \begin{equation}
    \label{eq:score_proportional_to_children}
    \exp o_i^j \propto \sum_{k \in \left\{\delta(n)\mid n\in \mathcal{C}(n_i^j)\right\}} \exp o_k^{j+1},
    \end{equation}   
    \end{small}and the proportionality coefficients within the same layer are equal, then the sampling probability of node $n_i^j$ in the tree-based sampling equals its softmax probability in layer $j$, expressed as:
\begin{small}
\begin{equation}
    \label{eq: softmax sampling probability in tree-based sampling}
    q_{i}^j = \frac{\exp o_{i}^j}{\sum_{k=1}^{N_j} \exp o_{k}^j}.
\end{equation}    
\end{small}
\label{theo: sampling theory to unbiased softmax}
\end{Theorem}

Although Theorem~\ref{theo: sampling theory to unbiased softmax} imposes a relatively stringent condition that may not be fully satisfied in practice, it is noteworthy that optimizing the standard multi-class cross-entropy loss can effectively align the score preferences with this condition. We provide the details in Appendix~\ref{appendix: details of condition alignment}.

Within the rectification method under beam search, the tree aligns the max-heap in expectation; then, the optimization will lead to the following results:
\begin{small}
\begin{equation}
\exp o_i^j \propto \max_{k \in \left\{\delta(n)\mid n\in \mathcal{C}(n_i^j)\right\}} \exp o_k^{j+1}.
\label{eq:score_proportional_to_max_children}
\end{equation}    
\end{small}Since 
\begin{equation}
\lambda_i^j\triangleq \frac{\max_{k \in \left\{\delta(n)\mid n\in \mathcal{C}(n_i^j)\right\}} \exp o_k^{j+1}}{\sum_{k \in \left\{\delta(n)\mid n\in \mathcal{C}(n_i^j)\right\}} \exp o_k^{j+1}}  
\label{eq: definition of lambda}
\end{equation} 
is approximately $1$ when $\max_{k \in \left\{\delta(n)\mid n\in \mathcal{C}(n_i^j)\right\}}o_k^{j+1}$ is much higher than the preference scores of other child nodes, the sampling probability remains close to the softmax probability. 
For any node \(n_i^j\), we define its probability bias between its sampling probability $q_i^j$ and softmax probability:
   \begin{small}
       $$\operatorname{bias}_i^j = q_i^j - \frac{\exp o_{i}^j}{\sum_{k=1}^{N_j} \exp o_{k}^j}.$$
   \end{small}
Specifically, we have the following proposition:
\begin{Proposition}
   Suppose the preference score $o_i^j$ satisfies the condition in Eq.~(\ref{eq:score_proportional_to_max_children}), the proportionality coefficients within the same layer are equal,  and $\lambda_i^j\in (1-\varepsilon,1)$, where $\varepsilon$ is a positive real number that is close to $0$. The probability bias of node $n_i^j$ can be controlled as follows:
\begin{small}
    \begin{equation}
  \left|\operatorname{bias}_i^j\right|\le \frac{(j-1)\varepsilon}{1-\varepsilon}.
\end{equation}
\end{small}   
\label{prop: rectification method bounds probability bias} 
\end{Proposition}
\noindent Based on the above proposition, we can further derive the following corollary:
\begin{Corollary}
\label{cor: bound of KL divergence}
    For the $j$-th layer, suppose the preference scores $o_i^j$ \((i\in [N_j])\) are bounded by a constant $B_o$.  The KL divergence between sampling distribution $Q^j$ w.r.t. softmax distribution $P^j$ can be bounded as follows:
    \begin{equation}
        D_{KL}(Q^j||P^j) \le \log \left(1+ B^j_{\varepsilon}N_je^{2B_o}\right),
    \end{equation}
    where $B_\varepsilon^j\triangleq \frac{(j-1)\varepsilon}{1-\varepsilon}$ is the upper bound of the probability bias of nodes in layer $j$.
\end{Corollary}

The proofs of Theorem \ref{theo: sampling theory to unbiased softmax}, Proposition \ref{prop: rectification method bounds probability bias} and Corollary \ref{cor: bound of KL divergence} are presented in Appendix \ref{appendix: proof of theory tree-based sampling eq to softmax distribution}, \ref{appendix: proof of prop bounded probability bias} and  \ref{appendix: proof of corollary bound}, respectively. Theorem \ref{theo: sampling theory to unbiased softmax} provides valuable insights for designing sampling schemes and justifies our tree-based sampling method. Furthermore, Proposition \ref{prop: rectification method bounds probability bias} and Corollary \ref{cor: bound of KL divergence} demonstrate that this sampling method also applies to the rectification method under beam search. It can be observed that when $\varepsilon$ is small, the difference between the sampling distribution and softmax distribution is also small, leading to a more accurate estimation of the full softmax loss gradient.


\subsection{Tree Index Updating}
\label{sec: tree index updating}
Without sufficient domain knowledge, initializing a well-structured tree is challenging. Even if the preference model converges during training, it can only adapt to the current, potentially suboptimal tree. Therefore, the tree index (i.e., the mapping $\pi$) should be updated alongside the optimization of preference model. In this section, we introduce how to update the tree with the fixed preference model.

Initially, we construct the tree using a method similar to TDM and JTM, making similar items organized in close positions on the last layer of the tree.  It's natural to rely on the category information of items to build the initial tree. Firstly, we shuffle all the categories and group items by category ID. Secondly, the items within each category are recursively divided into two equal parts until each partition contains only one item.  In this way, we can construct a near-complete binary tree from the root node to the leaf node. A simple balancing strategy can be used to adjust the tree. Specifically, for leaf nodes not at the deepest level, an internal node is repeatedly inserted between the leaf and its parent until all leaf nodes are positioned at the same depth. We use a binary tree in our experiments. In fact, any k-ary tree can be constructed in such a way by partitioning each cluster into more nearly equal parts at each iteration.

Tree index updating is also based on the training set $S=\{(u_i,y_i)\}_{i=1}^m$.  We partition $S$ into different user-item pair sets, denoted by $\mathcal{A}_{y}=\{(u_i, y_i)| y_i=y\}$, where $\mathcal{A}_{y}$ consists of all user-item pairs for which the item is exactly $y$. Given the old mapping $\pi$ and each user-item pair set $\mathcal{A}_{y}$, we fix the preference model and assign each item to the tree nodes step by step, from the root node to the leaf nodes, to obtain the new mapping $\pi'$. 

All items are initially assigned to the root node, i.e., current level $j=0$. Then, we try to assign them to the nodes at level $j^\prime \triangleq min(H,j+d)$. Here, the hyperparameter $d$ determines the number of layers to skip during the allocation process along the tree,  where the purpose of skipping layers is to accelerate the allocation and, to some extent, prevent items from being assigned to the inappropriate leaf node due to early incorrect allocation.
For an item $y$ initially assigned to the root node, it must be reassigned to one of the nodes at level $j^\prime$. Since we are using a binary tree, there will be $c\le 2^d$ candidate nodes. Without loss of generality, we denote these candidate nodes as $n_1^{j^\prime},n_2^{j^\prime},\dots,n_c^{j^\prime}$. The preference of user $u$ 
for node $n_i^{j^\prime}$ ($1\le i\le c$) is denoted as $o_i^{j^\prime}$ and is calculated by the fixed preference model. We define the matching score between item $y$ and node $n_i^{j^\prime}$ ($1\le i \le c$) as follows:
\begin{equation}
\label{eq:item_node_matching_socre}
\begin{aligned}
    Score(y,n_i^{j^\prime})=\sum_{(u,y)\in \mathcal{A}_y}\log \frac{\exp{o_i^{j^\prime}}}{\sum_{k=1}^c\exp{o_k^{j^\prime}}}.
\end{aligned}
\end{equation}
The above matching score represents the log-likelihood of assigning item \( y \) to node \( n_i^{j'} \), calculated based on \( \mathcal{A}_y \) and the fixed preference model. A higher score indicates a greater suitability for allocating the item to the corresponding node.
Then, we rank the $c$ scores and assign $y$ to the node with the highest matching score. All items are disjointly assigned to the nodes at level $j^\prime$, where the number of items assigned to the node $n_i^{j^\prime}$ equals the number of leaf nodes within the subtree rooted at $n_i^{j^\prime}$. If the number of assigned items exceeds the limitation, $y$ is reassigned to the node with the second-highest matching score, or to subsequent nodes in descending order of matching scores, until the assignment satisfies the limit. In this way, we can assign each item belonging to the root initially to a node at level $j^\prime$. Now, we regard $n_1^{j^\prime},\dots,n_c^{j^\prime}$ as root nodes of $c$ subtrees and repeat the process recursively to assign the items to deeper layers until each leaf node is assigned an item. Finally, we can get the new mapping $\pi'$. Computing the matching score (i.e., {Eq. (\ref{eq:item_node_matching_socre})}) requires a softmax calculation over $c$ candidate nodes. Therefore, $d$  must be kept small to fit the time complexity and computational resource; in our experiments, we set $d=7$.

To illustrate the tree update process, {Figure \ref{fig: Updating Mapping}} presents an example with eight items labeled $a$ to $h$ and a step $d=2$. Initially, all items are at the root node. We need to allocate items to the layer at a distance of $2$ from the root. This layer consists of four nodes, with each node receiving two items. According to the matching scores, item $f$ should ideally be assigned to node $4$, which has the highest matching score. However, since node $4$ already contains two items ($b$ and $c$), item $f$ is reallocated to node $5$, which has the next highest score. Once allocation in the second layer is complete, the process continues to deeper layers. As the leaf nodes are located in the third layer, just one step from the second layer (less than $d$), items are directly assigned to the third layer. Ultimately, each leaf node is allocated one item, completing the tree update.
\begin{figure}[ht]
    \centering
    \includegraphics[width=0.9\linewidth]{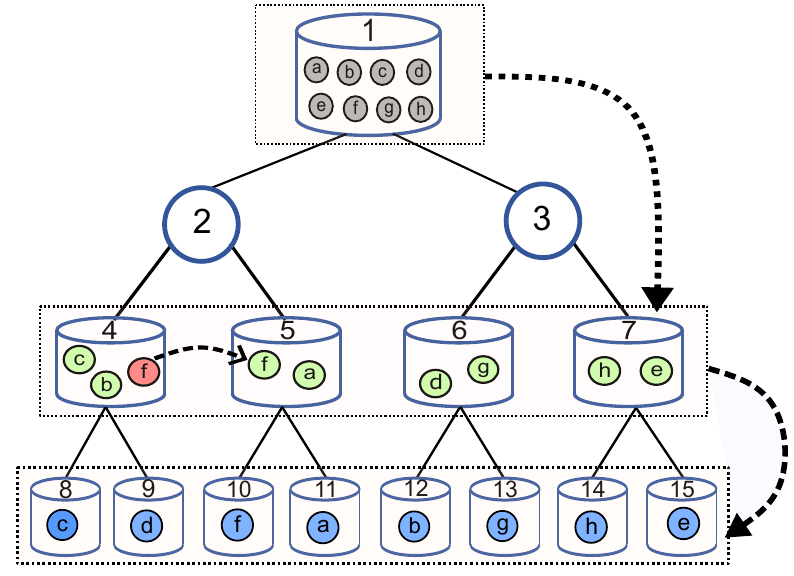}
    \caption{The process of updating the mapping between items and leaf nodes, where $d=2$ in this case.}
    \label{fig: Updating Mapping}
\end{figure}

In fact, our tree learning strategy assigns items to tree nodes in a top-down manner, which can also be viewed as a hierarchical, coarse-to-fine clustering where $d$ controls the granularity. The main difference between ours and the one in JTM is the calculation of matching scores. 
Both JTM and our method compute the log-likelihood as the matching score, but while JTM directly sums up the binary classification probabilities for this calculation, our strategy relies on the softmax-based multi-class probabilities, which is compatible with our training mode.

\section{Generalization Analysis}
\label{sec: generalization analysis}
In this section, we provide a generalization analysis of the deep tree-based retriever.
Our analysis, based on Rademacher complexity, derives a data-dependent generalization error bound for our proposed model. The final bound synthesizes three key aspects: the complexity of the model optimized with our rectified loss, the bias introduced by the sampled softmax technique, and the error from the external probability estimation.

\subsection{Generalization Bound of Modified Loss}
The variant of DIN is used as the preference model, which is formalized in Section~\ref{sec: formalization of DIN}. Suppose the norm of the weight matrix is bounded, we can get such function space for the model:
\begin{equation}
\begin{split}
    \mathcal{M}=\Big\{f:(u,n) \mapsto f_{\text{DIN}}(u,n) \mid \|\boldsymbol{w}_n\|_2\le B_0, \forall n \in \mathcal{N}; \\ \|\boldsymbol{W}_k\|_1 \le B_1, \forall k \in [L];\ \|\boldsymbol{W}_w^{(j)}\| \le B_2, \forall j \in {1,2} \Big\}.
\end{split}
\end{equation}

\begin{Lemma}[Rademacher Complexity of DIN]
\label{Lemma: Rademacher complexity of DIN}
Suppose $\forall u\in\mathcal{U},\ \forall k\in [K],\|\boldsymbol{a}_{k}^{(u)}\|_{2}\leq B_{a}$, then the Rademacher complexity of ${\mathcal{M}}$ can be bounded as follows:
\begin{equation}
    \widehat{\mathcal{R}}_{m}(\mathcal{M})\le \frac{2c_{\phi}^{L-1}B_1^L(B_0+KB_{w_1}+KB_{w_2}\mathcal{\tau})}{\sqrt{m}},
\end{equation}
where $B_{w_1}=c_{\phi}^2B_2^2B^2_a$, $B_{w_2}=c_{\phi}^2B_0B_2^2\left(B_a^2+B_a\right)$, and $\tau=\sqrt{(2B-1)|\mathcal{Y}|}/\sqrt{B-1}$.
\label{lemm: Rademacher complexity of DIN}
\end{Lemma}

The proof of Lemma \ref{lemm: Rademacher complexity of DIN} is provided in Appendix \ref{appendix: proof of lemm Rademacher complexity of DIN}. Building on the complexity  analysis of the DIN, we now derive the generalization bound for the modified loss with rectified labels (i.e., {Eq. (\ref{eq: rectified loss the whole tree})}) in Section~\ref{sec: rectification loss}. Since the overall loss for the tree is the sum of the losses at each layer, we begin by examining the generalization bound for the loss at the $j$-th layer. The result is presented as follows:
\begin{Lemma}
 \label{Lemma: generalization bound of single layer}
Suppose the function in $\mathcal{M}$ is bounded by a constant $B_\mathcal{M}$, the following inequality holds with a probability of at least $1-\delta$:
\begin{small}
   \begin{equation}
    \begin{split}
       \mathbb{E}_{(u,y)\sim \mathbb{P}}\left[\widetilde{\mathcal{L}}_j(u,y)\right] &\le \frac{1}{m} \sum_{i=1}^{m}\widetilde{\mathcal{L}}_j(u_i,y_i)   \\ &+8N_j\frac{c_{\phi}^{L-1}B_1^L(B_0+KB_{w_1}+KB_{w_2}\mathcal{\tau})}{\sqrt{m}}\\
       &+\Big(4\log N_j+8B_{\mathcal{M}}\Big)\sqrt{\frac{2\log\left(4/\delta\right)}{m}}.
    \end{split}
   \end{equation}
\end{small}
\end{Lemma}

\subsection{Bias of Sampled Softmax}
The analysis so far has established a generalization bound for the ideal case of minimizing the full softmax loss.  In practice, however, we adopt sampled softmax to accelerate optimization (see Section \ref{sec: negative sampling}). 
To align theoretical analysis with practice, we incorporate the bias introduced by the sampled softmax loss w.r.t. the full softmax loss into our theoretical result. For an instance $(u,y)$,  the model $\mathcal{M}$ outputs $N_j$ preference scores for nodes at $j$-th layer, forming a logit vector $\boldsymbol{o}^j \in \mathbb{R}^{N_j}$. Assuming the $i$-th node is positive (i.e., $i=\delta^j(y)$), the full softmax loss for layer $j$ is given by: 
$$\ell^j_{softmax}(u,y)=-\log \frac{\exp {o}^j_{i}}{\sum_{k=1}^{N_j}\exp {o}^j_k}.$$ 
A set of $M$ negative node indices, denoted by $\mathcal{I}^j_M$, is sampled from a distribution $Q^j(u, y)$, where $u$ and $y$  specify that the distribution depends on the instance $(u,y)$. 
For simplicity, we omit $u$ and $y$, and use the notation $Q^j$. The adjusted logit vector $\boldsymbol{\hat{o}}^j$ is computed according to $\boldsymbol{o}^j$ and $Q^j$ as specified in {Eq. (\ref{eq:adjusted_output})}. Then, the sampled softmax loss for layer $j$ is calculated as follows:
$$\ell^j_{sampled\text{-}softmax}(u,y,\mathcal{I}^{j}_{M})=-\log \frac{\exp {\hat{o}}^j_{i}}{\sum_{i'\in \mathcal{I}^{j}_{M}\cup\{i\}}\exp{\hat{o}}^j_{i'}}.$$ 
 The bias of the sampled softmax loss w.r.t. the full softmax loss is presented as follows:
\begin{Lemma}
\label{Lemma: bound of sampled softmax}
   For any instance \((u,y)\) and the \(j\)-th layer, the full softmax loss can be bounded by the sum of two terms: the expectation of the sampled softmax loss over \(\mathcal{I}^{j}_{M}\) and the KL divergence between \(Q^j\) w.r.t. \(P^j\). It is expressed as follows:
    \begin{small}
    \begin{equation*}
    \begin{split}
       \ell^j_{softmax}(u,y) \le  \mathbb{E}_{\mathcal{I}^{j}_{M}}\left[\ell^j_{sampled\text{-}softmax}(u,y,\mathcal{I}^{j}_{M})\right] + D_{\text{KL}}(Q^j \| P^j),
    \end{split}
    \end{equation*}
    \end{small}where  $P^j(u,y)$ is the softmax distribution in the layer $j$ for the instance $(u,y)$.
\end{Lemma}

\vspace{-0.8em}
\subsection{Probability Estimation Error}
As described in Section \ref{sec: practical implementation}, we use a small pre-trained model to estimate the conditional probability. To underpin a more comprehensive theoretical analysis, we now examine the error introduced by this probability estimation. For any instance $(u, y)$, let $\hat\eta_y$ and $\eta^*_y$ denote the estimated and true conditional probability, respectively. Since the model is well-pretrained,  we posit that its estimation error is bounded by a small constant: $|\hat\eta_y - \eta^*_y| < \epsilon$. Next, let $\hat{\bar{z}}^j_{\delta^j(y)}$ and ${\bar{z}}^{j*}_{\delta^j(y)}$ denote the rectified labels for node $n^j_{\delta^j(y)}$, computed using $\hat\eta_y$ and $\eta^*_y$, respectively. We now analyze the gap between the expected risks under these two label assignments: 
\begin{small}
\begin{equation*}
    \mathbb{E}_{(u,y)\sim \mathbb{P}}\left[\widetilde{\mathcal{L}}(u,y)\right]\quad\text{and}\quad\mathbb{E}_{(u,y)\sim \mathbb{P}}\left[\widetilde{\mathcal{L}}^*(u,y)\right],
\end{equation*}
\end{small}where $\widetilde{\mathcal{L}}$ uses $\hat{\bar{z}}^j_{\delta^j(y)}$ and $\widetilde{\mathcal{L}}^*$ uses ${\bar{z}}^{j*}_{\delta^j(y)}$. Their gap can be effectively controlled as follows:
\begin{Lemma} 
Suppose that the expected discrepancy in rectified labels across all layers is bounded by $g(\epsilon)$:
\begin{small}
\begin{equation*}\mathbb{E}_{(u,y)\sim \mathbb{P}} \left[\sum_{j=1}^H \left|  \hat{\bar{z}}^j_{\delta^j(y)} - {\bar{z}}^{j*}_{\delta^j(y)} \right|\right] \le g(\epsilon),\end{equation*}
\end{small}where $g$ is nondecreasing, and $g(\epsilon)=0$ if and only if $\epsilon=0$. The following inequality holds:
\begin{small}
\begin{equation*}
    \left|\mathbb{E}_{(u,y)\sim \mathbb{P}}[\widetilde{\mathcal{L}}(u,y)] 
- \mathbb{E}_{(u,y)\sim \mathbb{P}}[\widetilde{\mathcal{L}}^*(u,y)] \right|
\le g(\epsilon)\left(2B_{\mathcal{M}}+\log |\mathcal{Y}|\right).
\end{equation*}
\end{small}
\label{lemm:probability estimation error}
\end{Lemma}

\vspace{-1em}
\subsection{Generalization Bound of DTR}
By combining 
Lemma~\ref{Lemma: generalization bound of single layer} and Lemma~\ref{Lemma: bound of sampled softmax}, then aggregating from the 1st layer to the $H$-th layer, and further applying Lemma~\ref{lemm:probability estimation error}, we derive the following result:
\begin{Theorem}
\label{theo: final bound}
 The following inequality holds with a probability of at least $1-\delta$:
\begin{small}
\begin{equation}
\label{eq: final bound}
\begin{split}
    \mathbb{E}_{(u,y)\sim \mathbb{P}}\left[\widetilde{\mathcal{L}}^*(u,y)\right] &\le \frac{1}{m} \sum_{i=1}^{m}\sum_{j=1}^{H} \widehat{{\mathcal{L}}}_j(u_i,y_i)\\ 
    &+g(\epsilon)(2B_{\mathcal{M}}+\log|\mathcal{Y}|)\\
    &+\sum_{j=1}^H \mathbb{E}_{(u,y)\sim \mathbb{P}} \left[ D_{KL}(Q^j||P^j) \right] \\
    & + 8\tau^2\frac{c_{\phi}^{L-1}B_1^L(B_0+KB_{w_1}+KB_{w_2}\mathcal{\tau})}{\sqrt{m}}\\
    &+\Big(8\log  (\tau)+8HB_{\mathcal{M}}\Big)\sqrt{\frac{2\log\left(4/\delta\right)}{m}}.
\end{split}
\end{equation} 
\end{small}
\end{Theorem}Theorem \ref{theo: final bound} provides a comprehensive theoretical analysis of DTR.
As shown in Eq. \eqref{eq: final bound}, the left-hand side denotes the expected risk under the modified loss with rectified labels based on true conditional probabilities. Minimizing this risk leads to optimal performance under beam search (see Section \ref{sec: rectification loss}).
On the right-hand side, the first term is the empirical loss optimized in practice. 
The second term reflects the error from probability estimation, showing that more accurate estimates yields lower generalization risk.
The third term is the expected KL divergence of the negative sampling distribution w.r.t. the softmax distribution at each layer. 
It indicates that the closer the negative sampling distribution is to the softmax distribution, the lower the generalization risk. 
As shown in Corollary \ref{cor: bound of KL divergence}, our tree-based sampling can effectively control the KL divergence, providing insight into why it performs well from a generalization perspective.
The fourth term captures the influence of model parameters and tree structure. Notably, when the number of items (i.e., $|\mathcal{Y}|$) is fixed, increasing the number of branches $B$ results in a smaller $\tau$, leading to a tighter generalization bound. This observation is consistent with the findings in~\cite{zhang2024generalization}.
Moreover, when the sample size $m$ is sufficiently large, the last two terms tend to vanish, meaning that with enough data, the generalization risk can be effectively controlled by the empirical loss, the estimation error and the KL divergence.

\section{Experiments}
\label{sec: experiments}
We conduct a series of experiments designed to answer the following questions related to the joint optimization framework of preference models and tree indexes: \textbf{RQ-1}: \textit{Is the softmax-based multi-class classification training mode more suitable than the binary classification training mode for the tree model?} \textbf{RQ-2}: \textit{Is the modified loss with rectified labels effective in mitigating the suboptimality of the multi-class cross-entropy loss under beam search?} \textbf{RQ-3}: \textit{Is the tree-based sampling method effective  in mitigating approximation bias?} \textbf{RQ-4}: \textit{Can DTR improve the recommendation efficiency?} \textbf{RQ-5}: \textit{Can DTR learn a reasonable mapping between items and leaf nodes?} 

\subsection{Datasets}
The experiments are conducted on the four real-world datasets\footnote{\url{https://drive.google.com/drive/folders/1ahiLmzU7cGRPXf5qGMqtAChte2eYp9gI}}: Movie Lens 10M (\textbf{Movie}), MIND Small Dev (\textbf{MIND}), Amazon Books (\textbf{Amazon}), and Tmall Click (\textbf{Tmall}). {Movie} is a public film rating dataset, 
{MIND} is  a news recommendation dataset from Microsoft,
{Amazon} contains book purchase and rating data from Amazon platform,
{Tmall} records the shopping behaviors on Alibaba's Tmall marketplace.
We process all datasets into a form of implicit feedback, where interactions between users and items are marked as 1, and non-interactions as 0.
Users with fewer than 15 interactions are filtered out, with the post-filtering data statistics presented in {Table \ref{table:dataset_statistics}}. 
For each dataset, 10\% of users are randomly selected as validation users, 10\% as test users, and remaining users as training users. 
For training users, each interaction sequence is constrained to the most recent 70 interactions, with shorter sequences padded with zeros. The first 69 interactions serve as input, while the 70th is used as the label.
For validation and test users, the first half of their interaction history is used as input, with the second half serving as labels for prediction.

\begin{table}[htbp]
\caption{Statistics of datasets}
\label{table:dataset_statistics}
\centering
\begin{tabular}{ccccc}
\hline
Dataset & \#User & \#Item & \#Interaction & Density \\
\hline
Movie & 69,878 & 10,677 & 10,000,054 & 1.34\% \\
MIND & 36,281 & 7,129 & 5,610,960 & 2.16\% \\
Amazon & 29,980 & 67,402 & 2,218,926 & 0.11\% \\
Tmall & 139,234 & 135,293 & 10,487,585 & 0.05\% \\
\hline
\end{tabular}
\end{table}

\subsection{Baselines}
We compare the proposed DTR with related algorithms in two broad categories: 1) models that use brute-force search without indexing, and 2) models that rely on indexing for efficient retrieval. Note that brute-force methods require a significant amount of time for retrieval, so their results are provided for reference only.
\begin{itemize}[leftmargin=*]
    \item DIN~\cite{zhou2018deep} leverages an attention mechanism to capture user preferences from historical interactions. As it cannot directly enable efficient recommendation, we report its performance using brute-force search as a reference.
    \item YouTubeDNN~\cite{paul2016youtube} computes user-item preference scores via inner product, framing recommendation as a Maximum Inner Product Search (MIPS) over learned embeddings. Results are obtained through brute-force search.
    \item FMLP-Rec~\cite{zhou2022filter} employs stacked filter-enhanced blocks to generate user and item representations; we also report results obtained through brute-force search.
    \item SCANN~\cite{guo2020accelerating} is a competitive quantization-based approach that employs anisotropic vector quantization and score-aware loss to efficiently solve MIPS.
    \item IPNSW~\cite{morozov2018non} is a graph-based method that addresses MIPS efficiently by constructing non-metric similarity graphs.
    \item QINCo~\cite{huijben2024QINco} uses learned implicit codebooks for stepwise residual quantization, producing compact codes with accurate score reconstruction for efficient search.
    \item ROTLEX~\cite{10.1145/3711896.3737112} learns a K-ary tree with query-aware repartitioning; a query encoder routes to top buckets and beam search retrieves candidates efficiently.
    \item JTR~\cite{li2023constructing} optimizes the query encoding and tree index. It uses a flexible tree where each leaf contains multiple items, and applies overlapped clustering for assignment.
    \item PLT~\cite{jasinska2020probabilistic} constructs a label tree and computes item probabilities by multiplying conditional probabilities along root-to-leaf paths.
    \item TDM~\cite{zhu2018tdm} simultaneously learns the preference model and the tree index, performing efficient recommendations with beam search during inference.
    \item JTM~\cite{Zhu2019JointOO} improves TDM by jointly optimizing the index structure and the user preference model.
    \item OTM~\cite{zhuo2020learning} mitigates training-testing discrepancies of tree models under beam search by sampling negatives through beam search and correcting the pseudo-labels of nodes during training.
\end{itemize}

\subsection{Metric}
To evaluate the performance of the retrieval model, we use $Precision@K$, $Recall@K$, and $F\text{-}measure@K$ as the evaluation criteria. Suppose \( \mathcal{P}(u) \) (\( \left|\mathcal{P}(u)\right| = K \)) denotes the set of items retrieved by the model for user \( u \), and \( \mathcal{G}(u) \) represents the set of items that the user has interacted with, i.e., the ground truth associated with user \( u \). For user \( u \), the evaluation metrics $Precision@K$ and $Recall@K$ are calculated as follows:
\begin{small}
\begin{equation*}
\begin{split}
Precision@K(u)=\frac{|\mathcal{P}(u)\cap\mathcal{G}(u)|}K\ ,\\Recall@K(u)=\frac{|\mathcal{P}(u)\cap\mathcal{G}(u)|}{|\mathcal{G}(u)|}\ ,\ \ \
\end{split}
\end{equation*}    
\end{small}and $F\text{-}measure@K$ is calculated as follows:
\begin{small}
\begin{equation*}
    F\text{-}measure@K(u)=\frac{2\cdot Precision@K(u)\cdot Recall@K(u)}{Precision@K(u)+Recall@K(u)}.
\end{equation*}    
\end{small}

\subsection{Settings}
\label{subsec: settings}
\begin{table*}[t]
\centering
\caption{Performance comparison of different algorithms on four datasets under top-$K$ = 20 and 40. \\ \textbf{P}, \textbf{R}, and \textbf{F} denote Precision, Recall, and F-measure, respectively.}
\label{table: results}
\begin{tabular}{c|cccccc|cccccc}
\toprule
\textbf{Algorithm} & \textbf{P@20} & \textbf{R@20} & \textbf{F@20}
  & \textbf{P@40} & \textbf{R@40} & \textbf{F@40} & \textbf{P@20} & \textbf{R@20} & \textbf{F@20}
  & \textbf{P@40} & \textbf{R@40} & \textbf{F@40} \\ \midrule
  & \multicolumn{6}{c|}{Movie} & \multicolumn{6}{c}{MIND} \\
\midrule
\multicolumn{13}{c}{ } \\[-1.5em]
\midrule
DIN         & 0.2132 & 0.1142 & 0.1309 & 0.1957 & 0.2013 & 0.1706  & 0.4174 & 0.1912 & 0.2320 & \textbf{0.3485} & \textbf{0.2930} & \textbf{0.2756} \\
YouTubeDNN  & 0.2043 & 0.1113 & 0.1270 & 0.1874 & 0.1924 & 0.1630  & 0.4025 & 0.1850 & 0.2243 & 0.3357 & 0.2851 & 0.2671 \\
FMLP-Rec    & 0.2310 & 0.1273 & 0.1450 & 0.2022 & 0.2110 & 0.1778  & 0.4112 & 0.1902 & 0.2303 & 0.3363 & 0.2866 & 0.2684 \\ \midrule
SCANN       & 0.1902 & 0.1017 & 0.1164 & 0.1755 & 0.1782 & 0.1511  & 0.3264 & 0.1439 & 0.1766 & 0.2819 & 0.2342 & 0.2213 \\
IPNSW       & 0.2046 & 0.1114 & 0.1271 & 0.1877 & 0.1925 & 0.1631  & 0.4019 & 0.1839 & 0.2233 & 0.3343 & 0.2824 & 0.2652 \\
QINCo  & 0.2010 & 0.1096 & 0.1249 & 0.1841 & 0.1881 & 0.1595  & 0.3886 & 0.1780 & 0.2160 & 0.3277 & 0.2782 & 0.2606 \\
ROTLEX & 0.2143 & 0.1165 & 0.1334& 0.1842& 0.1875& 0.1598& 0.3985 &0.1903 &0.2297 & 0.3178 & 0.2787 & 0.2601 \\
JTR         & 0.2285 & 0.1251 & 0.1428 & 0.1961 & 0.2015 & 0.1708  & 0.3855 & 0.1740 & 0.2117 & 0.3204 & 0.2693 & 0.2529 \\
PLT         & 0.1606 & 0.0859 & 0.0981 & 0.1478 & 0.1489 & 0.1264  & 0.2341 & 0.0927 & 0.1168 & 0.2024 & 0.1532 & 0.1494 \\
TDM         & 0.1901 & 0.1012 & 0.1159 & 0.1741 & 0.1770 & 0.1499  & 0.3156 & 0.1365 & 0.1687 & 0.2706 & 0.2236 & 0.2116 \\
JTM         & 0.2141 & 0.1169 & 0.1335 & 0.1922 & 0.1993 & 0.1682  & 0.3457 & 0.1547 & 0.1892 & 0.2806 & 0.2354 & 0.2213 \\
OTM         & 0.2321 & 0.1290 & 0.1470 & 0.2009 & 0.2086 & 0.1767  & 0.3724 & 0.1609 & 0.1978 & 0.3144 & 0.2523 & 0.2404 \\ \midrule
DTR(U)      & 0.2367 & 0.1231 & 0.1431 & 0.2054 & 0.2032 & 0.1759  & 0.3904 & 0.1804 & 0.2188 & 0.3063 & 0.2621 & 0.2448 \\
DTR(T)      & 0.2511 & 0.1343 & 0.1546 & 0.2135 & 0.2169 & 0.1853  & 0.4044 & 0.1853 & 0.2250 & 0.3318 & 0.2808 & 0.2634 \\
DTR(T-RL)   & \textbf{0.2545} & \textbf{0.1374} & \textbf{0.1580} & \textbf{0.2176} & \textbf{0.2240} & \textbf{0.1905} 
            & \textbf{0.4210} & \textbf{0.1938} & \textbf{0.2350} & \textbf{0.3401} & \textbf{0.2869} & \textbf{0.2695} \\
\midrule
{}  & \multicolumn{6}{c|}{Amazon} & \multicolumn{6}{c}{Tmall} \\
\midrule
\multicolumn{13}{c}{ } \\[-1.5em]
\midrule
DIN        & 0.0706 & 0.0495 & 0.0528 & 0.0605 & 0.0828 & 0.0625  & \textbf{0.0488} & \textbf{0.0350} & \textbf{0.0376} & \textbf{0.0379} & \textbf{0.0531} & \textbf{0.0405} \\
YouTubeDNN & 0.0606 & 0.0399 & 0.0436 & 0.0522 & 0.0682 & 0.0527  & 0.0331 & 0.0211 & 0.0236 & 0.0276 & 0.0349 & 0.0278 \\
FMLP-Rec   & 0.0711 & 0.0500 & 0.0533 & 0.0601 & 0.0825 & 0.0622  & \textbf{0.0485} & \textbf{0.0345} & \textbf{0.0372} & 0.0366 & 0.0509 & 0.0390 \\ \midrule
SCANN      & 0.0543 & 0.0365 & 0.0395 & 0.0462 & 0.0613 & 0.0470  & 0.0300 & 0.0193 & 0.0215 & 0.0252 & 0.0320 & 0.0255 \\
IPNSW      & 0.0587 & 0.0396 & 0.0427 & 0.0503 & 0.0664 & 0.0511  & 0.0316 & 0.0199 & 0.0223 & 0.0252 & 0.0311 & 0.0251 \\
QINCo & 0.0562 &  0.0380 & 0.0411 & 0.0492 & 0.0653 & 0.0502  & 0.0317 & 0.0202 & 0.0225 & 0.0265 & 0.0332 & 0.0266 \\
ROTLEX &0.0633 &0.0438 &0.0471 & 0.0497 &0.0669 & 0.0510 &0.0404 &0.0288 & 0.0311 & 0.0291 &0.0405 &0.0309\\
JTR        & 0.0590 & 0.0399 & 0.0432 & 0.0502 & 0.0672 & 0.0515  & 0.0474 & 0.0340 & 0.0366 & 0.0352 & 0.0494 & 0.0377 \\
PLT        & 0.0238 & 0.0147 & 0.0162 & 0.0201 & 0.0245 & 0.0193  & 0.0167 & 0.0111 & 0.0122 & 0.0133 & 0.0176 & 0.0137 \\
TDM        & 0.0514 & 0.0333 & 0.0365 & 0.0419 & 0.0530 & 0.0416  & 0.0310 & 0.0209 & 0.0229 & 0.0253 & 0.0335 & 0.0262 \\
JTM        & 0.0606 & 0.0411 & 0.0444 & 0.0507 & 0.0674 & 0.0517  & 0.0394 & 0.0275 & 0.0298 & 0.0315 & 0.0427 & 0.0331 \\
OTM        & 0.0704 & 0.0440 & 0.0486 & 0.0584 & 0.0710 & 0.0566  & 0.0459 & 0.0283 & 0.0318 & 0.0361 & 0.0433 & 0.0353 \\ \midrule
DTR(U)     & 0.0709 & 0.0502 & 0.0535 & 0.0594 & 0.0820 & 0.0619  & 0.0442 & 0.0312 & 0.0337 & 0.0356 & 0.0489 & 0.0377 \\
DTR(T)     & 0.0757 & 0.0517 & 0.0557 & 0.0613 & 0.0829 & 0.0631  & 0.0464 & 0.0329 & 0.0355 & 0.0364 & 0.0503 & 0.0385 \\
DTR(T-RL)  & \textbf{0.0777} & \textbf{0.0542} & \textbf{0.0580} & \textbf{0.0626} & \textbf{0.0847} & \textbf{0.0644}
           & \textbf{0.0482} & \textbf{0.0342} & \textbf{0.0369} & \textbf{0.0373} & \textbf{0.0514} & \textbf{0.0395}
\\
\bottomrule
\end{tabular}
\end{table*}

We provide detailed settings of the baseline algorithm and our proposed DTR. 

DIN comprises an embedding layer followed by an MLP with layer sizes [128,64,2], which outputs the probability of a user liking a given item.
YouTubeDNN also includes an embedding layer and an MLP with layer sizes [128,64,24].
FMLP-Rec includes an embedding layer, followed by 4 filter layers, and an MLP with layer sizes [24,128,24].

For the two-stage methods (ScaNN, IPNSW, QINCo, and ROTLEX), which build an index over fixed item embeddings,  we use item embeddings from the pretrained DIN for a fair comparison. The query vector for each user is computed by averaging the item embeddings within their interaction sequence. For ScaNN, we follow the default settings\footnote{\href{https://github.com/google-research/google-research/tree/master/scann}{google-research/scann}}: the number of leaves $K_{vq}$ is set to 2000, the number of sub-spaces is set to 6, each with a dimensionality of 4, and the number of codewords per sub-space $K_{pq}$ is 16; the threshold is set to 0.2. To construct the IPNSW graph index, each node’s maximum degree is set to 16, and the search width is set to 100. For QINCo\footnote{\href{https://github.com/facebookresearch/Qinco/tree/main/qinco_v1}{facebookresearch/Qinco/qinco\_v1}}, we use 8 subquantizers, 4 residual blocks, a codebook size of 256, and a hidden dimension of 256 for the residual blocks. For ROTLEX\footnote{\href{https://github.com/USTCLLM/rotlex}{https://github.com/USTCLLM/rotlex}}, we use the official defaults (MLP width $W=512$, branching factor $k=32$, two layers) and convert each training-set user sequence into a query vector to accommodate its query-aware training design.

JTR consists of a query encoder and a tree  where each leaf node contains multiple items. Following the official settings\footnote{\url{https://github.com/CSHaitao/JTR}} and adapting to the sequential recommendation scenario, we use a pre-trained SASRec as the query encoder. The number of branches is set to 10, with the maximum items per leaf node $\gamma$ set to 100. The tree is initialized by k-means clustering on item embeddings from SASRec and updated using the overlapped clustering technique. 

For tree-based algorithms except JTR, a binary tree is adopted, with DIN or its variants serving as the preference model. The user’s 69 historical interactions are divided into 10 sliding windows, with window sizes of $\left[20, 20, 10, 10, 2, 2, 2, 1, 1, 1\right]$. Except for OTM, these algorithms initialize the tree based on item categories; OTM adopts a tree learned from JTM.  TDM, JTM, and DTR include a process of updating the tree. TDM is configured to include 4 iterations of tree updating, while JTM and DTR are set to include 12 iterations. 

DTR employs two sampling methods: uniform sampling and tree-based sampling, denoted as DTR(U) and DTR(T), respectively. DTR(T) further incorporates the modified loss with rectified labels, referred to as DTR(T-RL). For DTR(T-RL), we utilize the SASRec~\cite{kang2018self} to estimate the conditional probability. For a given user $u$, SASRec encodes his historical sequence into a semantic vector, which is then combined with item embeddings via inner product and softmax to yield conditional probabilities. SASRec is trained on the same trainset as other algorithms. 
When evaluated on the testset, the trained SASRec performs at most 0.66\% lower than DIN in terms of $F\text{-}measure@20$ across four datasets, indicating it is a well-behaved estimator. As SASRec is a two-tower model, the additional computation cost is relatively minimal. In our experiments, the time spent on estimating probabilities accounts for approximately 19.32\% to 23.96\% of the total time across four datasets, which is considered acceptable.

In all algorithms, the dimension of the item embedding is set to 24, and the batch size is set to 100. We employ Adam as the optimizer \cite{kingma2014adam}, with a learning rate initialized at $1.0e$-$3$ and subjected to exponential decay. For SCANN, the number of probed leaves $W$ is set to 1000. For IPNSW and tree-based algorithms, the beam size $k$ is set to 150.
The number of negative samples is set to 70 and 1,000 for DIN and YouTubeDNN, respectively. 
Negative nodes are the sibling nodes of the positive node in PLT, while for other tree-based algorithms, the negative nodes are sampled from each layer.
For TDM and JTM, the number of negative samples per layer is set to 6, 3, 5, and 6 for the {Movie}, {MIND}, {Amazon}, and {Tmall} datasets, respectively. In OTM, since sampled nodes are the candidates generated through beam search, the number of negative samples is set to 300. For JTR and DTR, the number of negative samples is set to 50 and 70, respectively.
All experiments are conducted on a Linux server equipped with a $3.00$ GHz Intel CPU, $300$ GB of main memory, and NVIDIA $2080$  GPUs. The source code of DTR is available at \url{https://github.com/marsh312/DTR}.


\subsection{Comparison with baselines}
The experiment results are shown in {Table \ref{table: results}}, where the best performance of index-based algorithms is highlighted in bold. Additionally, any results from brute-force search algorithms that exceed the best performance of index-based algorithms are also emphasized in bold for reference. Based on the results, we can make the following findings:
\begin{itemize}[leftmargin=*]
    \item The proposed DTR algorithm consistently outperforms both TDM and JTM across all four datasets. When compared to their enhanced version, OTM, DTR remains highly effective: DTR(U) outperforms OTM in most cases, while DTR(T) and DTR(T-RL) consistently surpass OTM. For example, considering $F\text{-}measure@20$, DTR(T-RL) shows improvements of 7.48\%, 18.81\%, 19.34\%, and 16.03\% across the four datasets compared to OTM. This result demonstrates that the softmax-based multi-class classification training mode is better suited for jointly optimizing the preference model and tree structure index than the binary classification training mode of TDM, JTM, and OTM, effectively answering \textbf{RQ-1}. 

    \item The superiority of DTR(T-RL) over DTR(U) is consistently observed across all datasets. For example, in terms of $F\text{-}measure@20$, DTR(T-RL) achieves improvements of 2.20\%, 4.44\%, 4.13\%, and 3.94\%, respectively, across the four datasets. This outcome is in line with the theoretical insights discussed in Sections~\ref{sec: suboptimality of ce loss} and~\ref{sec: rectification loss}. It confirms that the modified loss function with rectified labels effectively mitigates the suboptimality of the multi-class cross-entropy loss when applied under beam search conditions, thus providing an effective answer to \textbf{RQ-2}.

    \item The performance of DTR(T) consistently surpasses that of DTR(U). For instance, when the top-K value is set to 20, DTR(T) shows an increase in $F\text{-}measure$ by 8.04\%, 2.83\%, 4.11\%, and 4.41\% across the four datasets, respectively.  These improvements confirms that our tree-based sampling method better aligns with the requirements of sampled softmax theory, as its sampling probabilities are more closely aligned with the actual softmax probabilities. This result validates the theoretical analysis in Section~\ref{sec: theoretical support of tree-based sampling } (i.e., Theorem~\ref{theo: sampling theory to unbiased softmax}), and effectively addresses \textbf{RQ-3}.
    
    \item  DTR(T-RL) also demonstrates strong performance compared to brute-force search algorithms. It outperforms YouTubeDNN across all datasets, and surpasses DIN and FMLP-Rec on most of them. These results suggest that integrating the index with the learning of the preference model can yield better retrieval than optimizing the model alone. Moreover, the tree index allows the model to exclude some distracting items in certain scenarios, thereby enhancing retrieval performance.

    \item Compared with other index-based retrieval algorithms—including IPNSW (graph-based), SCANN and QINCo (quantization-based), ROTLEX (query-aware tree-based), and JTR (tree-based with a flexible tree and separate query encoder)—DTR(T-RL) consistently achieves superior performance across all four datasets. These results underscore the advantage of our joint optimization strategy, which tightly couples the tree index and the preference model, enabling more accurate and efficient retrieval.
    
\end{itemize}

\subsection{Comparison of Inference Time}
Table~\ref{tab:inference_time} reports the theoretical time complexities of several representative algorithms together with their running times on the testsets of four datasets. The algorithms compared are: the brute force-based DIN, the graph-based IPNSW, the quantization-based SCANN, and two tree‑based methods, JTR and DTR. Note that in JTR each leaf node corresponds to multiple items, so after retrieving the leaf nodes an additional re‑ranking stage is required; in contrast, in DTR each leaf node maps one‑to‑one to an item, meaning that once the leaf node is retrieved, the item retrieval is complete. The table shows that DTR achieves the lowest inference time on three datasets, and even on Tmall it requires only about an extra 21.3\% of the inference time compared to the best SCANN -- far outperforming the other methods. This demonstrates that DTR can indeed improve retrieval efficiency, thus answering \textbf{RQ‑4}.
\begin{table}[htbp]
\caption{Complexity and Inference Time (in seconds) of Representative Algorithms}
\label{tab:inference_time}
\renewcommand{\arraystretch}{1.2}
\centering
\resizebox{0.49\textwidth}{!}{
\begin{tabular}{c|c|cccc}
\hline
Algorithm & Complexity & Movie & MIND & Amazon & Tmall \\ \hline
DIN & $O(|\mathcal{Y}|)$ & 81.32 & 38.93 & 217.21 & 2104.05 \\
IPNSW & $O(\log |\mathcal{Y}|)$ & 16.84 & 10.98 & 11.15 & 53.45 \\
SCANN & $O(K_{vq} + K_{pq} + \frac{W |\mathcal{Y}|}{K_{vq}})$ & 20.39 & 18.64 & 4.13 & \textbf{15.49} \\
JTR & $O(k \log \frac{|\mathcal{Y}|}{\gamma} + k\gamma)$ & 11.95 & 9.24 & 15.87 & 82.79 \\
DTR & $O(k\log |\mathcal{Y}|)$ & \textbf{5.58} & \textbf{3.40} & \textbf{3.53} & 18.79 \\
\hline
\end{tabular}
}
\end{table}

\subsection{Effectiveness of Tree Index Updating}
To demonstrate the effectiveness of DTR’s tree-index updating (Section~\ref{sec: tree index updating}), we compare three DTR variants against structurally similar, update-capable tree-based algorithms. As PLT and OTM don't involve tree updating and JTM enhances TDM with a preference-aligned tree update strategy, we focus on JTM versus the DTR variants.
\begin{figure}[htbp]
  \centering \subfloat{\includegraphics[width=0.8\linewidth]{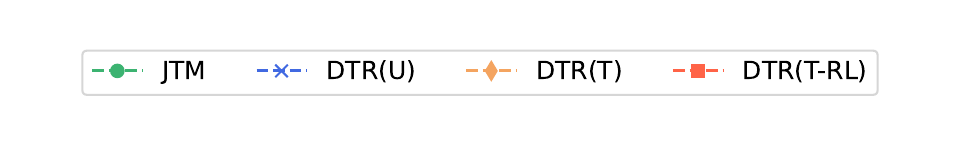}}\\
  \vspace{-0.5cm}
  \setcounter{subfigure}{0} 
  \subfloat[Movie]{\includegraphics[width=0.243\textwidth]{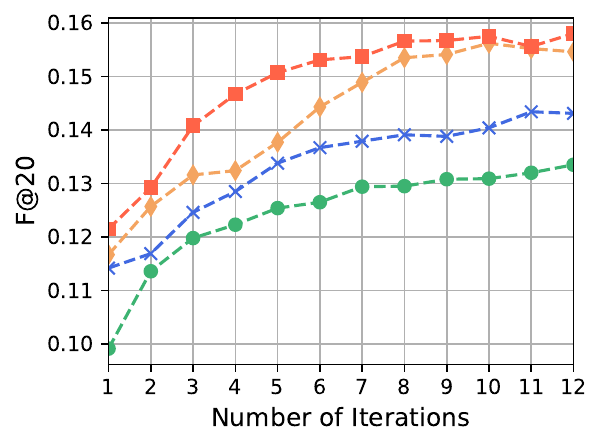}}
  \subfloat[MIND]{\includegraphics[width=0.243\textwidth]{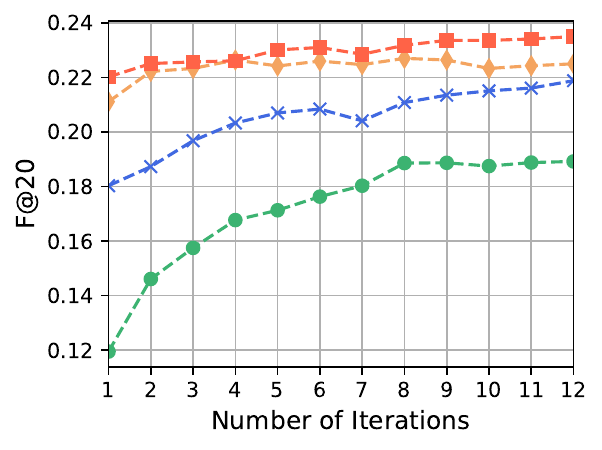}}\\
  \vspace{-0.25cm}
  \subfloat[Amazon]{\includegraphics[width=0.243\textwidth]{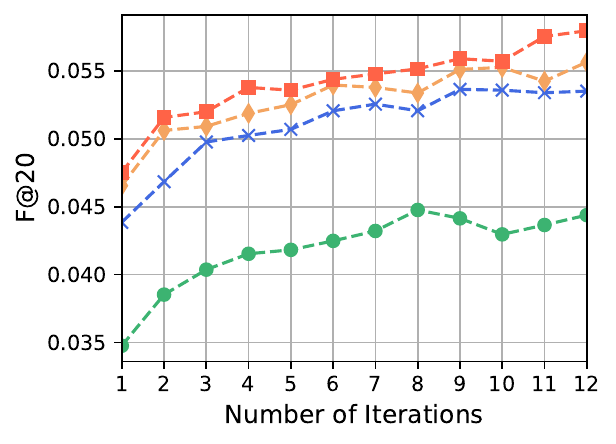}}
  \subfloat[Tmall]{\includegraphics[width=0.243\textwidth]{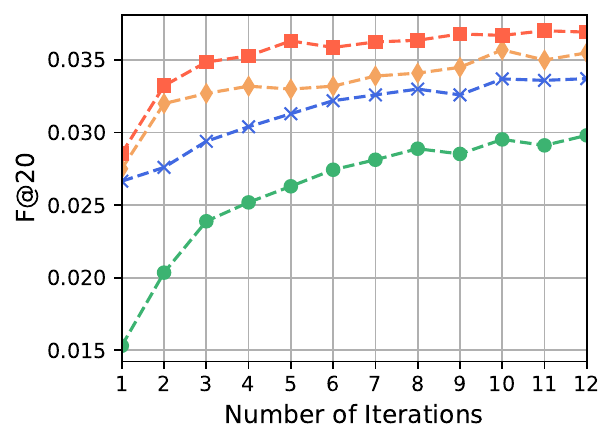}}\\
  \vspace{-0.1cm}
  \caption{The variation of $F\text{-}measure@20$ with tree updating for JTM and tree DTR variants across four datasets.}
  \label{fig:tree_update_fig}
\end{figure}

We show how $F\text{-}measure@20$ changes with tree updates across four datasets, as presented in {Figure \ref{fig:tree_update_fig}}. From these results, the following findings can be made:
\begin{itemize}[leftmargin=*]
    \item The three variants of DTR consistently show an initial increase followed by a gradual convergence across four different datasets. This phenomenon suggests that the tree update strategy utilized by DTR is effective in establishing a reasonable mapping between items and leaf nodes. Consequently, the index structure evolves alongside the preference model, leading to improved retrieval performance. This result convincingly addresses \textbf{RQ-5}.
    \item The three variants of DTR consistently exhibit superior performance, significantly outperforming JTM across four datasets. This result further illustrates the superiority of the softmax-based multi-class classification training mode employed by DTR over the binary classification training mode used by JTM, thereby answering \textbf{RQ-1} once again.
    \item The advantage of DTR(T-RL) over DTR(T) is also evident throughout the tree updating process. This further demonstrates that the modified loss function, incorporating the rectified label, can effectively enhance the performance of the tree model, aligning with our theoretical analysis and also answering \textbf{RQ-2}.
    \item The variants of DTR based on tree-based sampling consistently outperforms DTR(U). This result indicates that tree-based sampling can estimate the softmax gradient with greater accuracy compared to uniform sampling. It validates the proposed sampling method's rationality and echoes the theoretical conclusions underpinning tree-based sampling, also addressing \textbf{RQ-3}.
\end{itemize}

\subsection{Sensitivity Analysis}
We analyze the impact of the number of negative samples, the number of tree branches, and the embedding dimension on the performance of three DTR variants on the MIND and Movie datasets. Unless otherwise specified, all experiments in the following subsections adopt a default configuration: the tree is initialized based on item categories, the number of negative samples is set to 70, the number of tree branches to 2, and the embedding dimension to 24.

\subsubsection{Sensitivity w.r.t. Number of Negative Samples}

We investigate how the number of negative samples affects performance by varying it from $10$ to $90$ in increments of $10$. Results are shown in Figure~\ref{fig:vary_sampling_num}. 
The results show that DTR(U) is competitive when the number of negatives is relatively small. 
However, as the negative sample size becomes relatively large, the two tree-based sampling variants consistently outperform DTR(U). It indicates that while uniform sampling is effective, tree-based sampling becomes significantly advantageous with larger sample sizes.
Moreover, DTR(T-RL) outperforms DTR(T), which validates the effectiveness of the modified loss with rectified labels again. As the number of negative samples increases, the $F$-$measure@20$ of three DTR variants all show a gradual improvement and eventually converge. When the number of negative samples reaches $50$ or more, the performance of all three DTR variants stabilizes, indicating that DTR's robustness to  the number of negative samples. 
\begin{figure}[htbp]
  \centering \subfloat{\includegraphics[width=0.8\linewidth]{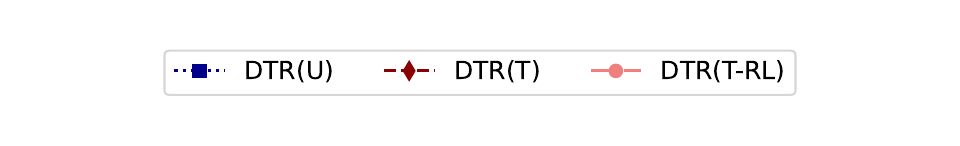}}\\
  \vspace{-0.6cm}
  \setcounter{subfigure}{0} 
  \subfloat[Movie]{\includegraphics[width=0.243\textwidth]{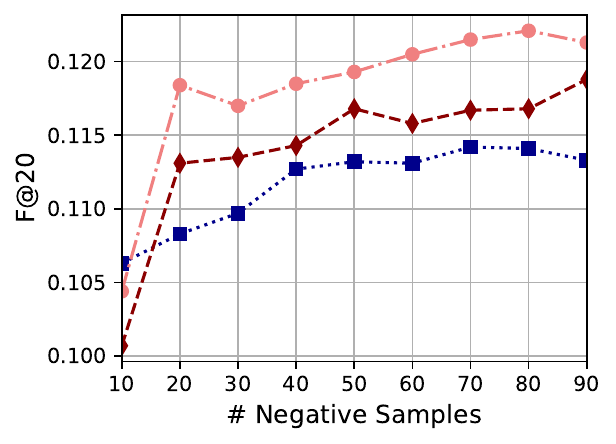}}
  \subfloat[MIND]{\includegraphics[width=0.243\textwidth]{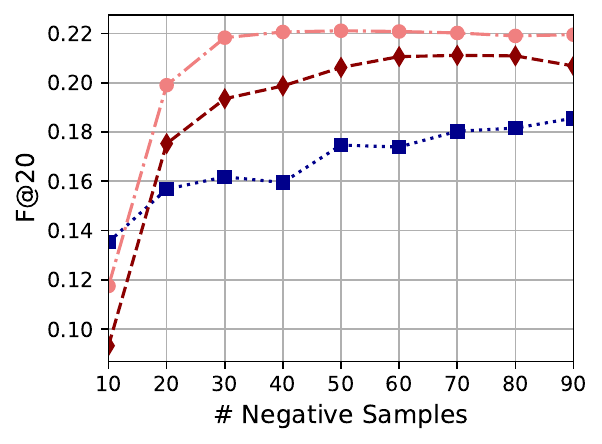}}
  \vspace{-0.1cm}
  \caption{The $F\text{-}measure@20$ of DTR with varying numbers of negative samples on MIND and Movie.}
  \label{fig:vary_sampling_num}
\end{figure}

\subsubsection{Sensitivity w.r.t Branch Number}
We explore the impact of the number of tree branches on the performance of DTR by varying the branch number from 2 to 10 in increments of 2. To ensure a fair comparison, all trees are randomly initialized, and the number of negative samples is adaptively adjusted to maintain a consistent total across the entire tree for different branch settings. The results are shown in Figure~\ref{fig:vary_branch_num}. The results indicate that the $F$-$measure@20$ of all three DTR variants generally improves as the number of tree branches increases. Among them, DTR(U) demonstrates the most substantial performance gain with more branches. These findings suggest that increasing the branch number can effectively enhance DTR's performance, which aligns with the theoretical insight presented in Theorem~\ref{theo: final bound}.

\begin{figure}[htbp]
  \centering \subfloat{\includegraphics[width=0.8\linewidth]{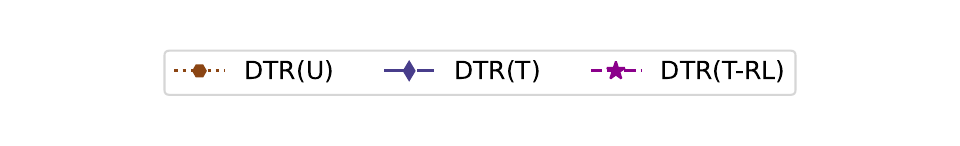}}\\
  \vspace{-0.6cm}
  \setcounter{subfigure}{0} 
  \subfloat[Movie]{\includegraphics[width=0.243\textwidth]{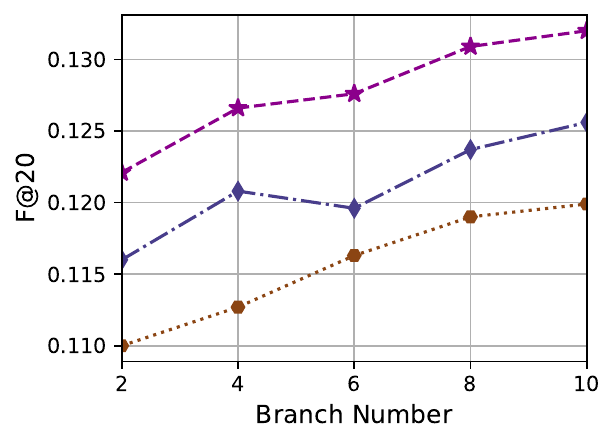}}
  \subfloat[MIND]{\includegraphics[width=0.243\textwidth]{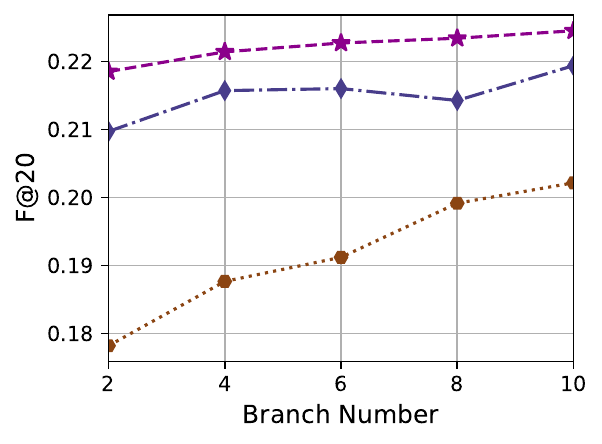}}
  \vspace{-0.1cm}
  \caption{The $F\text{-}measure@20$ of DTR with varying numbers of tree branch on MIND and Movie.}
  \label{fig:vary_branch_num}
\end{figure}

\subsubsection{Sensitivity w.r.t Embedding Dimension}
We examine the impact of the embedding dimension by varying it from 12 to 32 in increments of 4. The results are presented in Figure~\ref{fig:vary_embedding_size}. As illustrated, the performance of DTR is relatively low when the embedding dimension is small, but improves as the dimension increases. After reaching 24, the performance tends to stabilize, with further increases bringing only marginal gains and even causing slight performance degradation.
\begin{figure}[htbp]
  \centering \subfloat{\includegraphics[width=0.8\linewidth]{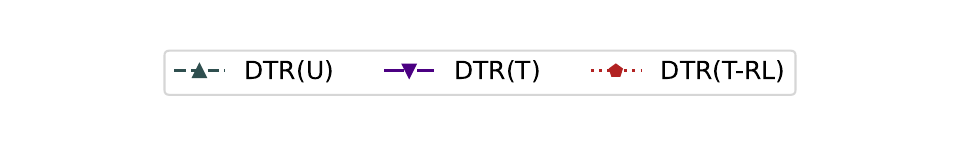}}\\
  \vspace{-0.6cm}
  \setcounter{subfigure}{0} 
  \subfloat[Movie]{\includegraphics[width=0.243\textwidth]{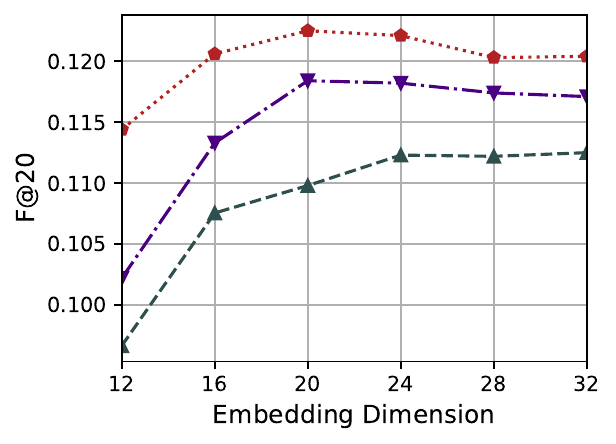}}
  \subfloat[MIND]{\includegraphics[width=0.243\textwidth]{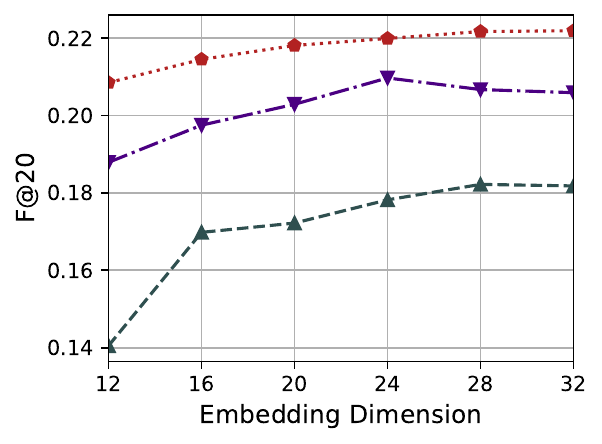}}
  \\
  \vspace{-0.1cm}
  \caption{The $F\text{-}measure@20$ of DTR with varying embedding dimension on
   MIND and Movie.}
  \label{fig:vary_embedding_size}
\end{figure}

\section{Conclusion}
\label{sec: conclusion}
The existing tree-based recommendation systems treat the preference model training task as independent binary classification problems for each node. Their training mode suffers from the gap between training and prediction due to the lack ofhttps://www.overleaf.com/project competition among nodes at the same layer. To address this problem, we develop a layer-wise training mode that frames the training task as a softmax-based multi-class classification problem and propose a tree learning method compatible with this training mode. Our theoretical analysis reveals the suboptimality of this training mode under beam search, promoting us to propose a rectification method. Given the significant time cost of calculating softmax probabilities with large-scale datasets, we develop a tree-based sampling method guided by the sampled softmax theory to estimate the loss gradient accurately, thus accelerating the training process.  Finally, we provide a generalization analysis of the deep tree-based retriever, demonstrating its great generalization capability. Experimental results validate the effectiveness of our proposed methods.

\section*{Acknowledgments}
The work was supported by grants from the National Natural Science Foundation of China (No. 62022077).

\bibliographystyle{IEEEtran}
\bibliography{reference}

\clearpage
\newpage
\begin{appendices}
\section{Notations}
\label{sec: notations}
Some important and common notations in this paper are summarized in Table~\ref{tab: notations}.
\vspace{-0.32cm}
\begin{table}[ht]
    \centering
    \caption{Some important and common notations in this paper.}
    \vspace{-0.3cm}
    \scalebox{0.83}{
    \begin{tabular}{c|c}
    \toprule
        Notations & Descriptions\\ \midrule
        $\mathcal{U},\mathcal{Y}$ & The user space and the item set.\\ 
        \midrule
        $\mathcal{T}, \mathcal{M}$ & The tree and the preference model.\\
        \midrule
        $\mathcal{N}, \mathcal{N}^j$ & The tree node set and the set of $j$-th level's nodes.\\
        \midrule
        $B,H$ & The branch number of $\mathcal{T}$ and the tree height.\\
        \midrule
        $K$ & The length of historical interaction sequence.\\
        \midrule
        $L$ & The number of hidden layers in the neural network.\\
        \midrule
        $\eta_y$ & The conditional probability of the item $y$.\\
        \midrule
        $n, n_i^j$ & The tree node and the $i$-th node in the $j$-th level. \\
        \midrule
        $N_j$ & The number of nodes in level $j$.\\
        \midrule
        $N$ & The number of classes in muti-class classification.\\ \midrule
        $\rho^l(n)$ & The ancestor of node $n$ in level $l$.\\
        \midrule
        $\pi$ & The bijective mapping maps leaf nodes to items.\\
        \midrule
        $\delta^j(y)$ & The index of $j$-th ancestor of $\pi^{-1}(y)$ within its level.\\
        \midrule
        $\mathcal{C}(n)$ &  The set of node $n$'s child nodes.\\
        \midrule
        $\mathfrak{L}(n)$ &The set of leaf nodes within the subtree rooted at $n$.\\
        \midrule
        $o_i^j,\hat{o}_i^j$ & The preference score and adjusted preference score for node $n_i^j$.\\
        \midrule
        $\mathcal{B}^j$ & The index set of selected nodes in layer $j$.\\
        \midrule
        $z_i^j, \bar{z}_i^j, \tilde{z}_i^j$ & The label, rectified label, and normalized rectified label for node $n_i^j$. \\
        \midrule
        $\mathcal{L}_j$ & The multi-class cross-entropy loss function of layer $j$. \\
        \midrule$\widetilde{\mathcal{L}}_j$ & The modified loss function with rectified labels of layer $j$.\\
        \midrule
        $\widehat{\mathcal{L}}_j$ & The loss function $\widetilde{\mathcal{L}}_j$ incorporating sampled softmax.\\
        \midrule
        $M$ & The number of negative samples in sampled softmax.\\ \midrule
        $\mathcal{I}_M^j$& The set of indices of $M$ negative nodes at layer $j$.\\  \midrule
        $\mathcal{R}_j, \widetilde{\mathcal{R}}_j$&The conditional risk of loss $\mathcal{L}_j$ and $\widetilde{\mathcal{L}}_j$.\\
        \midrule$\widehat{\mathcal{R}}_m$&The empirical Rademacher complexity for a sample size $m$.\\
        \midrule
        $Q^j, P^j$ & The sampling distribution and softmax distribution in layer $j$.\\
        \midrule
        $q_i^j$ & The negative sampling probability for node $n_i^j$.\\
        \bottomrule
        $p_i^j$ & The softmax probability for node $n_i^j$ based on the preference scores.\\
        \bottomrule 
    \end{tabular}
    }
    \label{tab: notations}
\end{table}

\section{Overview of the DTR Learning Process}
\label{appendix: overview of learning process of DTR}
A pseudocode overview of the DTR learning process is given in Algorithm~\ref{alg: learning process of DTR}.

\section{Proofs of Lemma, proposition, and Theorem}
\subsection{Proof of Proposition \ref{pro:suboptimality of multi-class cross-entropy loss}}
\label{appendix: proof of proposition suboptimality of multi-class cross-entropy loss}
Our proof relies on the fact that the multi-class cross-entropy loss (i.e., Eq. \eqref{eq:multi_cross_entropy_single_layer}) is actually a Bregman divergence: by taking $\psi(\boldsymbol{s})=\sum_{n=1}^N {s}_n\log {s}_n$ and $g(\boldsymbol{o})_n={\exp{o}_n}/{\sum_{k=1}^N\exp{o}_k}$, the multi-class cross-entropy loss  
\begin{equation*}
    \mathcal{L}(\boldsymbol{e}^{(i)},\boldsymbol{o})=-\log\frac{\exp{o}_i}{\sum_{n=1}^N\exp{o}_n}
\end{equation*}    
can be expressed as Bregman divergence:
\begin{equation*}
    \mathcal{L}(\boldsymbol{e}^{(i)},\boldsymbol{o}) = D_{\psi}(\boldsymbol{e}^{(i)},g(\boldsymbol{o})).
\end{equation*}
Here,  $\boldsymbol{o}, \boldsymbol{e}^{(i)} \in \mathbb{R}^N$, and $\boldsymbol{e}^{(i)}$ is a one-hot vector where the $i$-th component is $1$ and other components are $0$.
\begin{proof}
Since $\mathcal{L}_j(u,y)$ is a Bregman divergence and the softmax function $g$ satisfies $R(\boldsymbol{s},g(\boldsymbol{s}))$ for any $\boldsymbol{s}\in \mathbb{R}^N$, 
according to  Theorem \ref{theo: property of minimum of Bregman divergence}, for a given user $u$, we have
\begin{equation*}
    \operatornamewithlimits{argmin}\limits_{\boldsymbol{o}^j\in \mathbb{R}^{N_j}} \mathcal{R}_j(u) \subseteq  \left\{\boldsymbol{o}^j\in \mathbb{R}^{N_j}\mid R(\boldsymbol{o}^j,\mathbb{E}[\boldsymbol{z}^j|u])\right\}
\end{equation*}
holds for $1\le j\le H$. That is, optimization aligns the ranking of node scores in each layer with the expectation of the nodes' label. As a result, during beam search, the top-$k$ nodes selected at each layer are determined by the expected values $\mathbb{E}[z^j_i|u]$, which results in suboptimal results under beam search. 

We give a non-optimal example. Considering the item set $\mathcal{Y}$ with $|\mathcal{Y}|=8$, we directly map $i$-th item to $i$-th leaf node, and construct the corresponding binary tree with height $H=3$.   Given the user $u_0$, the conditional probability vector is $\boldsymbol{\eta}(u_0)=(0.21,\ 0,\ 0.12,\ 0.18,\ 0.19,\ 0,\ 0.16,\ 0.14)$, where the $i$-th component is the conditional probability of $i$-th item. According to Eq. \eqref{eq: expectation of layer random vector}, the expectation vector is $(0.21, 0.3, 0.19, 0.3)$ and $(0.51, 0.49)$ for $2$nd layer and $1$st layer, respectively.
    
Let beam size $k=3$. Beam search expands nodes according to top-$3$ of $\mathbb{E}[z^j_i|u]$ at each layer. Therefore, the expanded nodes in $2$nd layer are $\{1^{(2)},\ 2^{(2)},\ 4^{(2)}\}$, and in $3$rd layer are $\{1^{(3)},\ 4^{(3)},\ 7^{(3)}\}$, which is the final retrieved result. However, the optimal result is $\{1^{(3)},\ 4^{(3)},\ 5^{(3)}\}$ under beam search with beam size $k=3$.  
\end{proof}

\subsection{Proof of Proposition \ref{pro: modified softmax Bregman divergence}}
\label{appendix: proof of proposition modified softmax Bregman divergence}
\begin{proof}
     By taking $\psi(\boldsymbol{x})=\sum_{n=1}^N {x}_n\log {x}_n$ and $g(\boldsymbol{o})_n=\exp {o}_n/ \sum_{k=1}^N\exp {o}_k$, we have
     \begin{small}
   \begin{align*}
    D_{\psi}&\left(\boldsymbol{e}^{(i)}(w),g(\boldsymbol{o})\right)\\
    \displaybreak[1]
    &= \psi\left(\boldsymbol{e}^{(i)}(w)\right)-\psi\left(g(\boldsymbol{o})\right)-\nabla\psi(g(\boldsymbol{o}))^T\left(\boldsymbol{e}^{(i)}(w)-g(\boldsymbol{o})\right)\\
    &= w\log w-\sum_{n=1}^N g(\boldsymbol{o})_n\log g(\boldsymbol{o})_n \\
    &\quad- \Big< \left(\log g(\boldsymbol{o})_1+1 ,\dots,\log g(\boldsymbol{o})_i + 1, \dots,\log g(\boldsymbol{o})_n + 1\right), \\
    &\quad\quad\quad\left(-g(\boldsymbol{o})_1,\dots,w-g(\boldsymbol{o})_i,\dots,-g(\boldsymbol{o})_n\right)\Big>\\ 
    \displaybreak[1]
    &=w\log w -\sum_{n=1}^N g(\boldsymbol{o})_n\log g(\boldsymbol{o})_n \\
    &\quad- (-\sum_{n=1}^N g(\boldsymbol{o})_n\log g(\boldsymbol{o})_n -1+w\log g(\boldsymbol{o})_i+w) \\
    &=-w\log g(\boldsymbol{o})_i + w\log w + w-1 
    = \widetilde{\mathcal{L}}(\boldsymbol{e}^{(i)}(w),\boldsymbol{o}) . 
    \end{align*}  
     \end{small}
\end{proof}


\subsection{Proof of Proposition~\ref{prop: optimality of label z under beam search}}
\label{appendix: proof of proposition optimality of label z under beam search}
We begin with a recursive lemma and then apply it iteratively to prove the proposition.
\begin{Lemma}
   Given a user $u\in \mathcal{U}$, if the tree model, which consists of a tree $\mathcal{T}$ and a preference model $\mathcal{M}$, satisfies $R(\boldsymbol{o}^j, \mathbb{E}[\tilde{\boldsymbol{z}}^j|u])$ for all $1\le j \le H$, then
    \begin{align*}
        \mathcal{B}^j(u) \in \operatornamewithlimits{argTopk}\limits_{i\in \left[N_j\right]}\mathbb{E}[\tilde{\boldsymbol{z}}^j_i|u] \Longrightarrow \mathcal{B}^{j+1}(u) \in \operatornamewithlimits{argTopk}\limits_{i\in \left[N_{j+1}\right]} \mathbb{E}[\tilde{\boldsymbol{z}}^{j+1}_i|u]
    \end{align*}
    holds for any beam size $k$  within the range $1\le j \le H-1$.
    \label{lemma: recursive optimal beam search}
\end{Lemma}

\begin{proof}[Proof of Lemma~\ref{lemma: recursive optimal beam search}]
    We complete the proof by contradiction. When $N_{j+1}\le k$, $\operatornamewithlimits{argTopk}\limits_{i\in [N_{j+1}]}\mathbb{E}[\tilde{{z}}^{j+1}_i|u]=\{[N_{j+1}]\} = \mathcal{B}^{j+1}(u)$, so we only need to consider the case where $N_{j+1}>k$. 
    
    Denote $\operatornamewithlimits{argTopk}\limits_{i\in [N_{j+1}]}\mathbb{E}[\tilde{{z}}^{j+1}_i|u]$ as $\{\mathcal{D}^{j+1}_1,\mathcal{D}^{j+1}_2,\dots,\mathcal{D}^{j+1}_r\}$ without loss of generality, where $r$ is a natural number and each $\mathcal{D}_i^{j+1}$ is a set of $k$ elements. Suppose $\mathcal{B}^{j+1}(u)\notin \operatornamewithlimits{argTopk}\limits_{i\in [N_{j+1}]} \mathbb{E}[\tilde{{z}}^{j+1}_i|u]$, we can make the following derivation:
    \begin{align*}
        \mathcal{B}&^{j+1}(u)\notin \operatornamewithlimits{argTopk}\limits_{i\in [N_{j+1}]} \mathbb{E}[\tilde{{z}}^{j+1}_i|u]\\
        &\Longrightarrow \mathcal{B}^{j+1}(u) \neq \mathcal{D}^{j+1}_i(1\le i \le r)\\
        &\Longrightarrow \exists i_1,i_2 \in \left[N_{j+1}\right]\ s.t.\ i_1 \in \mathcal{B}^{j+1}(u) \wedge i_1\notin \mathcal{D}^{j+1}_i \\
        &\quad \quad \quad \quad \quad \quad \quad \quad \wedge i_2 \notin \mathcal{B}^{j+1}(u) \wedge i_2 \in \mathcal{D}_{i_0}^{j+1}\left(i_0 \in [r]\right)\\
        &\Longrightarrow \mathbb{E}[\tilde{{z}}^{j+1}_{i_1}|u] < \mathbb{E}[\tilde{{z}}^{j+1}_{i_2}|u] \wedge o^{j+1}_{i_1} \ge o^{j+1}_{i_2}\\
        &\Longrightarrow i_1 \in \widetilde{\mathcal{B}}^{j+1}(u) \wedge i_2 \notin \widetilde{\mathcal{B}}^{j+1}(u)\\
        &\Longrightarrow \delta(\rho^j(n_{i_1}^{j+1})) \in \mathcal{B}^j(u) \wedge \delta(\rho^j(n_{i_2}^{j+1})) \notin \mathcal{B}^j(u).
    \end{align*}
    
    Let's denote $\mathcal{B}^j(u)$ as $\{{s_1}, {s_2}, \dots, {s_k}\}$, where the elements satisfy $o^j_{s_1} \ge o^j_{s_2} \ge \dots \ge o^j_{s_k}$. Notice that $R(\boldsymbol{o}^j, \mathbb{E}[\tilde{\boldsymbol{z}}^j|u])$ and $\mathbb{E}[\tilde{{z}}^j_{s_i}|u]=\max\limits_{n'\in \mathfrak{L}(n_{s_1}^j)}\eta_{\pi(n')}$, we have: 
    \begin{equation}
     \max\limits_{n'\in \mathfrak{L}(n_{s_1}^j)}\eta_{\pi(n')}\ge \max\limits_{n'\in \mathfrak{L}(n_{s_2}^j)}\eta_{\pi(n')} \ge \dots \ge\max\limits_{n'\in \mathfrak{L}(n_{s_k}^j)}\eta_{\pi(n')}.
     \label{eq: rank of max probability }
    \end{equation}

     For each node $n_{s_i}^j$, we then consider its child node  $n^{j+1}_{c_i}\in \mathcal{C}(n_{s_i}^j)$ where $c_i=\operatornamewithlimits{argmax}\limits_{i\in \{\delta(n)\mid n \in \mathcal{C}(n^j_{s_i})\}} \mathbb{E}[\tilde{{z}}^{j+1}_{i}|u]$, and the following equation  \begin{equation}
         \max\limits_{n'\in \mathfrak{L}(n^{j+1}_{c_i})} \eta_{\pi(n')}=\max\limits_{n'\in \mathfrak{L}(n^j_{s_i})} \eta_{\pi(n')}
     \label{eq: probability inheriting}
     \end{equation}
     holds for $1\le i\le k$. Combine the {Eq. (\ref{eq: rank of max probability })} and {Eq. (\ref{eq: probability inheriting})}, we can get:
     \begin{small}
         \begin{equation*}
         \max\limits_{n'\in \mathfrak{L}(n^{j+1}_{c_1})} \eta_{\pi(n')}\ge \max\limits_{n'\in \mathfrak{L}(n^{j+1}_{c_2})} \eta_{\pi(n')}\ge \cdots \ge \max\limits_{n'\in \mathfrak{L}(n^{j+1}_{c_k})} \eta_{\pi(n')}.
     \end{equation*}
     \end{small}
    Then, for the $(j\text{+}1)$-th layer, with $R(\boldsymbol{o}^{j+1}, \mathbb{E}[\tilde{\boldsymbol{z}}^{j+1}|u])$ and $\mathbb{E}[\tilde{{z}}^{j+1}_{c_i}|u]=\max\limits_{n'\in \mathfrak{L}(n_{c_i}^{j+1})}\eta_{\pi(n')}$, we have $o^{j+1}_{c_1}\ge o^{j+1}_{c_2}\ge \dots \ge o^{j+1}_{c_k}$, which means that the rank of $o^{j+1}_{c_k}$ in the components of $\boldsymbol{o}^{j+1}$ is not lower than $k$. And because $i_1 \in \mathcal{B}^{j+1}(u)$, we can obtain the following results:
    \begin{small}
    \begin{align}
         o^{j+1}_{i_1} \ge o^{j+1}_{c_k}&\Longrightarrow \mathbb{E}[\tilde{{z}}^{j+1}_{i_1}|u] \ge \mathbb{E}[\tilde{{z}}^{j+1}_{c_k}|u] \Rightarrow \mathbb{E}[\tilde{{z}}^{j+1}_{i_2}|u] >\mathbb{E}[\tilde{{z}}^{j+1}_{c_k}|u]\notag\\
         &\Longrightarrow \max\limits_{n'\in \mathfrak{L}(n_{i_2}^{j+1})} \eta_{\pi(n')} > \max\limits_{n' \in \mathfrak{L}(n_{c_k}^{j+1})} \eta_{\pi(n')}.  
    \label{eq: max probability  i_2 gt c_k}
    \end{align}
    \end{small}
    
    Note that $i_{\rho}\triangleq\delta(\rho^j(n_{i_2}^{j+1}))\notin \mathcal{B}^j(u)$, $s_k\in\mathcal{B}^j(u)$, and  $ \mathcal{B}^j(u)\in \operatornamewithlimits{argTopk}\limits_{i\in [N_j]}\mathbb{E}[\tilde{{z}}^{j}_i|u]$, we can have $\mathbb{E}[\tilde{{z}}^j_{s_k}|u]\ge \mathbb{E}[\tilde{{z}}^j_{i_{\rho}}|u]$. Therefore, we have
    \begin{small}
    \begin{equation}
    \begin{split}
        \max\limits_{n'\in \mathfrak{L}(n^j_{s_k})} \eta_{\pi(n')}   \ge\max\limits_{n'\in \mathfrak{L}(n^j_{i_{\rho}})} \eta_{\pi(n')}\ge \max\limits_{n'\in \mathfrak{L}(n^{j+1}_{i_2})} \eta_{\pi(n')},
    \end{split}
    \label{eq: probability s_k ge i_2}
    \end{equation} 
    \end{small}where the last inequality holds as the $\mathfrak{L}(n^{j+1}_{i_2})\subseteq \mathfrak{L}(n^j_{i_{\rho}})$.
    Combine the {Eq. (\ref{eq: probability inheriting})} and {Eq. (\ref{eq: probability s_k ge i_2})}, we can get:
    \begin{equation}
    \max\limits_{n'\in \mathfrak{L}(n^{j+1}_{c_k})} \eta_{\pi(n')}\ge \max\limits_{n'\in \mathfrak{L}(n^{j+1}_{i_2})} \eta_{\pi(n')}. 
    \label{eq:  probability c_k ge i_2}
    \end{equation}
    {Eq. (\ref{eq: max probability  i_2 gt c_k})} and {Eq. (\ref{eq: probability c_k ge i_2})} induces the contradiction, so the proposition holds.
\end{proof}

\begin{proof}
  Let's consider the retrieval process along the tree for a given user $u$ and a given beam size $k$.  

  Let $j^*=\operatornamewithlimits{argmin}\limits_{\{j:N_j>k\}} N_j$. When  $1 \le j \le j^*-1$, $\mathcal{B}^j(u)=\{[N_j]\} \in \operatornamewithlimits{argTopk}\limits_{i \in [N_j]} \mathbb{E}[\tilde{{z}}^j_i|u]$. When $j\ge j^*$, by recursively applying {Lemma \ref{lemma: recursive optimal beam search}} from $j^*$ to $H-1$, we ultimately obtain $\mathcal{B}^H(u)\in \operatornamewithlimits{argTopk}\limits_{i\in [c(H)]}\mathbb{E}[\tilde{{z}}^H_i|u]$.
  Notice that $\mathbb{E}[\tilde{{z}}_i^H|u]=\eta_{\pi(n_i^H)}$, that is to say $\mathcal{B}^H(u) \in \operatornamewithlimits{argTopk}\limits_{i\in [N_H]}\eta_{\pi(n_i^H)}$,  meaning the final retrieved item set $\widehat{\mathcal{Y}}=\{\pi(n_i^H)|i\in \mathcal{B}^H(u)\} \in \operatornamewithlimits{argTopk}\limits_{y\in \mathcal{Y}} \eta_{y}(u)$.  The above analysis can be performed on any $u$ and $k$, so the tree model is top-$k$ retrieval Bayes optimal.
\end{proof}

\subsection{Proof of Proposition \ref{prop: the correctness of sampling probability}}
\label{appendix: proof of correctness of sampling probability}
\begin{proof}
    We use induction to prove this proposition. Assume $\sum_{i=1}^{N_j} q_i^j=1$, where $q_i^j$ is the sampling probability of node $n_i^j$ in the tree-based sampling, as presented in Eq. \eqref{eq: sampling probability in tree-based sampling}.

    For the base case, when $j=0$, the 0th layer contains only the root node $n_1^0$, which is sampled with probability $1$, and accordingly, the sum of sampling probabilities is exactly $1$. Suppose the assumption holds for $j$, we now need to prove that it also holds for $j+1$. For the layer $j+1$, the sum of sampling probabilities for nodes within this layer is calculated as follows:
    \begin{equation}
    \begin{split}
       \sum_{i=1}^{N_{j+1}}q_i^{j+1}&=\sum_{i^\prime=1}^{N_j} \sum_{i\in \{\delta(n) | n\in \mathcal{C}(n_{i^\prime}^j)\}}q_i^{j+1}\\
        &= \sum_{i^\prime=1}^{N_j} \sum_{i\in \{\delta(n) | n\in \mathcal{C}(n_{i^\prime}^j)\}}q_{i^{\prime}}^{j}\tilde{q}_i^{j+1}\\
        &=\sum_{i^\prime=1}^{N_j}q_{i^{\prime}}^{j} \sum_{i\in \{\delta(n) | n\in \mathcal{C}(n_{i^\prime}^j)\}}\tilde{q}_i^{j+1}\\
        &\underset{(a)}{=}\sum_{i^\prime=1}^{N_j}q_{i^{\prime}}^{j} = 1,
    \end{split}
    \end{equation}
    where (a) holds as the sum of local softmax probabilities of all child nodes of one node is $1$.
\end{proof}

\subsection{Proof of Theorem \ref{theo: tree-based sampling eq to softmax distribution}}
\label{appendix: proof of theory tree-based sampling eq to softmax distribution}
\begin{proof}
We use induction to prove this theorem. Assume $q_i^j= \frac{\exp{o_i^j}}{\sum_{k=1}^{N_j} \exp{o_{k}^{j}}}$, where $q_i^j$ denotes the sampling probability of node $n_i^j$ in the tree-based sampling.

For the base case, when $j = 0$, the 0th layer contains only the root node $n_1^0$, which is sampled with probability $1$, equaling its softmax probability. Suppose the assumption holds for all nodes in layer $j$, we now need to prove that it also holds for all nodes in layer $j+1$. For any $n_i^j$'s child node $n_c^{j+1}$ ($1\le c \le N_{j+1}$) in layer $j+1$, its sampling probability is calculated as follows:
\begin{align}
\label{eq:sampling_probability_of_child}
q_c^{j+1} &= q_i^{j} \cdot \frac{\exp{o_c^{j+1}}}{\sum\limits_{k\in \left\{\delta(n)\mid n\in \mathcal{C}(n_i^j)\right\}} \exp{o_k^{j+1}} } \\ 
&= \frac{\exp{o_i^j}}{\sum_{k=1}^{N_j} \exp{o_k^j}} \cdot \frac{\exp{o_c^{j+1}}}{ \sum\limits_{k\in \left\{\delta(n)\mid n\in \mathcal{C}(n_i^j)\right\}} \exp{o_k^{j+1}} } \notag
\end{align}
According to {Eq. (\ref{eq:score_proportional_to_children})}, and with the condition that the proportionality coefficients within the same layer are equal, we have
\begin{align}
\label{eq:nodes_score_propto_relation}
\sum_{k=1}^{N_j}  \exp{o_k^j} &\propto
\sum_{k=1}^{N_j} \sum\limits_{k'\in \left\{\delta(n)\mid n\in \mathcal{C}(n_k^j)\right\}} \exp{o_{k'}^{j+1}} \\
&= \sum_{k=1}^{N_{j+1}} \exp{o_{k}^{j+1}}. \notag
\end{align}
Therefore, by substituting equations {Eq. (\ref{eq:score_proportional_to_children})} and {Eq. (\ref{eq:nodes_score_propto_relation})} into {Eq. (\ref{eq:sampling_probability_of_child})}, we can obtain:
\begin{equation*}
q_c^{j+1} = \frac{\exp{o_{c}^{j+1}}}{\sum_{k=1} ^{N_{j+1}} \exp{o_{k}^{j+1}}}.
\end{equation*}
That is, when sampling to the \((j \text{+} 1)\)-th layer in the tree-based sampling, the probability of sampling the corresponding node is equivalent to its softmax probability in layer \(j \text{+} 1\). As \(j\) increases, when reaching the last layer where the leaf nodes are located, this conclusion also holds. 
\end{proof}

\subsection{Proof of Proposition \ref{prop: rectification method bounds probability bias}}
\label{appendix: proof of prop bounded probability bias}
\begin{proof}[Proof]
    For node $n_i^j$, let the indices of its ancestors from the root to itself be denoted as $\left(\rho^0(i), \rho^1(i),\dots, \rho^{j-1}(i), \rho^j(i) \right)$, where $\rho^l(i)=\delta(\rho^l(n_i^j))$ and $\rho^j(i)$ is exactly $i$. Denote the softmax probability $ \frac{\exp{o_i^j}}{\sum_{k=1}^{N_j} \exp{o_{k}^{j}}}$ as $p_i^j$, then we have  $\operatorname{bias}_i^j=q_i^j-p_i^j$.  Since the sampling probability is calculated as Eq. \eqref{eq: sampling probability in tree-based sampling}, we derive the following recurrence relation: 
    $$\begin{aligned}
    \operatorname{bias}^{j}_i &=q^{j-1}_{\rho^{j-1}(i)}\tilde{q}_{i}^{j}-p_{i}^{j} \\
    &=(q^{j-1}_{\rho^{j-1}(i)}-p_{\rho^{j-1}(i)}^{j-1})\tilde{q}_{i}^{j}+p_{\rho^{j-1}(i)}^{j-1}\tilde{q}^{j}_i-p_{i}^{j} \\
    &\underset{(a)}{=}(q_{\rho^{j-1}(i)}^{j-1}-p_{\rho^{j-1}(i)}^{j-1})\tilde{q}_{i}^{j}\\
    &\quad\quad+\frac{\lambda_{\rho^{j-1}(i)}^{j-1}\exp o_{i}^{j}}{\sum_{k=1}^{N_{j}}\lambda_{\rho^{j-1}(k)}^{j-1}\exp o_{k}^{j}}-\frac{\exp o_{i}^{j}}{\sum_{k=1}^{N_{j}}\exp o_{k}^{j}} \\
    &=\tilde{q}_{i}^{j}\operatorname{bias}_{\rho^{j-1}(i)}^{j-1}+I(\rho^{j}(i),j), \\
    \end{aligned}
    $$where (a) results from combining Eq. \eqref{eq: definition of lambda}, Eq. \eqref{eq:score_proportional_to_max_children}, and analogous techniques used in Eq. \eqref{eq:nodes_score_propto_relation}, and $I(\rho^j(i),j)\triangleq\frac{\lambda_{\rho^{j-1}(i)}^{j-1}\exp o_{\rho^j(i)}^{j}}{\sum_{k=1}^{N_{j}}\lambda_{\rho^{j-1}(k)}^{j-1}\exp o_{k}^{j}}-\frac{\exp o_{\rho^j(i)}^{j}}{\sum_{k=1}^{N_{j}}\exp o_{k}^{j}}$.
    
    Notably, $ \operatorname{bias}^0_1 = 0 $ and $\operatorname{bias}^1_i = 0$ for \( i \in [N_1] \). For \( j \ge 2 \), we obtain:
    $$\operatorname{bias}_{i}^{j}=\sum_{k=2}^{j-1}I(\rho^{k}(i),k)\cdot\prod_{t=k+1}^{j}\tilde{q}_{\rho^t(c)}^{ t}+I(\rho^{j}(i),j).$$
    Given that $\lambda_{i}^{j}\in(1-\varepsilon,1)$ holds for any $i\in N_j$ and $j\in [H]$,  it follows that $$|I(\rho^{k}(i),k)|\leq\frac{\varepsilon}{1-\varepsilon}p_{\rho^k (i)}^{k}\leq\frac{\varepsilon}{1-\varepsilon}.$$
    Thus, we conclude 
    \begin{equation*}
        \vert\operatorname{bias}_{i}^{j}\vert\leq\frac{(j-1)\varepsilon}{1-\varepsilon}.
    \end{equation*}
\end{proof}

\subsection{Proof of Corollary \ref{cor: bound of KL divergence}}
\label{appendix: proof of corollary bound}
\begin{proof}[Proof]
    For $\forall i \in [N_j]$, $|o_i^j|\le B_o$, then we can obtain:
    $$
        p_i^j = \frac{\exp o_i^j}{\sum_{k=1}^{N_j} \exp o_k^j}\ge \frac{\exp (-B_o)}{N_j\exp B_o}= \frac{1}{N_j\exp 2B_o}.
    $$
    Then, we have 
    \begin{equation}
    \begin{split}
        D_{KL}(Q^j||P^j)&=\sum_{i=1}^{N_j} q_i^j\log \frac{q_i^j}{p_i^j}\\ 
        &= \sum_{i=1}^{N_j} q_i^j \log \left( 1+\frac{\operatorname{bias}_i^j}{p_i^j}\right)\\
        &\le \sum_{i=1}^{N_j} q_i^j \log \left( 1+{\operatorname{bias}_i^j}N_j e^{2B_o}\right)\\
        &\le \sum_{i=1}^{N_j} q_i^j \log \left( 1+ B_\varepsilon^j N_j e^{2B_o}\right)\\
        &=\log \left( 1+ B_\varepsilon^j N_j e^{2B_o}\right),
    \end{split}
    \end{equation}
    where $B_\varepsilon^j \triangleq \frac{(j-1)\varepsilon}{1-\varepsilon}$ and $\operatorname{bias}_i^j \le B_\varepsilon^j$ as presented in Proposition \ref{prop: rectification method bounds probability bias}.
\end{proof}

\subsection{Proof of Lemma \ref{lemm: Rademacher complexity of DIN}}
\label{appendix: proof of lemm Rademacher complexity of DIN}
\begin{proof}[Proof]
For any vectors $\boldsymbol{v},\boldsymbol{u}_i\in\mathbb{R}^d,\|\boldsymbol{v}\|_1 \leq B_v$, notice the following inequality:  
    \begin{equation}
    \label{eq:inequlity_of_vector}
    \sup_{\boldsymbol{v}}\sum_i\boldsymbol{v}^{\mathsf{T}}\boldsymbol{u}_i\leq B_v\max_{j\in[d]}\left|\sum_i\boldsymbol{e}_j\boldsymbol{u}_i\right|\leq\sum_iB_v\max_{\substack{j\in[d]\\ s\in\{-1,1\}}}s\boldsymbol{e}_j\boldsymbol{u}_i.
    \end{equation}
By applying the {Eq. (\ref{eq:inequlity_of_vector})}, we can get:
\begin{align}
\label{eq: main_rademacher_DIN}
  \widehat{\mathcal{R}}_m&(\mathcal{M}) = \mathbb{E}_{\sigma}\left[\sup\limits_{f\in \mathcal{M}}\frac{1}{m}\sum_{i=1}^{m} \sigma_i f(u_i,n_i)\right] \notag\\
  &=\mathbb{E}_{\sigma} \Bigg[\sup\limits_{\boldsymbol{z}, \boldsymbol{w}_n, \{\boldsymbol{W}_k\}^L_{k=1}}\frac{1}{m}\sum_{i=1}^{m}\sigma_i\Big<\boldsymbol{W}_L, \phi_{L-1}\circ \notag\\
  &\quad \quad \phi_{L-2}\circ\dots\circ\phi_{1}\left(\boldsymbol{z}^{(u_i)}_1;\boldsymbol{z}^{(u_i)}_2;\dots;\boldsymbol{z}^{(u_i)}_{K'};\boldsymbol{w}_{n_i}\right)\Big>\Bigg]\notag\\
  &\underset{(a)}{\le} \|\boldsymbol{W}_L\|_1 \mathbb{E}_{\sigma} \Bigg[\sup\limits_{\substack{s\in\{-1,1\},j\in[d_{L-1}]\\ \boldsymbol{z}, \boldsymbol{w}_n,\{\boldsymbol{W}_k\}_{k=1}^{L-1}}}\frac{1}{m}\sum_{i=1}^{m}s\sigma_i\Big<\boldsymbol{e}_j, \phi_{L-1}\circ \notag\\
  \displaybreak[1]
  &\quad \quad \phi_{L-2}\circ\dots\circ \phi_{1}\left(\boldsymbol{z}^{(u_i)}_1;\boldsymbol{z}^{(u_i)}_2;\dots;\boldsymbol{z}^{(u_i)}_{K'};\boldsymbol{w}_{n_i}\right)\Big>\Bigg] \notag\\
  &\underset{(b)}{\le} c_{\phi}\|\boldsymbol{W}_L\|_1 \mathbb{E}_{\sigma} \Bigg[\sup\limits_{\substack{s\in \{-1,1\},j\in[d_{L-1}]\notag\\ 
  \boldsymbol{z}, \boldsymbol{w}_n,\{\boldsymbol{W}_k\}_{k=1}^{L-1}}}\frac{1}{m}\sum_{i=1}^{m}s\sigma_i\Big<\boldsymbol{e}_j, \boldsymbol{W}_{L-1}\cdot \notag\\
  &\quad \quad \phi_{L-2}\circ\dots\circ \phi_{1}\left(\boldsymbol{z}^{(u_i)}_1;\boldsymbol{z}^{(u_i)}_2;\dots;\boldsymbol{z}^{(u_i)}_{K'};\boldsymbol{w}_{n_i}\right)\Big>\Bigg] \notag\\
  &\underset{(c)}{\le} c_{\phi} \|\boldsymbol{W}_L\|_1 \|\boldsymbol{W}_{L-2}\|_1 \mathbb{E}_{\sigma} \Bigg[\sup\limits_{\substack{s\in \{-1,1\},j\in[d_{L-2}] \\
  \boldsymbol{z}, \boldsymbol{w}_n,\{\boldsymbol{W}_k\}_{k=1}^{L-1}}}\frac{1}{m}\sum_{i=1}^{m}s\sigma_i \notag\\
  \displaybreak[1]
   &\quad \quad \Big<\boldsymbol{e}_j,\phi_{L-2}\circ\dots\circ \phi_{1}\left(\boldsymbol{z}^{(u_i)}_1;\boldsymbol{z}^{(u_i)}_2;\dots;\boldsymbol{z}^{(u_i)}_{K'};\boldsymbol{w}_{n_i}\right)\Big>\Bigg] \notag\\
   &\underset{(d)}{\le} 2c_{\phi}^{L-1}\prod_{k=1}^{L}\|\boldsymbol{W}_k\|_1 \mathbb{E}_{\sigma}\Bigg[\sup\limits_{\substack{\boldsymbol{z}, \boldsymbol{w}_n, \notag\\ j\in[(K'+1)d]}}\frac{1}{m}\sum_{i=1}^m \sigma_i\Big<\boldsymbol{e}_j,  \notag\\
   \displaybreak[1]
   &\quad \quad \left(\boldsymbol{z}^{(u_i)}_1;\boldsymbol{z}^{(u_i)}_2;\dots;\boldsymbol{z}^{(u_i)}_{K'};\boldsymbol{w}_{n_i}\right)\Big>\Bigg] \notag\\
   &\le 2c_{\phi}^{L-1} \prod_{k=1}^{L}\|\boldsymbol{W}_k\|_1 \Bigg(\underbrace{\mathbb{E}_{\sigma}\Big[\sup_{j\in[d],\boldsymbol{w}_{n}}\frac{1}{m}\sum_{i=1}^m \sigma_{i}\left<\boldsymbol{e}_{j},\boldsymbol{w}_{n_i}\right>\Big]}_{I_1}  \notag\\  
   &\quad + \sum_{t=1}^{K^{\prime}}\ \underbrace{\mathbb{E}_{\sigma}\Big[\sup_{j\in [d],\boldsymbol{z}_{t}}\frac{1}{m}\sum_{i=1}^{m}\sigma_{i}\left<\boldsymbol{e}_{j},\boldsymbol{z}_{t}^{(u_i)}\right>\Big]}_{I_2} \Bigg)
\end{align}
where (a) use {Eq. (\ref{eq:inequlity_of_vector})} as $\boldsymbol{W}_L\in \mathbb{R}^{1\times d_{L-1}}$ is actually a vector; (b) holds since $\phi$ is applied element-wise, we can bring $\boldsymbol{e}^T_j$ inside the function and the use of contraction inequality \cite{hahn1994probability}; (c) use {Eq. (\ref{eq:inequlity_of_vector})} again as $\boldsymbol{e}_{j}^\mathsf{T} \boldsymbol{W}_{L-1}$ is still a vector; (d) holds by applying {Eq. (\ref{eq:inequlity_of_vector})} recursively and utilizing the fact that $\widehat{\mathcal{R}}_m(\mathcal{F}\cup -\mathcal{F})\le 2 \widehat{\mathcal{R}}_m(\mathcal{F})$.

As the term $I_1$, using Cauchy-Schwarz inequality and Jensen inequality, we have:
    \begin{align}
    \label{eq: rademacher of I1}
    I_1& \le \frac{1}{m}\mathbb{E}_{\sigma}\left[ \left\|\sum_{i=1}^m\sigma_i\boldsymbol{w}_{n_i}\right\|_2\right] \le \frac{1}{m}\left(\mathbb{E}_{\sigma}\left[\left\|\sum_{i=1}^m\sigma_i\boldsymbol{w}_{n_i}\right\|_2^2\right]\right)^{1/2} \notag\\
    & =\frac{1}{m}\left(\sum_{i=1}^{m}\left\|\boldsymbol{w}_{n_i}\right\|^2\right)^{1/2} \le \frac{B_0}{\sqrt{m}}.
    \end{align}
As the term $I_2$, we have:
\begin{align}
\label{eq: rademacher of I2}
I_2&=\frac{1}{m}\mathbb{E}_{\sigma}\left[\sup_{j\in[d],\boldsymbol{w}}\sum_{i=1}^{m}\sigma_i\left\langle\boldsymbol{e}_{j},\sum_{k\in T_t}w_k^{(u_i)}\boldsymbol{a}_{k}^{(u_i)}\right\rangle\right]  \notag\\
&\leq\frac{1}{m}\sum_{k\in T_t}\mathbb{E}_{\sigma}\left[\sup_{j\in[d],\boldsymbol{w}}\sum_{i=1}^m\sigma_i\left\langle\boldsymbol{e}_{j},w_k^{(u_i)}\boldsymbol{a}_{k}^{(u_{i})}\right\rangle\right] \notag\\
&=\frac{1}{m}\sum_{k\in T_t}\mathbb{E}_{\sigma}\Bigg[\sup_{j\in[d],\boldsymbol{w}}\sum_{i=1}^m\sigma_i\phi\Big(\boldsymbol{W}_w^{(2)}\phi\Big(\boldsymbol{W}_w^{(1)} \notag\\
&\quad \quad \quad  \left[\boldsymbol{a}_k^{(u_i)};\boldsymbol{a}_k^{(u_i)}\odot\boldsymbol{w}_{n_i};\boldsymbol{w}_{n_i}\right]\Big)\Big)\left\langle\boldsymbol{e}_j,\boldsymbol{a}_k^{(u_i)}\right\rangle\Bigg] \notag\\
&\leq \frac{1}{m}c_\phi^2\left\|\boldsymbol{W}_w^{(1)}\right\|_1\left\|\boldsymbol{W}_w^{(2)}\right\|_1\sum_{k\in T_t}\mathbb{E}_{\sigma}\Bigg[\sup_{\substack{\boldsymbol{w},j\in[d]\\j^{\prime}\in[3d]}}\sum_{i=1}^{m}\sigma_i\notag\\
&\quad \quad \left\langle\boldsymbol{e}_{j^{\prime}},\left[\boldsymbol{a}_k^{(u_i)};\boldsymbol{a}_k^{(u_i)}\odot\boldsymbol{w}_{n_i};\boldsymbol{w}_{n_i}\right]\right\rangle\left\langle\boldsymbol{e}_j,\boldsymbol{a}_{k}^{(u_i)}\right\rangle\Bigg].
\displaybreak[1]
\end{align}    
Furthermore, we can get
\begin{small}
\begin{equation}
\begin{aligned}
\label{eq: rademacher of I3_I4_I5}
&\mathbb{E}_{\sigma}\Bigg[\sup_{\substack{\boldsymbol{w},j\in[d]\\j^{\prime}\in[3d]}}\sum_{i=1}^m\sigma_i \left\langle \boldsymbol{e}_{j^{\prime}},\left[\boldsymbol{a}_{k}^{(u_i)};\boldsymbol{a}_{k}^{(u_i)}\odot\boldsymbol{w}_{n_i};\boldsymbol{w}_{n_i}\right]\right\rangle\left\langle\boldsymbol{e}_{j},\boldsymbol{a}_{k}^{(u_i)}\right\rangle\Bigg] \\
&\quad\leq\mathbb{E}_{\sigma}\left[\sup_{j\in[d],j'\in[d]}\sum_{i=1}^{m}\sigma_i\left\langle\boldsymbol{e}_{j^{\prime}},\boldsymbol{a}_k^{(u_i)}\right\rangle\left\langle\boldsymbol{e}_j,\boldsymbol{a}_k^{(u_i)}\right\rangle\right] \\
&\quad\quad +\mathbb{E}_{\sigma}\left[\sup_{j\in[d],j^{\prime}\in[d],\boldsymbol{w}}\sum_{i=1}^{m}\sigma_i\left\langle\boldsymbol{e}_{j^{\prime}},\boldsymbol{a}_k^{(u_i)}\odot\boldsymbol{w}_{n_i}\right\rangle\left\langle\boldsymbol{e}_j,\boldsymbol{a}_k^{(u_i)}\right\rangle\right] \\
&\quad\quad +\mathbb{E}_{\sigma}\left[\sup_{j\in[d],j^{\prime}\in[d],\boldsymbol{w}}\sum_{i=1}^m\sigma_i\left<\boldsymbol{e}_{j^{\prime}},\boldsymbol{w}_{n_i}\right>\left<\boldsymbol{e}_j,\boldsymbol{a}_{k}^{(u_i)}\right>\right] \\
&=I_{3}+I_{4}+I_{5}.
\end{aligned}
\end{equation}
\end{small}
As the term $I_3$, notice that
\begin{small}
\begin{equation*}
\begin{aligned}
&\sum_{i=1}^m\sigma_{i}\left<\boldsymbol{e}_{j^{\prime}},\boldsymbol{a}_{k}^{(u_i)}\right>\left<\boldsymbol{e}_{j},\boldsymbol{a}_{k}^{(u_i)}\right> =\sum_{i=1}^m\sigma_{i}\boldsymbol{e}_{j^{\prime}}^{\mathsf{T}}\boldsymbol{P}_{a}^{(u_i)}\boldsymbol{e}_{j}\\
&=\sum_{i=1}^m \sigma_{i}\mathrm{Tr}\left(\boldsymbol{e}_{j}\boldsymbol{e}_{j^{\prime}}^{\mathsf{T}}\boldsymbol{P}_{a}^{(u_i)}\right) =\mathrm{Tr}\left(\boldsymbol{e}_{j}\boldsymbol{e}_{j^{\prime}}^{\mathsf{T}}\left(\sum_{i=1}^m\sigma_{i}\boldsymbol{P}_{a}^{(u_i)}\right)\right)\\
&=\left\langle \boldsymbol{e}_{j}\boldsymbol{e}_{j^{\prime}}^{\mathsf{T}},\sum_{i=1}^m\sigma_{i}\boldsymbol{P}_{a}^{(u_i)}\right\rangle_{F},
\end{aligned}
\end{equation*}
\end{small}
where $\boldsymbol{P}_{a}^{(u_i)}=\boldsymbol{a}_{k}^{(u_i)}{\boldsymbol{a}_{k}^{(u_i)}}^{\mathsf{T}}$. Then, we can get
\begin{equation}
    \begin{aligned}
    \label{eq: rademacher of I3}
I_3& =\mathbb{E}_\sigma\left[\sup_{j\in[d],j^{\prime}\in[d]}\sum_{i=1}^m\sigma_i\left\langle\boldsymbol{e}_{j^{\prime}},\boldsymbol{a}_k^{(u_i)}\right\rangle\left\langle\boldsymbol{e}_j,\boldsymbol{a}_k^{(i)}\right\rangle\right]  \\
&=\mathbb{E}_{\sigma}\left[\sup_{j\in[d],j^{\prime}\in[d]}\left\langle \boldsymbol{e}_{j}\boldsymbol{e}_{j^{\prime}}^{\mathsf{T}},\sum_{i=1}^m\sigma_{i}\boldsymbol{P}_{a}^{(u_i)}\right\rangle_{F}\right] \\
&\leq\mathbb{E}_{\sigma}\left[\left\|\sum_{i=1}^m\sigma_i\boldsymbol{P}_a^{(u_i)}\right\|_F\right]=\sqrt{\sum_{i=1}^m\left\|\boldsymbol{P}_{a}^{(u_i)}\right\|_{F}^{2}}\leq\sqrt{m B_{a}^{4}}.
\end{aligned}
\end{equation}

As the term $I_4$, use the same analysis technique as for $I_3$, we can get

\begin{align}
&I_{4}=\mathbb{E}_{\sigma}\left[\sup_{j\in[d],j^{\prime}\in[d],\boldsymbol{w}}\sum_{i=1}^m\sigma_{i}\left\langle\boldsymbol{e}_{j^{\prime}},\boldsymbol{a}_{k}^{(u_i)}\odot\boldsymbol{w}_{n_i}\right\rangle\left\langle\boldsymbol{e}_{j},\boldsymbol{a}_{k}^{(u_i)}\right\rangle\right] \notag \\
&=\mathbb{E}_{\sigma}\left[\sup_{j\in[d],j^{\prime}\in[d],\boldsymbol{w}}\sum_{i=1}^m\sigma_i\left\langle\boldsymbol{e}_{j^{\prime}}\odot\boldsymbol{w}_{n_i},\boldsymbol{a}_k^{(u_i)}\right\rangle\left\langle\boldsymbol{e}_{j},\boldsymbol{a}_k^{(u_i)}\right\rangle\right] \notag \\
&=\mathbb{E}_{\sigma}\left[\sup_{j\in[d],j^{\prime}\in[d],\boldsymbol{w}}\sum_{i=1}^m\sigma_{i}\left\langle \boldsymbol{e}_{j}\boldsymbol{e}_{j^{\prime}}^{\mathsf{T}}\odot\boldsymbol{w}_{n_i}^{\mathsf{T}},\boldsymbol{P}_{a}^{(u_i)}\right\rangle_{F}\right] \notag \\
&=\mathbb{E}_{\sigma}\left[\sup_{j\in[d],j^{\prime}\in[d],\boldsymbol{w}}\sum_{n\in \mathcal{N}}\sum_{i:n_i=n}\sigma_{i}\left\langle \boldsymbol{e}_{j}\boldsymbol{e}_{j^{\prime}}^{\mathsf{T}}\odot\boldsymbol{w}_{n}^{\mathsf{T}},\boldsymbol{P}_{a}^{(u_i)}\right\rangle_{F}\right] \notag \\
\displaybreak[1]
&\le \sum_{n\in\mathcal{N}}\mathbb{E}_{\sigma}\left[\sup_{j\in[d],j^{\prime}\in[d]}\left\langle \boldsymbol{e}_{j}\boldsymbol{e}_{j^{\prime}}^{\mathsf{T}}\odot\boldsymbol{w}_{n}^{\mathsf{T}},\sum_{i:n_i=n}\sigma_{i}\boldsymbol{P}_{a}^{(u_i)}\right\rangle_{F}\right] \notag \\
&\leq B_0 \sum_{n\in \mathcal{N}} \mathbb{E}_{\sigma}\left[\left\|\sum_{i:n_i=n}\sigma_i \boldsymbol{P}_{a}^{(u_i)}\right\|_{F}\right] \notag \\
&\leq B_0 B_a^2\sum_{n\in \mathcal{N}}\sqrt{\left|\{i:n_i=n\}\right|}\le B_0B_a^2\sqrt{\frac{2B-1}{B-1}|\mathcal{Y}|m},
\label{eq: rademacher of I4}
\end{align}
where the last inequality holds as $\sum_{n\in \mathcal{N}}\left|\{i:n_i=n\}\right|=m$, $|\mathcal{N}|=\sum_{j=0}^{\lceil\log_B |\mathcal{Y}|\rceil-1}B^j+|\mathcal{Y}|\le \frac{2B-1}{B-1}|\mathcal{Y}|$ for the $B$-ary tree, and use the Cauchy-Schwarz inequality.

As the term $I_5$, use the same technique as for $I_4$, we have
\begin{small}
\begin{align}
\label{eq: rademacher of I5}
I_{5}& =\mathbb{E}_{\sigma}\left[\sup_{j\in[d],j^{\prime}\in[d],\boldsymbol{w}}\sum_{i=1}^m \sigma_i\left<\boldsymbol{e}_{j^{\prime}},\boldsymbol{w}_{n_i}\right>\left<\boldsymbol{e}_{j},\boldsymbol{a}_{k}^{(u_i)}\right>\right]  \notag\\
&=\mathbb{E}_{\sigma}\left[\sup_{j\in[d],j^{\prime}\in[d],\boldsymbol{w}}\sum_{n\in \mathcal{N}}\sum_{i:n_i=n}\sigma_i\left<\boldsymbol{e}_{j^{\prime}},\boldsymbol{w}_{n_i}\right>\left<\boldsymbol{e}_{j},\boldsymbol{a}_{k}^{(u_i)}\right>\right]  \notag\\
\displaybreak[1]
& \leq \sum_{n\in\mathcal{N}}\mathbb{E}_{\sigma}\left[\sup_{j\in[d],j^{\prime}\in[d]}\left<\boldsymbol{e}_{j^{\prime}},\boldsymbol{w}_{n}\right>\left<\boldsymbol{e}_{j},\sum_{i:n_i=n} \sigma_i\boldsymbol{a}_{k}^{(u_i)}\right>\right]\notag\\
&\leq B_{0}\sum_{n\in\mathcal{N}}\mathbb{E}_{\sigma}\left[\left\|\sum_{i:n_i=n}\sigma_{i}\boldsymbol{a}_{k}^{(u_i)}\right\|_{2}\right]\leq B_{0}B_{a}\sqrt{\frac{2B-1}{B-1}|\mathcal{Y}|m}.
\end{align}
\end{small}

Combine equations {Eq. (\ref{eq: main_rademacher_DIN})}$\sim${Eq. (\ref{eq: rademacher of I5})}, and notice the fact that $\sum_{t=1}^{K'}T_t=K$, we have
\begin{equation*}
\begin{split}
&\widehat{\mathcal{R}}_m(\mathcal{M})=\mathbb{E}_{\sigma}\left[\sup_{f\in\mathcal{M}}\frac{1}{m}\sum_{i=1}^m\sigma_{i}f\left(u_{i},n_i\right)\right] \\
&\leq 2c_{\phi}^{L-1}\prod_{k=1}^{L}\|\boldsymbol{W}_{k}\|_{1}\Bigg[\frac{B_0}{\sqrt{m}}+c_{\phi}^{2}\left\|\boldsymbol{W}_{w}^{(1)}\right\|_{1}\left\|\boldsymbol{W}_{w}^{(2)}\right\|_{1}K\cdot\\
&\quad \left(\frac{B_a^2}{\sqrt{m}}+\frac{B_0B^2_a}{\sqrt{m}}\sqrt{\frac{2B-1}{B-1}|\mathcal{Y}|}+\frac{B_0B_a}{\sqrt{m}}\sqrt{\frac{2B-1}{B-1}|\mathcal{Y}|}\right)\Bigg]\\
&\le \frac{2c_{\phi}^{L-1}B_1^L(B_0+KB_{w_1}+KB_{w_2}\mathcal{\tau})}{\sqrt{m}},
\end{split}
\end{equation*}
where $B_{w_1}=c_{\phi}^2B_2^2B^2_a, B_{w_2}=c_{\phi}^2B_0B_2^2\left(B_a^2+B_a\right),\tau=\sqrt{(2B-1)|\mathcal{Y}|}/\sqrt{B-1}$.
\end{proof}

\subsection{Proof of Lemma ~\ref{Lemma: generalization bound of single layer}}
\label{appendix: proof of theory generalization bound of single layer}
\begin{proof}
Given the sample set  $S=\{u_i,y_i\}_{i=1}^{m}$,  we consider the loss function space  $\mathcal{F}_{\ell}^{j}=\{(u,y)\rightarrow f_{\ell}^j(u,y)\}$, where function $f_{\ell}^j(u,y)=\widetilde{\mathcal{L}}_j(u,y)=-\bar{z}^j_{\delta^j(y)}\log \frac{\exp o^j_{\delta^j(y)}(u)}{\sum_{k=1}^{N_j}\exp o^j_{k}(u)}$ and $o_k^j(u)=f_{\text{DIN}}(u,n_k^j)$. Due to $\left| f_{\text{DIN}}\right|\le B_{\mathcal{M}}$, we can get $\left| f_{\ell}^j\right|\le 2B_{\mathcal{M}}+\log N_j$. Then, by the {Lemma \ref{lemma: Rademacher complexity lemma}}, we have:
\begin{small}
\begin{equation}
\label{eq: generalization bound of j-th layer}
    \begin{split}
       \mathbb{E}_{(u,y)\sim \mathbb{P}}\left[\widetilde{\mathcal{L}}_j(u,y)\right] &\le \frac{1}{m} \sum_{i=1}^{m}\widetilde{\mathcal{L}}_j(u_i,y_i) + 2\widehat{\mathcal{R}}_m(\mathcal{F}_\ell^j, S) \\
       &+\Big(4\log N_j+8B_{\mathcal{M}}\Big)\sqrt{\frac{2\log\left(4/\delta\right)}{m}}.
    \end{split}
\end{equation} 
\end{small}For the empirical Rachemader complexity of $\mathcal{F}_{\ell}^j$, we have
\begin{small}
\begin{align}
&\widehat{\mathcal{R}}_m(\mathcal{F}_{\ell}^j,S)=\mathbb{E}_{\sigma}\left[\sup_{f_{\ell}^j\in \mathcal{F}_{\ell}^j}\frac{1}{m}\sum_{i=1}^{m}\sigma_i f_{\ell}^j(u_i,y_i)\right] \notag\\
&=\mathbb{E}_{\sigma}\left[\sup_{\substack{f_{\text{DIN}}\in \mathcal{M}\\\bar{z}\in \{0,1\}}}\frac{1}{m}\sum_{i=1}^m \sigma_i\bar{z}^j_{\delta^j(y_i)}\log \frac{\sum_{k=1}^{N_j}\exp o^j_k(u_i)}{\exp o_{\delta^j(y_i)}^j(u_i)}\right] \notag \\
\displaybreak[1]
&=\mathbb{E}_{\sigma}\left[\sup_{\substack{f_{\text{DIN}}\in\mathcal{M}\\z^{\prime}\triangleq2\bar{z}-1\in \{-1,1\}}}\frac{1}{m}\sum_{i=1}^m \sigma_i\frac{z^{\prime}+1}{2}\log \frac{\sum_{k=1}^{N_j}\exp o^j_k(u_i)}{\exp o_{\delta^j(y_i)}^j(u_i)}\right]\notag\\
&\le\mathbb{E}_{\sigma}\left[\sup_{\substack{f_{\text{DIN}}\in\mathcal{M}}}\frac{1}{m}\sum_{i=1}^m \sigma_i * -\log \frac{\exp o_{\delta^j(y_i)}^j(u_i)}{\sum_{k=1}^{N_j}\exp o^j_k(u_i)}\right] \notag\\
&= \widehat{\mathcal{R}}_m(\ell\circ\mathcal{M}, S), 
\end{align}
\label{eq: transformation of Rademacher complexity}
\end{small}where the last equality holds  because 
\begin{small}
    $$ -\log \frac{\exp o_{\delta^j(y_i)}^j(u_i)}{\sum\limits_{k=1}\limits^{N_j}\exp o^j_k(u_i)}=-\sum_{n=1}^{N_j}\mathbb{I}(n=\delta^j(y_i))\cdot \log \frac{\exp o^j_n(u_i)}{\sum\limits_{k=1}\limits^{N_j}\exp o_k^j(u_i)}$$ 
\end{small}is a composition of logistic loss $\ell_{y}(\boldsymbol{o})=-\sum_{i}y_i\log \frac{\exp \boldsymbol{o}_i}{\sum_{j} \exp \boldsymbol{o}_j}$ and $f_{\text{DIN}}$. 
Since the partial derivative of $\ell_{y}(\boldsymbol{o})$ w.r.t. each component is bounded by $1$, the logistic loss function $\ell$ is $1$-Lipschitz~\cite{wan2013regularization}.  For the $j$-th layer, where the number of classes is $N_j$, by {Lemma \ref{Lemma: logistic rademacher}}, we have 
\begin{equation}
\label{eq: using lemma to transfer the Rademacher complexity}
\widehat{\mathcal{R}}_m(\ell\circ\mathcal{M}, S) \le 2N_j\widehat{\mathcal{R}}_m(\mathcal{M}, S).
\end{equation}

Combine the {Eq. (\ref{eq: generalization bound of j-th layer})}, {Eq. (\ref{eq: transformation of Rademacher complexity})} and {Eq. (\ref{eq: using lemma to transfer the Rademacher complexity})}, and utilize the {Lemma \ref{Lemma: Rademacher complexity of DIN}}, we can obtain the desired result.
\end{proof}

\subsection{Proof of Lemma \ref{Lemma: bound of sampled softmax}}
\begin{proof}
The $i$-th node is positive among  $N_j$ nodes in the $j$-th layer, and the sampling distribution is $Q^j$, with the sampling probability $q_{i'}^j$ for node $n_{i'}^j$. Therefore, according to {Eq. (\ref{eq:adjusted_output})}, for $i^{\prime}\in \mathcal{I}^j_M\cup \{i\}$, the adjusted logit is calculated as follows:
\begin{equation}
\label{eq:adjusted_output_lemma}
    {\hat{o}}^j_{i^{\prime}}=\left\{
    \begin{aligned}
    &{{o}}^j_{i}-\ln(Mq^j_{i'})\ \ \  \text{if}\ \ i^{\prime}\neq i \\
    &{{o}}^j_{i}-\ln(1) \quad \quad  \ \ \text{if}\ \ i^{\prime} = i
    \end{aligned}
    \right.
\end{equation}
Then, we can have:
\begin{small}
\begin{align*}
&\ell^j_{softmax}(u,y) - \mathbb{E}_{\mathcal{I}^{j}_{M}}\left[\ell^j_{sampled\text{-}softmax}(u,y,\mathcal{I}^{j}_{M})\right]\\
    &\quad=-\log \frac{\exp {o}^j_{i}}{\sum_{k=1}^{N_j}\exp {o}^j_k}- {\mathbb{E}}_{\mathcal{I}'_{M}}\Big[-\log \frac{\exp {\hat{o}}_{i}^j}{\sum\limits_{i'\in \mathcal{I}^{j}_{M}\cup\{i\}}\exp {\hat{o}}_{i'}^j}\Big]\\ &\quad=\mathbb{E}_{\mathcal{I}^{j}_{M}}\left[\log \Big(\sum_{k=1}^{N_j}\exp {o}^j_k\Big)-\log \left(\exp {o}^j_i+\sum_{i^{\prime}\in\mathcal{I}^{j}_{M}}\frac{\exp {o}^j_{i^{\prime}}}{Mq_{i^{\prime}}^j}\right)\right]\\
    &\quad=\mathbb{E}_{\mathcal{I}^{j}_{M}}\left[-\log \left(p_i^j +\frac{1}{M}\sum_{i^{\prime}\in \mathcal{I}^{j}_{M}}\frac{p_{i^{\prime}}^j}{q_{i^{\prime}}^j}\right)\right]\\
    \displaybreak[1]
   &\quad=\mathbb{E}_{\mathcal{I}^{j}_{M}}\left[\log\frac{M}{\sum_{i^{\prime}\in\mathcal{I}^{j}_M}\Big(\frac{q_{i^{\prime}}^j}{p_{i^{\prime}}^j+p_i^jq_{i^{\prime}}^j}\Big)^{-1}}  \right]\\
   &\quad\underset{(a)}{\le} \mathbb{E}_{\mathcal{I}^{j}_{M}}\left[\log\left({\prod_{i^{\prime}\in\mathcal{I}^{j}_M}\frac{q_{i^{\prime}}^j}{p_{i^{\prime}}^j+p_i^jq_{i^{\prime}}^j}}  \right)^{\frac{1}{M}}\right]\\
   &\quad=\frac{1}{M}\sum_{i'\in \mathcal{I}^{j}_{M}}\mathbb{E}_{i'\sim Q^j}\left[\log \frac{q_{i^{\prime}}^j}{p_{i^{\prime}}^j+p_i^jq_{i^{\prime}}^j}\right] \\
   &\quad\le \frac{1}{M}\sum_{i'\in \mathcal{I}^{j}_{M}}\mathbb{E}_{i'\sim Q^j}\left[\log \left(\frac{q_{i^{\prime}}^j}{p_{i^{\prime}}^j}\right)\right]
   =D_{\text{KL}}(Q^j || P^j),
\end{align*}
\end{small}where (a) uses Harmonic-Geometric inequality. 
\end{proof}

\subsection{Proof of Lemma~\ref{lemm:probability estimation error}}
\begin{proof}
As the expected discrepancy in rectified labels across all layers is bounded by $g(\epsilon)$:
$$
\mathbb{E}_{(u,y)\sim \mathbb{P}} \left[\sum_{j=1}^H \left|  \hat{\bar{z}}^j_{\delta^j(y)} - {\bar{z}}^{j*}_{\delta^j(y)} \right|\right] \le g(\epsilon),$$
the number of node for each layer $N_j\le |\mathcal{Y}|$, and $o^j_k\le B_\mathcal{M}$ (see Appendix~\ref{appendix: proof of theory generalization bound of single layer}), the gap between the expected risks under estimated labels and true rectified labels can be bounded as follows:
\begin{align*}
&\left|\mathbb{E}_{(u,y)\sim \mathbb{P}}[\widetilde{\mathcal{L}}(u,y)] 
- \mathbb{E}_{(u,y)\sim \mathbb{P}}[\widetilde{\mathcal{L}}^*(u,y)] \right| \\
&= \left|\mathbb{E}_{(u,y)\sim \mathbb{P}} \left[\sum_{j=1}^H\left(\hat{\bar{z}}^{j*}_{\delta^j(y)} 
- \hat{\bar{z}}^j_{\delta^j(y)}\right)\log\left(\frac{\exp o^j_{\delta^j(y)}}
{\sum_{k=1}^{N_j}\exp(o^j_k)}\right)\right] \right| \\
&\le \mathbb{E}_{(u,y)\sim \mathbb{P}} \left|\left(2B_{\mathcal{M}}+\log |\mathcal{Y}|\right)
\sum_{j=1}^H\left(\hat{\bar{z}}^{j*}_{\delta^j(y)} - \hat{\bar{z}}^j_{\delta^j(y)}\right)\right| \\
&\le g(\epsilon)\left(2B_{\mathcal{M}}+\log |\mathcal{Y}|\right)
\end{align*}   
\end{proof}

\subsection{Proof of Theorem~\ref{theo: final bound}}
\begin{proof}
    By Lemma \ref{Lemma: bound of sampled softmax}, we have
    \begin{equation}
    \label{eq: cor1}
    \begin{aligned}
        & \mathbb{E}_{(u,y) \sim \mathbb{P}}\left[\widetilde{\mathcal{L}}_j(u,y)\right] - \mathbb{E}_{(u,y), \mathcal{I}_M^{j}}\left[\widehat{\mathcal{L}}_j(u,y)\right] \\
        & \leq \mathbb{E}_{(u,y)\sim \mathbb{P}} \left[ D_{KL}(Q^j||P^j) \right] ,
    \end{aligned}
    \end{equation}
    For convenience, we omit the subscripts of the expectations and use \(\mathbb{E}\) to denote the expectation taken over \(u\), \(y\), and \(\mathcal{I}_M^{j}\). From Lemma \ref{lemma: Rademacher complexity lemma}, we have:

\begin{equation}
\label{eq: generalization bound of j-th layer with sampled softamx}
    \begin{split}
       \mathbb{E}\left[\widehat{\mathcal{L}}_j(u,y)\right] &\le \frac{1}{m} \sum_{i=1}^{m}\widehat{\mathcal{L}}_j(u_i,y_i) + 2 \widehat{\mathcal{R}}_m(\mathcal{ \hat{F}}_\ell^j, S) \\
       &+\Big(4\log N_j+8B_{\mathcal{M}}\Big)\sqrt{\frac{2\log\left(4/\delta\right)}{m}}.
    \end{split}
\end{equation} 
where 
\begin{small}
\begin{equation}
\label{eq: cor2}
\begin{aligned}
& \widehat{\mathcal{R}}_m(\mathcal{ \hat{F}}_\ell^j, S) = \mathbb{E}_{\sigma}\left[\sup_{\hat{f}_{\ell}^j\in \mathcal{\hat{F}}_{\ell}^j}\frac{1}{m}\sum_{i=1}^{m}\sigma_i \hat{f}_{\ell}^j(u_i,y_i)\right]\\
& = \mathbb{E}_{\sigma}\left[\sup_{\substack{f_{\text{DIN}}\in \mathcal{M}\\\bar{z}\in \{0,1\}}}\frac{1}{m}\sum_{i=1}^m \sigma_i\bar{z}^j_{\delta^j(y_i)}\log \frac{ 
\sum_{k\in \mathcal{I}^{\prime}_{M}\cup\{{\delta^j(y_i)}\}} \exp o^j_k(u_i)
}{\exp o_{\delta^j(y_i)}^j(u_i)}\right]. 
\end{aligned}
\end{equation}
\end{small}Using the same analysis method with Eq. \eqref{eq: transformation of Rademacher complexity} and Eq. \eqref{eq: using lemma to transfer the Rademacher complexity} in the proof of Lemma~\ref{Lemma: generalization bound of single layer}, we can get
\begin{equation}
\label{eq: cor3}
    \widehat{\mathcal{R}}_m(\mathcal{ \hat{F}}_\ell^j, S) \le 2(|\mathcal{I}_M^{j}| + 1)\widehat{\mathcal{R}}_m(\mathcal{M}, S) \le 2N_j\widehat{\mathcal{R}}_m(\mathcal{M}, S).
\end{equation}
Combining Eq. \eqref{eq: cor1} $\sim$ Eq. \eqref{eq: cor3} and Lemma \ref{Lemma: Rademacher complexity of DIN}, then aggregating from the 1st layer to the $H$-th layer, and further applying Lemma~\ref{lemm:probability estimation error}, we obtain the desired result.
\end{proof}

\begin{algorithm*}[htbp]
\caption{\textsc{Model Optimization and Tree Index Updating for DTR}}\label{alg:dtr_overview}
\textbf{Input:} Training set $S=\{(u_i,y_i)\}_{i=1}^m$, initial tree $\mathcal{T}^{(0)}$ with mapping $\pi^{(0)}$, initial preference model $\mathcal{M}_{\theta^{(0)}}$; the number of negatives $M$; update stride $d$; probability estimator $\hat{\eta}$; max iterations $R$\\
\textbf{Output:} Learned model $\mathcal{M}$ and tree $\mathcal{T}$\\
\begin{algorithmic}[1]
\FOR{$t=0,1,\dots,R-1$}
  \STATE \emph{// (A) Optimize model $\mathcal{M}_{\theta^{(t)}}$ with fixed tree $\mathcal{T}^{(t)}$ } 
  \WHILE{not converged}
    \STATE Draw a minibatch $\mathrm{MB}\subset S$;
    \FOR{each $(u,y)\in\mathrm{MB}$}
      \FOR{$j=1$ to $H$}
        \STATE $i^+=\delta^j(y)$; ~ \emph{// the index of ancestor node at level $j$}
        \STATE  $w_j=\bar z^{(u,y)}_{j,i^+}\leftarrow \mathbf{I}\!\left[y=\operatornamewithlimits{argmax}\limits_{y'\in \pi(\mathfrak{L}(n^j_{i^+}))}\hat{\eta}_{y'}(u)\right]$; ~\emph{// label rectification}
        \STATE Sample $M$ negatives $I^j_M$ by expanding  nodes using local softmax $q^{j+1}_{c}\propto\exp(o^{j+1}_c)$; ~\emph{// tree-based sampling}
        \STATE  For $i\in I^j_M\cup\{i^+\}$,~ $\hat o^j_i \!\leftarrow\! o^j_i-\log\!\big(\mathbf{I}[i=i^+]+M\,q^j_i\,\mathbf{I}[i\neq i^+]\big)$; ~\emph{// sampled-softmax correction}
        \STATE  $\widehat {\mathcal{L}}_j(u,y) \!=\! -\,w_j\log\frac{\exp(\hat o^j_{i^+})}{\sum_{i\in I^j_M\cup\{i^+\}}\exp(\hat o^j_i)}$; ~\emph{// rectification-weighted sampled softmax loss}
      \ENDFOR
      \STATE $\widehat{\mathcal{L}}(u,y) =\sum_{j=1}^{H}\widehat{\mathcal{L}}_j(u,y)$;
    \ENDFOR
    \STATE $\theta^{(t+1)} \leftarrow \theta^{(t)} - \epsilon\, \nabla_\theta\sum_{(u,y)\in MB}\widehat{\mathcal{L}}_j(u,y)$; ~\emph{// update model parameters}
  \ENDWHILE
  \STATE \emph{// (B) Update tree index $\mathcal{T}^{(t)}$ with fixed $\mathcal{M}_{\theta^{(t+1)}}$}
  \STATE Assign all items to root; $j\leftarrow 0$;
  \WHILE{$j < H$}
    \STATE $j' \leftarrow \min(H,\, j+d)$;
    \FOR{each node $n$ at level $j$}
      \STATE Let $\mathcal{Y}_n$ be the set of items assigned to node $n$;
      \STATE Let $\mathcal{C}(n,j')$ be  the set of descendants of $n$ at level $j'$;
      \FOR{each $y\in \mathcal{Y}_n$}
      \STATE  
        For $c\in\mathcal{C}(n,j')$, $\mathrm{Score}(y,c)\leftarrow\sum_{(u,y)\in A_y}\log\frac{\exp(o^{j'}_{c}(u))}{\sum_{c'\in\mathcal{C}(n,j')}\exp(o^{j'}_{c'}(u))};$~ \emph{// matching score between item $y$ and node $c$}
      \STATE Assign $y$ to $c^*=\operatornamewithlimits{argmax}_{c}\mathrm{Score}(y,c)$;
    \ENDFOR 
    \ENDFOR
    \STATE $j\leftarrow j'$;
  \ENDWHILE
\ENDFOR
\STATE \textbf{return} $\mathcal{M}_{\theta}$ and updated $\mathcal{T}$.
\end{algorithmic}
\label{alg: learning process of DTR}
\end{algorithm*}

\section{Auxiliary Lemmas}
\begin{Lemma}[Theorem 26.5(2) of \cite{shalev2014understanding}]
\label{lemma: Rademacher complexity lemma}
    If the magnitude of loss function $l$ is bounded above by $c$, with probability greater than $1-\delta$ for all $h\in \mathcal{H}$, we have 
    \begin{small}
        \begin{equation*}
    \begin{split}
   \mathbb{E}_{(x,y)\sim\mathcal{D}}[\ell(h(x),y)] &\le \frac{1}{m}\sum_{i=1}^{m}\ell\left(h\left(x_{i}\right),y_{i}\right) \\&+ 2\widehat{\mathcal{R}}_{m}(\ell\circ\mathcal{H},S)+4c\sqrt{\frac{2\ln(4/\delta)}{m}}   
    \end{split}
    \end{equation*}
    \end{small}
where $\ell \circ \mathcal{H}=\{\ell(h(x),y)\mid (x,y)\in \mathcal{X}\times \mathcal{Y}, h\in \mathcal{H}\}$.
\end{Lemma}

\begin{Lemma}[Lemma 2 of \cite{wan2013regularization}]
\label{Lemma: logistic rademacher}
Let $\mathcal{F}$ be class of real functions and $\mathcal{H}=[\mathcal{F}_j]^k_{j=1}$ be a k-dimensional function class. If $\mathcal{A}:\mathbb{R}^k\rightarrow \mathbb{R}$ is a Lipschitz function with constant $L$ and satisfies $\mathcal{A}(0)=0$, then 
$$\widehat{\mathcal{R}}_m(\mathcal{A}\circ \mathcal{H})\le 2kL\widehat{\mathcal{R}}_m(\mathcal{F}).$$ 
\end{Lemma}

\begin{Lemma}[Theorem 1 of \cite{banerjee2005optimality}]
    Let $\psi:\mathbb{R}^N\mapsto \mathbb{R}$ be a strictly convex differentiable function, and $D_\psi:\mathbb{R}^N\times\mathbb{R}^N\mapsto \mathbb{R}$ is the Bregman divergence induced by $\psi$. Let $Y$ be an arbitrary random variable taking values in $\mathbb{R}^N$ for which both $\mathbb{E}[Y]$ and $\mathbb{E}[\psi(Y)]$ are finite, we have
    $$\operatornamewithlimits{argmin}\limits_{\boldsymbol{s}\in \mathbb{R}^N} \mathbb{E}_{Y}[D_\psi(Y,\boldsymbol{s})]=\mathbb{E}[Y].$$
\label{lemm: minimizing of Bregman divergence}
\end{Lemma}

\section{Details of Cross-Entropy Alignment to the Score Proportionality Condition}
\label{appendix: details of condition alignment}
When optimized with standard multi-class softmax loss, according to Eq. \eqref{eq: expectation of layer random vector}, for the original label $z_i^j$, we have  
$$\mathbb{E}[{z}^j_i|u]=\sum_{k \in \left\{\delta(n)\mid n\in \mathcal{C}(n_i^j)\right\}} \mathbb{E}[{z}^{j+1}_k|u].$$ Then, by Lemma \ref{lemm: minimizing of Bregman divergence}, as we are optimizing a loss function that can be expressed as Bregman divergence, the optimization leads to the preference score satisfying $g(\boldsymbol{o}^j)=\mathbb{E}[{\boldsymbol{z}}^j|u]$, where $g$ is the softmax function. Consequently, we obtain: 
\begin{small}
\begin{equation*}
\begin{split}
     \frac{\exp o^j_i}{\sum_{m=1}^{N_j}\exp o^j_m} &= \sum_{k \in \left\{\delta(n)\mid n\in \mathcal{C}(n_i^j)\right\}} \frac{\exp o^{j+1}_k}{\sum_{m=1}^{N_{j+1}}\exp o^{j+1}_m}.
\end{split}
\end{equation*}    
\end{small}Furthermore, we can get:
\begin{small}
$$\exp o_i^j \propto \sum_{k \in \left\{\delta(n)\mid n\in \mathcal{C}(n_i^j)\right\}} \exp o_k^{j+1},$$    
\end{small}where ${\sum_{m=1}^{N_j}\exp o^j_m}\ /\ {\sum_{m=1}^{N_{j+1}}\exp o^{j+1}_m}$ is the proportionality coefficient shared across layer $j$, satisfying the condition in Theorem \ref{theo: sampling theory to unbiased softmax}.

\section{Illustration of Label Rectification}
\label{appendix: illustration of label rectification}
Figure~\ref{fig:label_rectification} illustrates the proposed label rectification method. For the training instance \((u,f)\), item \(f\) is mapped to leaf node \(20\); its ancestors are nodes \(9\), \(4\), and \(1\). Without rectification, all nodes along the path to \(f\) would be treated as positive \((z=1)\). With rectification, we relabel each internal node as follows: set \(\bar{z}=1\) only if the target item \(f\) has the highest conditional probability among the items assigned to the leaves in the subtree rooted at that node; otherwise, set \(\bar{z}=0\). In the figure, the subtree of node \(1\) spans items \(a\sim h\), where item \(c\) has the highest conditional probability (0.13), not item \(f\); hence its rectified label \(\bar{z}=0\). For node \(4\), whose subtree coves item \(e,f,g,h\), the probability of \(f\) is the highest (0.11), so \(\bar{z}=1\). This adjustment better aligns the training objective with the \emph{max-heap} assumption and improves performance under beam search.

\begin{figure*}[htbp]
    \centering
    \includegraphics[width=0.95\linewidth]{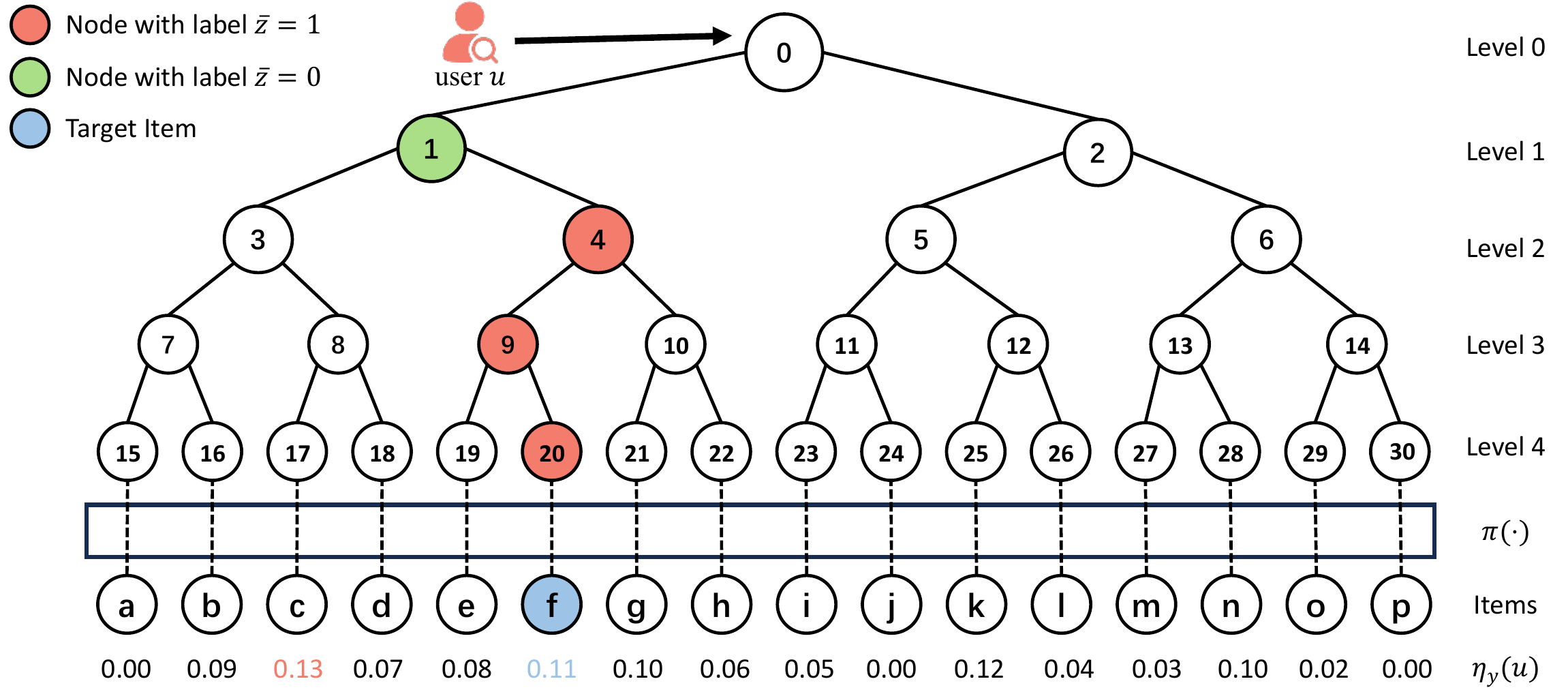}
    \caption{The illustration of label rectification. For the training instance $(u,f)$, item $f$ is mapped to leaf node $20$, whose ancestors are nodes $9$, $4$, and $1$. Among these, nodes $20$, $9$, and $4$ are labeled $\bar{z}=1$, whereas node $1$ is labeled $\bar{z}=0$.}
    \label{fig:label_rectification}
\end{figure*}

\section{Probability Perturbing Experiment}
\label{appendix: probability perturbing experiments}
\subsection{Setting}
We design a probability perturbing experiment on existing datasets to demonstrate that, when the estimation error of probabilities is small, the resulting change in rectified labels remains limited (i.e., a small $\epsilon$ will lead to a small $g(\epsilon)$).

First, we aggregate all user–item interaction counts and normalize them to form a probability matrix $P \in [0,1]^{|\mathcal{U} \times \mathcal{Y}|}$, where $P_{ij}$ denotes the preference probability of user $i$ for item $j$.
We then perturb each nonzero entry of $P$ by adding Gaussian noise $\epsilon \sim \mathcal{N}(0,1)$, clipped to the interval $[-L, L]$. In our experiments, $L$ is set to one of the following values: $[1e\text{-}7,1e\text{-}6, 1e\text{-}5, 1e\text{-}4,1e\text{-}3, 1e\text{-}2]$. After perturbation, we renormalize to obtain the matrix $\widetilde{P}$.
Finally, using both $P$ and $\widetilde{P}$, we compute the unnormalized rectified labels via Eq. (\ref{eq: normalized rectified label}) on tree $\mathcal{T}$ learned by DTR, and measure the percentage of labels that differ due to the perturbation.

\subsection{Results}
Table~\ref{tab:eps-results} reports the ratio of  changed rectified labels under varying perturbation thresholds across four datasets. We observe that, for small thresholds, label changes are negligible, and even at the largest threshold ($L=10^{-2}$), the change ratio remains below 0.10 in all datasets except MIND. These findings confirm that rectified labels are robust to small biases in probability estimation.

\begin{table}[ht]
  \centering
   \caption{Percentage of Rectified Labels Altered under Different Thresholds}
  \resizebox{\linewidth}{!}{
  \begin{tabular}{c|cccccc}
    \toprule
    \diagbox[width=2.3cm]{Dataset}{Threshold}
      & \(1e\text{-}7\) & \(1e\text{-}6\) & \(1e\text{-}5\)
      & \(1e\text{-}4\) & \(1e\text{-}3\) & \(1e\text{-}2\) \\
    \midrule
    Amazon & 0.0326 & 0.0328 & 0.0328 & 0.0330 & 0.0369 & 0.0511 \\
    MIND   & 0.0597 & 0.0599 & 0.0599 & 0.0622 & 0.0847 & 0.1766 \\
    Movie  & 0.0553 & 0.0554 & 0.0554 & 0.0564 & 0.0623 & 0.0784 \\
    Tmall  & 0.0602 & 0.0605 & 0.0606 & 0.0607 & 0.0675 & 0.0925 \\
    \bottomrule
  \end{tabular}
   }
  \label{tab:eps-results}
\end{table}


\section{Comparison with Sequential and Graph-based Methods}
We compare our proposed DTR (T-RL) with sequential recommendation models—GRU4Rec~\cite{hidasi2015session}, SASRec~\cite{kang2018self}, and FMLP-Rec~\cite{zhou2022filter}—as well as a graph-based method, LightGCN~\cite{he2017neural}.
GRU4Rec consists of an embedding layer followed by 6 gated recurrent layers, each with a hidden size of 64.
SASRec consists of an embedding layer, 4 self-attention layers, and an MLP with layer dimensions [24,128,24], while FMLP-Rec adopts a similar architecture, replacing the self-attention layers with 4 filter layers.
For LightGCN, we build the user-item interaction graph based on users' historical sequences and set the number of graph convolutional layers to 3. The results are shown in Table~\ref{table: additional_results}. Except for the marginal gap at top-$K=20$ on the Tmall dataset, our proposed DTR (T-RL) consistently outperforms all baselines. The gap on Tmall may be due to data sparsity, where models with filter layers like FMLP-Rec perform better.
\begin{table}[ht]
\centering
\caption{Performance comparison of DTR(T-RL) against sequential and graph-based  methods on four datasets under top-$K$ = 20 and 40. \\ \textbf{P}, \textbf{R}, and \textbf{F} denote Precision, Recall, and F-measure, respectively.}
\label{table: additional_results}
\begin{tabular}{c|cccccc}
\toprule
\textbf{Algorithm} & \textbf{P@20} & \textbf{R@20} & \textbf{F@20}
  & \textbf{P@40} & \textbf{R@40} & \textbf{F@40}  \\ \midrule
  & \multicolumn{6}{c}{Movie}  \\
\midrule
\multicolumn{7}{c}{} \\[-1.5em]
\midrule
GRU4Rec       & 0.2072	&0.1121	&0.1281	&0.1883	&0.1944	&0.1643 \\
SASRec  & 0.2217	&0.1220	&0.1392	&0.1965	&0.2044	&0.1726	  \\
FMLP-Rec & 0.2310	&0.1273	&0.1450	&0.2022	&0.2110	&0.1778  \\
LightGCN  & 0.1272	& 0.0577	&0.0688	&0.1188	&0.1042 &0.0931\\
DTR(T-RL)   & \textbf{0.2545} & \textbf{0.1374} & \textbf{0.1580} & \textbf{0.2176} & \textbf{0.2240} & \textbf{0.1905}   \\
\midrule
{}  & \multicolumn{6}{c}{MIND}  \\
\midrule
\multicolumn{7}{c}{} \\[-1.5em]
\midrule
GRU4Rec        &0.3930	&0.1820	&0.2203	&0.3184	&0.2730	&0.2547  \\
SASRec  & 0.4041	&0.1860	&0.2254	&0.3318	&0.2818	&0.2639  \\
FMLP-Rec & 0.4112	&0.1902	&0.2303	&0.3363	&0.2866	&0.2684  \\
LightGCN  & 0.1807	&0.0620	&0.0806	&0.1665	&0.1138	&0.1147 \\
DTR(T-RL)  & \textbf{0.4210} & \textbf{0.1938} & \textbf{0.2350} & \textbf{0.3401} & \textbf{0.2869} & \textbf{0.2695} \\
\midrule
{}  & \multicolumn{6}{c}{Amazon}  \\
\midrule
\multicolumn{7}{c}{} \\[-1.5em]
\midrule
GRU4Rec        & 0.0624	&0.0421	&0.0455	&0.0527	&0.0698	&0.0536	  \\
SASRec  &0.0683	&0.0475	&0.0509	&0.0578	&0.0784	&0.0595 \\
FMLP-Rec  & 0.0711	&0.0500	&0.0533	&0.0601	&0.0825	&0.0622  \\
LightGCN  & 0.0228	&0.0132	&0.0149	&0.0198	&0.0227	&0.0185  \\
DTR(T-RL) &\textbf{0.0777} & \textbf{0.0542} & \textbf{0.0580} & \textbf{0.0626} & \textbf{0.0847} & \textbf{0.0644}  \\
\midrule
{}  & \multicolumn{6}{c}{Tmall}  \\
\midrule
\multicolumn{7}{c}{} \\[-1.5em]
\midrule
GRU4Rec      &0.0379&0.0263	&0.0286	&0.0292	&0.0399	&0.0308  \\
SASRec  & 0.0462	&0.0324	&0.0352	&0.0350	&0.0482	&0.0370 \\
FMLP-Rec & \textbf{0.0485}	&\textbf{0.0345} &\textbf{0.0372}	&0.0366	&0.0509	&0.0390  \\
LightGCN  & 0.0069	&0.0047	&0.0051	&0.0053	&0.0071 &0.0055 \\
DTR(T-RL)  & 0.0482 & 0.0342 & 0.0369 & \textbf{0.0373} & \textbf{0.0514} & \textbf{0.0395}  \\
\bottomrule
\end{tabular}
\end{table}

\section{Adaptive Sampling Strategy}
The sampling layer size is fixed at each tree level in our current design. However, an adaptive sampling strategy based on node importance or uncertainty could potentially be more effective. In our tree-based sampling method, each training iteration samples a single top-down path from the root to a leaf, and the nodes along this path are used as negative samples. Due to this path-based sampling mechanism, the number of samples per layer is implicitly determined by the tree structure and cannot be independently adjusted. As a result, it is not feasible to apply adaptive sampling strategies to DTR(T) and DTR(T-RL). Therefore, we conduct adaptive sampling experiments on DTR(U).
Motivated by the intuition that deeper layers should be sampled more heavily—as they are more tightly connected to items and contain more nodes—we explore a depth-aware adaptive sampling strategy. Specifically, we begin by setting the per-layer sample size to 70 and compute the total sampling budget $N$. For each level $l$, we assign a sampling weight $w_l = 1 + \alpha \cdot l$, where $\alpha$ is a hyperparameter that controls the emphasis on deeper levels. The number of samples allocated to level $l$ is then computed as:
$$
N_l = \frac{w_l}{\sum_{j=1}^H w_j} \cdot N
$$
As $\alpha$ increases, more samples are assigned to deeper layers. We vary $\alpha$ from 0 to 1.0 in steps of 0.2. The corresponding results, shown in Figure~\ref{fig:adaptive_num_fig}, demonstrate that moderately increasing the number of samples at deeper levels can lead to improved model performance, validating the effectiveness of adaptive sampling strategy.

\begin{figure}[htbp]
  \centering 
  \setcounter{subfigure}{0} 
  \subfloat[Movie]{\includegraphics[width=0.243\textwidth]{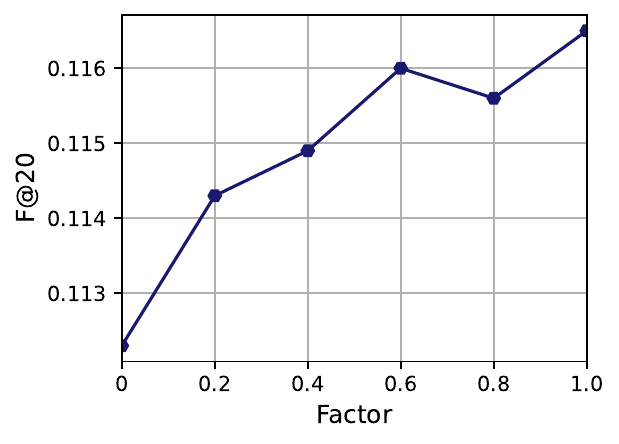}}
  \subfloat[MIND]{\includegraphics[width=0.243\textwidth]{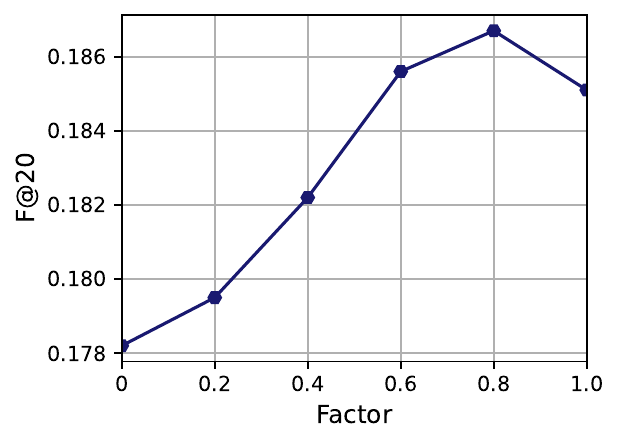}}\\
  \vspace{-0.25cm}
  \subfloat[Amazon]{\includegraphics[width=0.243\textwidth]{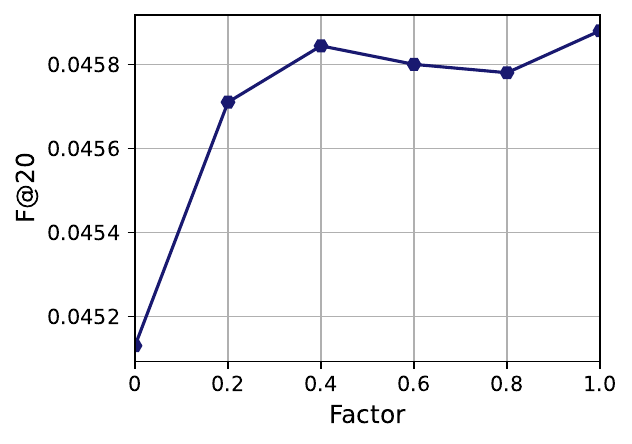}}
  \subfloat[Tmall]{\includegraphics[width=0.243\textwidth]{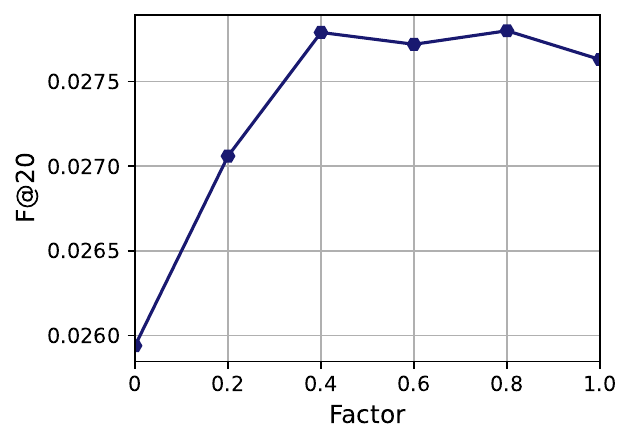}}\\
  \vspace{-0.1cm}
  \caption{Variation of $F\text{-}measure@20$ for DTR(U) with adaptive Sampling across four datasets.}
  \label{fig:adaptive_num_fig}
\end{figure}

\section{Detailed Time and Memory Costs of Algorithms}

\subsection{Training Time}
Table~\ref{tab:training_time_detail} reports the training time per 1,000 steps across four datasets. ScaNN, IPNSW, QINCo, and ROTLEX are omitted because they are two-stage indexing methods. For our proposed method, DTR with uniform sampling (DTR(U)) incurs training time comparable to that of training the preference model alone. In contrast, tree-based sampling (DTR(T)) introduces additional overhead due to the need to compute node scores during training. Further, when using the modified loss with rectified labels (DTR(T-RL)), training time increases further because of the probability estimation. Compared to OTM, which also performs inference-based sampling during training, all DTR variants demonstrate lower training times, underscoring the efficiency of our method. Although the additional computations introduce moderate overhead, they represent a reasonable trade-off for improved performance. In practice, the impact can be further reduced: since the auxiliary model is pre-trained and independent of the main training loop, the conditional probabilities can be pre-computed and cached. Moreover, distributed training can be leveraged to further alleviate the cost in real-world deployments.

\begin{table}[htbp] 
\caption{Training Time (in seconds) of Algorithms}
\label{tab:training_time_detail}
\centering
\begin{tabular}{c|cccc}
\hline
Algorithm & Movie & MIND & Amazon & Tmall \\ \hline
DIN & 59.95 & 58.29 & 50.24 & 51.44 \\
YouTubeDNN & 32.65 & 19.82 & 69.20 & 130.80 \\
FMLP-Rec & 23.10 & 21.88 & 22.13 & 23.13 \\
JTR & 121.9 & 85.38 & 113.14 & 134.38 \\
PLT & 17.05 & 15.78 & 17.22 & 23.22 \\
TDM & 36.91 & 29.35 & 33.55 & 59.07 \\
JTM & 35.11 & 29.90 & 33.76 & 57.36 \\
OTM & 261.45 & 259.00 & 278.74 & 305.57 \\ \hline
DTR(U) & 60.21 & 57.39 & 69.15 & 66.47 \\
DTR(T) & 126.97 & 118.83 & 177.83 & 193.70 \\
DTR(T-RL) & 158.25 & 147.09 & 230.16 & 254.72 \\ \hline
\end{tabular}
\end{table}

\subsection{Inference Time}
The detailed inference times are presented in Table~\ref{tab:inference_time_detail}. It is worth noting that PLT, TDM, JTM, OTM, and DTR adopt similar index structures and preference models, resulting in comparable inference time overheads.
\begin{table}[htbp]
\caption{Inference Time (in seconds) of Algorithms}
\label{tab:inference_time_detail}
\centering
\begin{tabular}{c|cccc}
\hline
Algorithm & Movie & MIND & Amazon & Tmall \\ \hline
DIN & 81.32 & 38.93 & 217.21 & 2104.05 \\
YouTubeDNN & 48.90 & 25.27 & 121.60 & 1022.94 \\
SCANN & 20.39 & 18.64 & 4.13 & 15.49 \\
IPNSW & 16.84 & 10.98 & 11.15 & 53.45 \\
JTR & 11.95 & 9.24 & 15.87 & 82.79 \\
QINCo & 9.41 & 3.71 & 7.98 & 19.40 \\
ROTLEX & 5.53 & 3.56 & 3.64 & 17.89 \\
PLT & 5.65 & 3.40 & 3.48 & 18.47 \\
TDM & 5.62 & 3.36 & 3.46 & 18.50 \\
JTM & 5.71 & 3.38 & 3.47 & 18.49 \\
OTM & 5.57 & 3.41 & 3.48 & 18.45 \\
DTR & 5.58 & 3.40 & 3.53 & 18.79 \\ \hline
\end{tabular}
\end{table}

\subsection{Memory Consumption}
The memory consumption of the recommendation model is inevitable in any recommender systems, so we mainly focus on the size of the index structure.  The memory overhead of different index structures is reported in Table~\ref{tab:memory_consumption_detail}. Notably, PLT, TDM, JTM, OTM, and DTR all adopt the same indexing structure, resulting in identical memory usage; thus, we only report the index memory consumption for OTM and DTR in the table. Importantly, DTR achieves strong performance while maintaining excellent memory efficiency—its memory usage is on par with IPNSW and OTM, significantly lower than that of JTR. Compared with ROTLEX, DTR has lower memory consumption on Movie and MIND, whereas ROTLEX is more memory-efficient on Amazon and Tmall. This is because ROTLEX employs a bucket strategy that maps multiple embeddings to a single bucket, thereby reducing the number of internal nodes for large vector spaces. However, this breaks the one-to-one correspondence between leaf nodes and item embeddings, leading to the performance degradation shown in Table~\ref{table: results}. When compared to quantization-based methods like SCANN and QINCo, which are known for minimal memory usage, DTR’s memory consumption is higher; yet their performance lags significantly behind (as shown in Table~\ref{table: results}), underscoring that DTR achieves a better balance between memory cost and retrieval effectiveness.

\begin{table}[htbp]
\caption{Memory Consumption (MB) of Index}
\label{tab:memory_consumption_detail}
\centering
\begin{tabular}{c|cccc}
\hline
Algorithm & Movie & MIND & Amazon & Tmall \\ \hline
SCANN & 1.93 & 1.48 & 9.26 & 18.06 \\
IPNSW & 6.38 & 4.01 & 43.02 & 80.00 \\
JTR & 11.32 & 7.44 & 71.21 & 142.22 \\
QINCo & 2.31 & 2.07 & 6.21 & 10.87 \\
ROTLEX & 8.54 & 4.07 & 11.26 & 24.66 \\
OTM & 6.00 & 3.48 & 45.34 & 90.78 \\
DTR & 6.00 & 3.48 & 45.34 & 90.78 \\ \hline
\end{tabular}
\end{table}

\end{appendices}

 





\end{document}